\documentclass[12pt]{article}
\usepackage{booktabs}
\usepackage{amsmath}
\usepackage{algorithm}
\usepackage{amsfonts}
\usepackage{amsthm}
\usepackage{mathtools}
\usepackage{float}
\usepackage{graphicx}
\usepackage{enumerate}
\usepackage{natbib}
\usepackage{url} 
\usepackage{subcaption}

\newcommand{\blind}{1}

\usepackage{bm}
\DeclareBoldMathCommand{\ba}{a}
\DeclareBoldMathCommand{\bb}{b}
\DeclareBoldMathCommand{\be}{e}
\DeclareBoldMathCommand{\br}{r}
\DeclareBoldMathCommand{\bg}{g}
\DeclareBoldMathCommand{\bt}{t}
\DeclareBoldMathCommand{\bu}{u}
\DeclareBoldMathCommand{\bv}{v}
\DeclareBoldMathCommand{\be}{e}
\DeclareBoldMathCommand{\bw}{w}
\DeclareBoldMathCommand{\bx}{x}
\DeclareBoldMathCommand{\bz}{z}
\DeclareBoldMathCommand{\bm}{m}
\DeclareBoldMathCommand{\by}{y}
\DeclareBoldMathCommand{\bh}{h}

\DeclareBoldMathCommand{\bA}{A}
\DeclareBoldMathCommand{\bD}{D}
\DeclareBoldMathCommand{\bB}{B}

\DeclareBoldMathCommand{\bY}{Y}
\DeclareBoldMathCommand{\bX}{X}
\DeclareBoldMathCommand{\bZ}{Z}
\DeclareBoldMathCommand{\bM}{M}
\DeclareBoldMathCommand{\bN}{N}
\DeclareBoldMathCommand{\bP}{P}
\DeclareBoldMathCommand{\bI}{I}
\DeclareBoldMathCommand{\bT}{T}
\DeclareBoldMathCommand{\bU}{U}
\DeclareBoldMathCommand{\bS}{S}
\DeclareBoldMathCommand{\bQ}{Q}
\DeclareBoldMathCommand{\bW}{W}
\DeclareBoldMathCommand{\bE}{E}
\DeclareBoldMathCommand{\bV}{V}
\DeclareBoldMathCommand{\bJ}{J}
\DeclareBoldMathCommand{\bL}{L}
\DeclareBoldMathCommand{\bR}{R}
\DeclareBoldMathCommand{\bG}{G}

\DeclareBoldMathCommand{\bzero}{0}
\DeclareBoldMathCommand{\bone}{1}
\DeclareBoldMathCommand{\balpha}{\alpha}
\DeclareBoldMathCommand{\bxi}{\xi}
\DeclareBoldMathCommand{\bbeta}{\beta}
\DeclareBoldMathCommand{\bkappa}{\kappa}
\DeclareBoldMathCommand{\brho}{\rho}
\DeclareBoldMathCommand{\btau}{\tau}
\DeclareBoldMathCommand{\beeta}{\eta}
\DeclareBoldMathCommand{\btheta}{\theta}
\DeclareBoldMathCommand{\bdelta}{\delta}
\DeclareBoldMathCommand{\bPhi}{\Phi}
\DeclareBoldMathCommand{\bzeta}{\kappa}
\DeclareBoldMathCommand{\bgamma}{\gamma}
\DeclareBoldMathCommand{\bSigma}{\Sigma}
\DeclareBoldMathCommand{\bmathcaly}{\mathcal{Y}}
\DeclareBoldMathCommand{\bTheta}{\Theta}
\DeclareBoldMathCommand{\bmu}{\mu}
\DeclareBoldMathCommand{\bpsi}{\psi}
\DeclareBoldMathCommand{\bzeta}{\zeta}
\DeclareBoldMathCommand{\bvarepsilon}{\varepsilon}
\DeclareBoldMathCommand{\bepsilon}{\epsilon}

\DeclareBoldMathCommand{\rn}{\frac{\phi}{\phi+\|\tilde{\bZ}_n\|_2^2}}

\DeclareMathOperator{\Tr}{Tr}
\def\bTheta{\boldsymbol{\Theta}}

\def\bSigma{\boldsymbol{\Sigma}}
\def\bOmega{\boldsymbol{\Omega}}

\newcommand{\pw}{\bP_\bW}
\newcommand{\px}{\bP_\bX}

\newcommand{\pwn}{\bP_{\bW_n}}
\newcommand{\pxn}{\bP_{\bX_n}}

\newtheorem{theorem}{Theorem}[section]
\newtheorem{proposition}[theorem]{Proposition}
\newtheorem{lemma}[theorem]{Lemma}
\newtheorem{corollary}[theorem]{Corollary}
\newtheorem{definition}[theorem]{Definition}

\begin{document}

\def\spacingset#1{\renewcommand{\baselinestretch}%
{#1}\small\normalsize} \spacingset{1}


\if1\blind
{
  \title{\bf Anytime-Valid Inference in Linear Models and Regression-Adjusted Causal Inference}
  \author{Michael Lindon$^*$,   Dae Woong Ham$^\dagger$, Martin Tingley$^\star$,\\ and Iavor Bojinov$^\diamond$}
  \maketitle
} \fi

\if0\blind
{
  \bigskip
  \bigskip
  \bigskip
  \begin{center}
    {\LARGE\bf Anytime-Valid Inference in Linear Models and Regression-Adjusted Causal Inference}
\end{center}
  \medskip
} \fi

\bigskip
\begin{abstract}

Linear models are foundational tools in statistics and ubiquitous across the applied sciences. However, conventional statistical inference —such as $t$-tests and $F$-tests— 
are only valid at fixed sample sizes,
making them unsuitable for sequential settings such as online A/B testing. 
We develop an anytime-valid theory of inference for the linear model, introducing
sequential analogues of classical tests and confidence sets that provide Type-I error control and coverage guarantees uniformly over all sample sizes.
Our construction is based on likelihood ratios of invariantly sufficient statistics, yielding simple closed-form expressions of ordinary least squares estimators and standard errors.
The resulting tests are optimal in the GROW/REGROW sense for both frequentist and Bayesian alternative hypotheses.
We then relax the linear model assumptions to provide heteroskedasticity-robust asymptotic sequential tests and confidence sequences, which enable sequential regression-adjusted inference for causal estimands in randomized controlled experiments.
This formally allows experiments to be continuously monitored for significance, stopped early, and safeguards against statistical malpractices in data collection. 
We demonstrate the practical utility of our approach through simulations and applications to real A/B test data from Netflix.

\end{abstract}

\noindent%
\vfill
\footnotesize
\qquad
\\
$^*$Netflix, Los Gatos, CA. michael.s.lindon@gmail.com
\\
$^\dagger$Ross School of Business, Ann Arbor, MI. daewoong@umich.edu
\\
$^\star$Microsoft, Redmond, WA. martin.tingley@gmail.com
\\
$^\diamond$Harvard School of Business, Harvard, MA. ibojinov@hbs.edu
\normalsize
\newpage
\spacingset{1.9} 
\section{Introduction}
\label{sec:intro}

We develop an anytime-valid approach to inference in linear models using $e$-variables.
Linear regression models are fundamental to statistics with many practical applications across the applied sciences.
In causal inference they are often employed to estimate and test hypotheses about treatment effects in randomized experiments \citep{lin}.
Their ability to achieve variance reduction through regression adjustment and their robustness to model misspecification under randomized assignment result in their widespread recommendation in many official research guidelines \citep{ye}.
Linear models can also be computed recursively, updating parameters as new data arrives, making them a scalable option for the analysis of streaming data common in online A/B tests.

However, current testing and interval estimation procedures are limited in that their Type-I error and coverage guarantees only hold at a fixed, often pre-determined, sample size (which we refer to as ``fixed-$n$'').
This is a severe limitation in online A/B tests, where data arrives sequentially, as it is desirable to continuously monitor experiments for nonzero treatment effects.
It also causes problems in connection with the replication crisis of published research \citep{asapvalue, redefinestatisticalsignificance, sagarin2014}.
For example, it is estimated that only one-third of statistically significant results published in top psychology journals can be reproduced in follow-up experiments \citep{oss}.
A recognized issue is that researchers often do not disclose details of the data collection process at the time of submission \citep{ioannidis2005}.
In an anonymous survey of 2000 published researchers in psychology, approximately 60\% of respondents admitted to collecting more data after seeing whether the results were significant, with approximately 20\% admitting to stopping data collection after achieving a significant result \citep{john2012}.
The high numbers of published false positives can at least partially be attributed to the inability of current methods to control Type-I errors under optional stopping.

The paper is structured as follows.
In Section~\ref{sec:safe_testing} we review the literature and key concepts of safe anytime-valid inference.
In Section~\ref{sec:group_invariance} we introduce the prerequisites of group-invariance and its applications to statistics.
In Section~\ref{sec:linear_model_results} we develop anytime-valid inference for the linear model through group-invariance principles, presenting sequential analogues of classical $t-$tests, $F-$tests, and confidence sets.
In Section \ref{sec:heteroskedastic} we relax the assumption of homoskedasticity and provide robust standard errors, while also introducing asymptotic ``t'' confidence sequences.
In Section~\ref{sec:ate} we consider randomized experiments, in which we are able to completely relax the assumptions of the linear model, providing asymptotic confidence sequences and $e$-processes for causal estimands.
In Section~\ref{sec:examples} we provide both simulated and real examples, which demonstrate robustness to model misspecification and the practical utility of our approach.

\section{Safe Testing and Anytime-Valid Inference}
\label{sec:safe_testing}
In this section we introduce key concepts and definitions for safe anytime-valid inference (see \citet{e-book} for a thorough introduction)
Consider a set of probability measures $\mathcal{P}$ with a subset $\mathcal{H}_0 \subset \mathcal{P}$ defining a null hypothesis.
An $e$-\textit{variable} \citep{grunwald} relative to a collection of distributions $\mathcal{H}_0$ is a nonnegative random variable satisfying
$\mathbb{E}_P[E] \leq 1$ for all $P \in \mathcal{H}_0$.
$e$-variables are used to reject the null hypothesis when $E\geq\alpha^{-1}$, obtaining a Type-I error guarantee $\mathbb{P}_P[E \geq \alpha^{-1}] \leq \mathbb{E}_P[E] / \alpha^{-1} = \alpha$ for all $P \in \mathcal{H}_0$ by Markov's inequality.
A \textit{conditional} $e$-\textit{variable} $C_n$ on a sample space equipped with a filtration $\mathcal{F}_n$, is a nonnegative $\mathcal{F}_n$-measurable random variable satisfying 
$\mathbb{E}_P[C_n | \mathcal{F}_{n-1}] \leq 1$ for all $ P \in \mathcal{H}_0$.
A sequence $(E_n)_{n=1}^\infty$ adapted to $\mathcal{F}_n$ is a \textit{test martingale} \citep{shafer} relative to $H_0$ if $E_n$ is a nonnegative supermartingale for all $P \in \mathcal{H}_0$.
Clearly, $E_n \coloneqq \prod_{i=1}^n C_i$ defines a test-martingale because $\mathbb{E}_P[E_n | \mathcal{F}_{n-1}] \leq E_{n-1}$ for all $P \in \mathcal{H}_0$.
Test martingales are used to reject the null hypothesis when $E_n \geq \alpha^{-1}$, obtaining a Type-I error guarantee $\mathbb{P}_P[ \exists n \in \mathbb{N} : E_n \geq \alpha^{-1}] \leq \alpha$ for all $P \in \mathcal{H}_0$ by Ville's inequality \citep{ville}.
Alternatively we may say that the stopped process $E_\tau$, for any stopping time $\tau$ with respect to $\mathcal{F}_n$, is an $e$-variable relative to $\mathcal{H}_0$ \citep{howard}.

A test martingale is actually a stronger condition than necessary to obtain the former time-uniform Type-I error guarantee. 
Instead, it suffices to show that for all $P \in \mathcal{H}_0$ there exists a $P$-nonnegative supermartingale $(M_n^P)_{n=1}^\infty$ such that $E_n\leq M_n^P$ $\forall n \in \mathbb{N}$ $P$-almost surely.
Such a process is called an $e$-\textit{process} \citep{grunwaldramdas}.
A $p$-\textit{process} $(p_n)_{n=1}^\infty$ is a sequence of nonnegative random variables satisfying $\mathbb{P}_P[\exists n \in \mathbb{N} : p_n\leq \alpha ] \leq \alpha$ for all $P \in \mathcal{H}_0$, and is readily obtained using $p_n = 1/E_n$.
Similarly, the stopped process $p_\tau$ is a $p$-\textit{variable} (more commonly $p$-\textit{value}).
A \textit{confidence sequence} $(C_n^\alpha)_{n=1}^\infty$ for a sequence of estimands $(\theta_n)_{n=1}^\infty$ is a sequence of sets satisfying $\mathbb{P}_P[\theta_n \in C_n^\alpha\,\, \forall n\in\mathbb{N}] \geq 1-\alpha$ for all $P \in \mathcal{H}_0$, providing a time-uniform $1-\alpha$ \textit{coverage} guarantee.
Consider the $\alpha$-level sequential test which rejects the null hypothesis using the stopping time $\tau = \inf\{n \in \mathbb{N} : E_n \geq \alpha^{-1}\}$.
Such a test is said to be \textit{power 1} \citep{power1} if $\mathbb{P}_P[\tau < \infty] = 1$ for all $P \in \mathcal{H}_1 \coloneqq \mathcal{P}\setminus \mathcal{H}_0$.
A sufficient condition for power 1 is to show $\underset{n \rightarrow \infty}{\lim} (1/n)\log E_n = c > 0$ $P$-almost surely for all $P \in \mathcal{H}_1$.

In the second half of this paper we leverage asymptotic arguments. A process $(\bar{E}_n)_{n=1}^\infty$ is an \textit{asymptotic $e$-process} \citep{asymptoticcs} for $\mathcal{H}_0$ if there exists a nonasymptotic $e$-process $(E_n)_{n=1}^\infty$ for $\mathcal{H}_0$ such that $\log(\bar{E}_n) / \log(E_n) \rightarrow 1$ $P$-almost surely for all $P$.
Similarly a sequence of sets $(\bar{C}_n^\alpha)_{n=1}^\infty$ is an \textit{asymptotic confidence sequence} for $(\theta_n)_{n=1}^\infty$ if there exists a nonasymptotic confidence sequence $(C_n^\alpha)_{n=1}^\infty$ such that the normalized measure of the symmetric difference vanishes almost surely i.e.
$\mu(C_n^\alpha \Delta \bar{C}_n^\alpha) / \mu(C_n^\alpha) \rightarrow 0$ almost surely for an appropriate measure $\mu$.

\section{Composite Nulls and Group-Invariance}
\label{sec:group_invariance}

Constructing a valid $e$-process for a composite null (i.e. with nuisance parameters so that $\mathcal{H}_0$ is not a singleton) can be challenging.
Instead, we seek to construct a sequence of statistics $(m_i(\bY_i))_{i=1}^\infty$ whose distribution does not depend on the nuisance parameters, reducing a composite null to a simple null.
Additionally, we further desire that the reduction of the data to a statistic retains as much information about the parameters of interest as possible, being optimal in some sense to be defined.
To this end we appeal to group-invariance principles.

Let $\mathcal{P} = \{P_\theta : \theta \in \Theta\}$ be a parametric family of probability measures defined on a measurable space $(\mathcal{Y}, \mathcal{B})$.
Let $G$ be a group of transformations such that every $g\in G$ is a a one-to-one measurable function from $\mathcal{Y}$ to $\mathcal{Y}$. 
For any $P \in \mathcal{P}$ and $g \in G$ the probability measure $gP$ is defined by $gP(B) = P(g^{-1}B)$ for any $B \in \mathcal{B}$.
A statistical model is $G$-\textit{invariant} if for each $P \in\mathcal{P}$, $gP\in\mathcal{P}$ for all $g \in G$.
For parametric models this implies that if $y \sim P_\theta$ for some $\theta \in \Theta$, then for any $g \in G$, we have $gy\sim P_{\theta'}$ for some $\theta' \in \Theta$ i.e. the law of $gy$ belongs to the same parametric family.
This in turn induces a group of transformations on the parameter space defined by $\theta' = g\theta$.

Now consider a sequence $(y_i)_{i=1}^\infty$ of i.i.d. outcomes from $P_{\theta}$ for some $\theta \in \Theta$ and let $\bY_n = (y_1, \dots, y_n) \in \mathbb{R}^n$ denote first $n$ outcomes.
A statistic $m_n(\bY_n)$ is said to be \textit{invariant} if $m_n(g \bY_n) = m_n(\bY_n)$ for all $g \in G$.
The \textit{orbit} of $\bY_n$ is the set $\mathcal{O}(\bY_n) = \{g\bY_n : g \in G\}$.
A \textit{maximal invariant} statistic is an invariant statistic that is constant on orbits and assigns a unique value to each orbit i.e., $m(\bY_n')=m(\bY_n)$ $\Rightarrow$ $\bY_n' = g\bY_n$ for some $g \in G$.
A standard result is that any invariant statistic can be written as a function of a maximal invariant statistic.
An invariant statistic can be written $m_n(\bY_n) \sim P_{t(\theta)}$ where $t(\theta)$ is an invariant function of the parameters i.e. $t(g \theta) = t(\theta)$.
Without loss of generality it is helpful to simply redefine the parameter space as $\theta = (\rho, \xi)$ where $\rho \in \mathrm{P}$ is a vector of nuisance parameters
and $\xi \in \Xi$ is a $G$-invariant function of the parameters of interest. For example, if $y_i \sim N(\mu, \sigma^2)$, we may write $\theta = (\sigma, \mu / \sigma)$. 
The group action $g: \bY_n \rightarrow c \bY_n$ sends $(\sigma, \mu / \sigma) \rightarrow (c\sigma, \mu/\sigma)$, as the standardized mean $\xi \coloneqq \mu / \sigma$ is invariant under $G$.
The statistic $m_n(\bY_n) = (y_1 / |y_1|, \dots, y_n / |y_1|)$ is a maximal invariant under $G$, whose distribution depends only on the parameter of interest $\xi$ and not on the nuisance parameter $\sigma$.

Consider testing $H_0: \xi = \xi_0$ against either a frequentist alternative $H_1: \xi = \xi_1$ or a Bayesian alternative $H_1: \xi \sim \Pi$.
In either case the only free parameters are the nuisance parameters $\rho$, so we can write $H_0: \bY_n \sim P_\rho$ vs $H_1: \bY_n \sim Q_\rho$ for some $\rho \in \mathrm{P}$ ($Q_\rho$ is the Bayes marginal for the Bayesian alternative)\footnote{Note we refrain from introducing additional notation to distinguish between the probability measures for single outcomes $P$, vectors $P^{(n)}$, and sequences $P^{(\infty)}$, as it is clear from the context.}. 
If $m_n(\bY_n)$ is a maximal invariant then we can simply write $m_n(\bY_n) \sim P$ and $m_n(\bY_n) \sim Q$ (dropping the index $\rho$) as the distribution of $m_n(\bY_n)$ is unaffected by $\rho$.
\begin{lemma}
    \label{lem:invariant_lr}
Let $(m_i(\bY_i))_{i=1}^\infty$ be a sequence of invariant statistics. The likelihood ratio
\begin{equation}
    \label{eq:simple_lr}
    E_n \coloneqq \frac{q(m_1(\bY_1), \dots, m_n(\bY_n))}{p(m_1(\bY_1), \dots, m_n(\bY_n))}
\end{equation}
is a test martingale for $H_0$ with respect to the filtration $\mathcal{F}_n = \sigma(m_i(\bY_i) : i\leq n)$.
If $m_i(\bY_i)$ are maximally invariant, then equation \eqref{eq:simple_lr} simplifies to
\begin{equation}
    \label{eq:simple_mi_lr}
    E_n = \frac{q(m_n(\bY_n))}{p(m_n(\bY_n ))}.
\end{equation}
\end{lemma}
The first result follows because the factor $\frac{q(m_n(\bY_n) | m_1(\bY_1),\dots,m_{n-1}(\bY_{n-1}))}{p(m_n(\bY_n) | m_1(\bY_1),\dots,m_{n-1}(\bY_{n-1}))}$ is a conditional $e$-variable.
The second result follows because $m_i(\bY_i)$ for $i<n$ are invariant functions of $\bY_n$ and can therefore be written as functions of $m_n(\bY_n)$.
Hence $E_n$ is both by itself an $e$-variable and the sequence $(E_n)_{n \in \mathbb{N}}$ is a test martingale (and therefore also an $e$-process) for $\mathcal{H}_0$.
The advantage of working with an invariant statistic, therefore, is that it reduces a composite null to a simple null, permitting the use of $E_n$ in equation \eqref{eq:simple_lr} to test $H_0$.
The advantage of working with a \textit{maximal} invariant statistic, is that it simplifies computation of $E_n$ to equation \eqref{eq:simple_mi_lr}.
The primary advantage of working with a maximal invariant statistic, however, is that the resulting $e$-process is GROW and REGROW optimal \citep{muriel}.
\begin{definition}[GROW/REGROW $e$-Variable \citep{grunwald}]
    \label{def:grow}
Let $\mathcal{E}_0 = \{T_n : \mathbb{E}_{P_\rho}[T_n] \leq 1 \,\forall \rho \in \mathrm{P}\}$ denote the set of $e$-variables for $\{P_\rho\}_{\rho \in \mathrm{P}}$. The growth rate optimal in the worst case (GROW) $e$-variable $E_n$ for $\{P_\rho\}_{\rho \in \mathrm{P}}$ against $\{Q_\rho\}_{\rho \in \mathrm{P}}$ satisfies
\begin{equation}
    \label{eq:grow}
    \inf_{\rho}\mathbb{E}_{Q_\rho}[\log E_n] = \sup_{T_n \in \mathcal{E}_0} \inf_{\rho}\mathbb{E}_{Q_\rho}[\log T_n],
\end{equation}
and the relatively GROW $e$-variable $E_n$ satisfies
\begin{equation}
    \label{eq:regrow}
    \inf_{\rho}\{\mathbb{E}_{Q_\rho}[\log E_n] - \sup_{T_n' \in \mathcal{E}_0}[ \mathbb{E}_{Q_\rho} \log T_n']\} = \sup_{T_n \in \mathcal{E}_0}\inf_{\rho}\{\mathbb{E}_{Q_\rho}[\log T_n] - \sup_{T_n' \in \mathcal{E}_0} \mathbb{E}_{Q_\rho} [\log T_n']\}
\end{equation}
\end{definition}
The GROW $e$-variable maximizes the worst case expected logarithmic growth, while the REGROW $e$-variable maximizes the worst case expected logarithmic growth relative to an oracle that knows the values of the nuisance parameters.

The challenge remains in identifying maximal invariant statistics in general.
There are at least three approaches to finding maximal invariant statistics.
The first approach is via Ansatz --- postulate a statistic via problem intuition, verify that it is a maximal invariant, and derive the densities under null and alternative.
The second approach is constructive via the following lemma.
\begin{lemma}[Wijsman's Representation Theorem]
    \label{lem:wijsman}
Let $\pi(\rho)$ denote the corresponding right-Haar prior for the group $G$, then
\begin{equation}
    \label{eq:wijsman}
    \frac{q(m_n(\bY_n))}{p(m_n(\bY_n ))} = \frac{\int q_\rho(\bY_n)\pi(\rho)d\rho}{\int p_\rho(\bY_n)\pi(\rho)d\rho},
\end{equation}
\end{lemma}
The proof can be read in \citet{andersson1982}. 
While lemma \ref{lem:wijsman} does not require finding a maximal invariant statistic upfront, the integral representation can be difficult to calculate in closed form.

The third path proceeds via a reduction to sufficient statistics first, followed by applying invariance arguments.
Let $S_n(\bY_n) \in \mathcal{S}$ denote a set of sufficient statistics for the parameters. 
The group of transformations on the outcome space induces a group of transformations on the space of sufficient statistics.
For example, consider $y_i \sim N(\mu, \sigma^2)$ with $\theta = (\sigma, \mu / \sigma)$.
The sufficient statistic is $S_n = (\bar{y}_n, \sum_{i=1}^n y_i^2)$. The group action $g: \bY \rightarrow c \bY$ induces a group action $g: (\bar{y}_n, \sum_{i=1}^n y_i^2) \rightarrow (c\bar{y}_n, c^2\sum_{i=1}^n y_i^2)$. 

\begin{theorem}[C. Stein]
    \label{thm:stein}
Under the assumptions of \citet[Section ~II.3]{hall}, if $S_n$ is sufficient for $\theta$, and $m_n^s(S_n)$ is a maximal invariant under the induced group action on $S_n$, then $m_n^s(S_n)$ is sufficient for the invariant parameter $\xi$ and
\begin{equation}
    \label{eq:stein_mi_lr}
    E_n = \frac{q(m_n(\bY_n))}{p(m_n(\bY_n ))} = \frac{q(m^s_n(S_n))}{p(m^s_n(S_n))}.
\end{equation} 
\end{theorem}
A maximal invariant statistic $m_n^s(S_n)$ that is sufficient for $\xi$ is called \textit{invariantly sufficient}.
As the distribution of $m_n(\bY_n)$ depends only on $\xi$ under both the null and alternative, and $m_n^s(S_n)$ is sufficient for $\xi$, the identity in equation \eqref{eq:stein_mi_lr} follows from the Fisher-Neyman factorization theorem.
For example, the usual $t$-statistic $t_n(\bar{y}_n, \sum_{i=1}^n y_i^2) = \bar{y}_n / \sqrt{\sum_{i=1}^n y_i^2 - \bar{y}_n^2}$ is an invariantly sufficient statistic,
and the likelihood ratio of this is arguably easier to calculate than the likelihood ratio of the maximal invariant $m_n(\bY_n) = (y_1 / |y_1|, \dots, y_n / |y_1|)$ 

\section{Anytime-Valid Inference for Linear Models}
\label{sec:linear_model_results}
Consider a sequence of observations from the parametric (relaxed later) linear model 
\begin{equation}
\label{eq:base_model}
    y_i| \bbeta, \bdelta, \sigma^2 \overset{iid}\sim N(\bx_i'\bbeta + \bz_i'\bdelta,\sigma^2),
\end{equation}
with $\bbeta \in \mathbb{R}^{p}$ and $\bdelta \in \mathbb{R}^d$.
The parameterization in  \eqref{eq:base_model} splits the parameters into those of interest, $\bdelta$, and those that are nuisance, $\rho = (\bbeta,\sigma^2)$.
In matrix notation, $\bY_n|\bbeta, \bdelta, \sigma^2 \sim N(\bX_n\bbeta + \bZ_n\bdelta, \sigma^2 \bI_n)$, which we write concisely as $\bY_n| \bgamma, \sigma^2 \sim N(\bW_n \bgamma, \sigma^2 \bI_n)$ where $\bW_n = [\bX_n, \bZ_n]$ and $\bgamma' = [\bbeta', \bdelta']$.
Let $\mathcal{C}(\bA) = \{\sum_{i} c_i \bA_i : c_i \in \mathbb{R}, \bA_i = \bA \be_i\}$ denote the column space of a matrix $\bA$, $\mathcal{C}(\bA)^\perp$ the orthogonal complement of $\mathcal{C}(\bA)$, and $P_\bA = \bA (\bA ' \bA)^{-1}\bA'$ the orthogonal projection operator onto $\mathcal{C}(\bA)$.

Without loss of generality we take the null hypothesis $H_0: \bdelta = \bdelta_0 = \bzero$ as a test for nonzero $\bdelta_0$ can be obtained by applying the resulting test to $\bY_n - \bZ_n\bdelta_0$.
To apply invariance arguments we redefine the parameters as $(\bbeta, \sigma, \bxi)$, where $\bxi = \bdelta / \sigma$ are the standardized coefficients and $\rho = (\bbeta,\sigma)$ are the nuisance parameters, so that the null becomes $H_0: \bxi = \bxi_0 = 0$.
Notice that the transformation $g: \bY_n \rightarrow c\bY_n + \bX_\alpha$ leaves the model invariant by inducing a group on the parameters defined by $g: (\bbeta, \sigma, \bxi) \rightarrow (c\bbeta + \bX_n\balpha, c\sigma, \bxi)$
and that the standardized coefficients $\bxi$ are $G$-invariant. The set of summary statistics is $S(\bY_n) = (\hat{\bbeta}_n, \hat{\bdelta}_n, s^2_n) \in \mathcal{S} \coloneqq \mathbb{R}^p \times \mathbb{R}^d \times \mathbb{R}_{\geq 0}$.
The induced group action on $\mathcal{S}$ is $g: (\hat{\bbeta}_n, \hat{\bdelta}_n, s^2_n) \rightarrow (c\hat{\bbeta}_n + \balpha, c \hat{\bdelta}_n, c^2 s^2_n)$. 

Recall in a classical linear model theory the sampling distribution of the ordinary least squares estimator for $\bdelta$ is $\hat{\bdelta}_n | \sigma^2 \sim N(\bdelta, \sigma^2(\tilde{\bZ}_n'\tilde{\bZ}_n)^{-1})$, where $\tilde{\bZ}_n'\tilde{\bZ}_n = \bZ_n'(\pwn - \pxn)\bZ_n=\bZ_n'(\bI_n - \pxn)\bZ_n$.
Let $s^2_n(\bY_n) = \bY_n'(\bI_n-\pwn)\bY_n / \nu_n$ where $\nu_n = n-p-d$ denote the linear model estimator of $\sigma^2$.
We can estimate the parameter $\bxi$ via $\hat{\bxi}_n = \hat{\bdelta}_n / \sqrt{s^2_n(\bY_n)}$ and form the $F$-statistic
\begin{equation}
    F_n = \frac{\hat{\bxi}_n \tilde{\bZ}_n'\tilde{\bZ}_n \hat{\bxi}_n}{d}= \frac{\bY_n(\pwn - \pxn) \bY_n / d}{ \bY_n (\bI_n - \pwn) \bY_n / (n-p-d) },
\end{equation}
which has a noncentral $F$-distribution $F_n | \bxi \sim F(d, \nu_n, \bxi' \tilde{\bZ}_n'\tilde{\bZ}_n\bxi)$ (see Appendix \ref{app:classicalftest} for a review).
Hence, we can test the null hypothesis $\bxi = \bzero$ by comparing $F_n$ to an $F(d,\nu_n)$ distribution.

We now turn our attention to the anytime-valid analogue of this procedure.
A null hypothesis on the vector $H_0: \bdelta = \bdelta_0$ implies a collection of null probability measures $\mathcal{H}_0 = \mathcal{N}_{\bdelta_0}$.
It is pedagogically easier to first consider Bayesian alternatives with $\bdelta \sim N(\bzero, \sigma^2 \bPhi^{-1})$ and then frequentist point alternatives in Section \ref{sec:simple_alternative}.
Let $p(\bY_n | \bbeta, \sigma^2)$ and $q(\bY_n | \bbeta, \sigma^2)$ denote $N(\bX_n \bbeta, \sigma^2 \bI_n)$ and Bayes marginal $N(\bX_n\bbeta, \sigma^2 (\bI_n+\bZ_n\bPhi^{-1}\bZ_n'))$ densities, respectively.
Let $t_{\nu_n}(\cdot | \bzero, \bSigma)$ denote a multivariate $t$-density with $\nu_n$ degrees of freedom and mean $\bzero$ and covariance $\bSigma$.
\begin{theorem}
    \label{thm:bayesian_lm_e_variable}
    Under the linear model in \eqref{eq:base_model} with right-Haar prior $\pi(\bbeta,\sigma^2) \propto 1/\sigma^2$, the process
    \begin{equation}
        \label{eq:main_e_variable}
            \begin{split}
                E_n \coloneqq& \frac{\int q(\bY_n| \bbeta, \sigma^2) \pi(\bbeta, \sigma^2)  d\bbeta d\sigma^2}{\int p(\bY_n| \bbeta, \sigma^2) \pi(\bbeta, \sigma^2) d\bbeta d\sigma^2}\\
                =&\sqrt{\frac{\det(\bPhi)}{\det(\bPhi + \tilde{\bZ}_n'\tilde{\bZ}_n)}} \frac{\left(1+\frac{\hat{\bxi}_n'(\tilde{\bZ}_n'\tilde{\bZ}_n - \tilde{\bZ}_n'\tilde{\bZ}_n(\bPhi + \tilde{\bZ}_n'\tilde{\bZ}_n)^{-1}\tilde{\bZ}_n'\tilde{\bZ}_n)\hat{\bxi}_n}{\nu_n}\right)^{-\frac{\nu_n + d}{2}}}{\left(1+\frac{\hat{\bxi}_n'\tilde{\bZ}_n'\tilde{\bZ}_n\hat{\bxi}_n}{\nu_n}\right)^{-\frac{\nu_n + d}{2}}},
            \end{split}
        \end{equation}
        is a test martingale and also the GROW/REGROW $e$-variable for $H_0: \bdelta = \bzero$ against $H_1: \bdelta \sim N(\bzero, \sigma^2\bPhi^{-1})$. The statistic $\hat{\bxi}_n =  \hat{\bdelta}_n/ \sqrt{s^2_n}$ is invariantly sufficient and $E_n$ can be written
        \begin{equation}
            \label{eq:e_lr}
            E_n = \frac{t_{\nu_n}(\hat{\bxi}_n| \bzero, \bPhi^{-1} + (\tilde{\bZ}_n'\tilde{\bZ}_n)^{-1})}{t_{\nu_n}(\hat{\bxi}_n| \bzero, (\tilde{\bZ}_n'\tilde{\bZ}_n)^{-1})}.
        \end{equation}
\end{theorem}
The proof is broken down into a sequence of steps in Appendix \ref{sec:direct_calculation}. We derive expression \ref{eq:main_e_variable} by computing the integrals in the Wijsman representation (Lemma~\ref{lem:wijsman}).
Inspecting the statistic through which the outcomes $\bY_n$ enter the expression, we then prove by counterexample that $\hat{\bxi}_n$ is not a maximal invariant.
Instead, we prove $\hat{\bxi}_n$ is an invariantly sufficient statistic, which allows $E_n$ to be written as the likelihood ratio in Equation~\eqref{eq:e_lr} by Theorem~\ref{thm:stein}.
The GROW and REGROW optimality properties of $E_n$ follow from \citet{muriel}, which is discussed further in Section~\ref{sec:muriel}.
We now show that test martingale $E_n$ has power 1, under the definition in Section~\ref{sec:safe_testing}, when $\bxi \neq \bzero$.

\begin{theorem}
    \label{thm:almost_sure_limit}
    Suppose $\frac{1}{n}\bW_n'\bW_n$ converges almost surely to a positive definite matrix $\bOmega_{\bW}$, then the sequential test is consistent with asymptotic growth rate
    \begin{equation}
        \frac{\log E_n}{n} \overset{a.s.}{\rightarrow} \frac{1}{2} \log(1 + \bxi'\bOmega_{\tilde{\bZ}}\bxi),
    \end{equation}
    where $\bOmega_{\tilde{\bZ}}$ is the almost sure limit of $\frac{1}{n}\tilde{\bZ}_n'\tilde{\bZ}_n$.
\end{theorem}
The strong convergence of $\frac{1}{n}\bW_n'\bW_n$ to $\bOmega_{\bW}$ is automatically satisfied if the covariates are i.i.d.
Theorem \ref{thm:almost_sure_limit} implies that $E_n$ diverges almost surely under the alternative $\bxi \neq \bzero$, and thus $E_n$ crosses the rejection boundary of $\alpha^{-1}$ with probability 1.
Informally, this means the sequential test is guaranteed to eventually reject the null hypothesis in finite time, provided it is not stopped prematurely.

Recall that to test a nonzero point null $H_0: \bdelta_0 \neq \bzero$ we may simply replace $\bY_n$ with $\bY_n - \bZ_n\bdelta_0$, in which case $s^2_n(\bY_n - \bZ_n\bdelta_0) = s^2_n(\bY_n)$ and $\hat{\bxi}_n(\bY_n - \bZ_n\bdelta_0) = (\hat{\bdelta}_n - \bdelta_0) / s^2_n$.

\begin{corollary}
\label{cor:nonasymp_cs_F}
    Define $E_n^{\bdelta_0} \coloneqq E_n^{\bdelta_0}(\bY_n) = E_n(\bY_n - \bZ_n \bdelta_0)$, then $E_n^{\bdelta_0}$ is a test martingale for testing $H_0: \bdelta=\bdelta_0$.
    Define $p^{\bdelta_0}_n = 1/E^{\bdelta_0}_n$, then $p_n$ is a sequential $p$-value satisfying $\mathbb{P}_P[\exists n \in \mathbb{N}: p^{\bdelta_0}_n \leq \alpha^{-1}] \leq \alpha$ for all $P \in \mathcal{H}_0$. Define $C_n^\alpha =  \lbrace \bdelta \in \mathbb{R}^d : E^{\bdelta}_n \leq \alpha^{-1}\rbrace = \lbrace \bdelta \in \mathbb{R}^d : p^{\bdelta}_n > \alpha\rbrace$. 
    Then
    \begin{equation*}
    \label{eq:confidence_set}
                C_n^\alpha =\left\{ \bdelta \in \mathbb{R}^d : \|\bdelta - \hat{\bdelta}_n(\bY_n)\|^2_{\bA_n}\leq \nu_n s^2_n(\bY_n)\left(1-\left(\frac{\alpha^2 \det (\bPhi)}{\det(\bPhi + \tilde{\bZ}_n'\tilde{\bZ}_n)} \right)^{\frac{1}{\nu_n + d}}\right)\right\},\\
    \end{equation*}
    where
    \begin{equation*}
        \bA_n = \left(\left(\frac{\alpha^2 \det (\bPhi)}{\det(\bPhi + \tilde{\bZ}_n'\tilde{\bZ}_n)} \right)^{\frac{1}{\nu_n + d}}-1\right)\tilde{\bZ}_n'\tilde{\bZ}_n +\tilde{\bZ}_n'\tilde{\bZ}_n(\bPhi + \tilde{\bZ}_n'\tilde{\bZ}_n)^{-1}\tilde{\bZ}_n'\tilde{\bZ}_n,
    \end{equation*}
    is a confidence sequence satisfying $\mathbb{P}_P[\bdelta \in C_n^\alpha \text{ for all } n\in\mathbb{N}]\geq 1-\alpha$ for all $P$.
\end{corollary}

Lastly we define tests for linear hypotheses such as contrasts. 
\begin{corollary}
\label{cor:contrasts}
Under the linear model in \eqref{eq:base_model} and a full row rank matrix $\bE \in \mathbb{R}^{k \times d}$, then the process
\begin{equation}
    E_n=\sqrt{\frac{\det(\bPhi)}{\det(\bPhi + \tilde{\bZ}_n^{\dagger '}\tilde{\bZ}_n^\dagger)}} \frac{\left(1+\frac{\hat{\bxi}_n^{\dagger '}(\tilde{\bZ}_n^{\dagger '}\tilde{\bZ}_n^\dagger - \tilde{\bZ}_n^{\dagger '}\tilde{\bZ}_n^\dagger(\bPhi + \tilde{\bZ}_n^{\dagger '}\tilde{\bZ}_n^\dagger)^{-1}\tilde{\bZ}_n^{\dagger '}\tilde{\bZ}_n^\dagger)\hat{\bxi}_n}{\nu_n}\right)^{-\frac{\nu_n + k}{2}}}{\left(1+\frac{\hat{\bxi}_n^{\dagger '}\tilde{\bZ}_n^{\dagger '}\tilde{\bZ}_n^\dagger\hat{\bxi}_n^\dagger}{\nu_n}\right)^{-\frac{\nu_n + k}{2}}},
\end{equation}
is a test martingale for testing $H_0: \bE\bdelta = \bzero$, where $\tilde{\bZ}_n^\dagger = \tilde{\bZ}_n\bE' = (\bI-\pxn)\bZ_n\bE'$, $\hat{\bxi}_n^\dagger =\hat{\bdelta}_n^\dagger / \sqrt{s^2_n}$, $\hat{\bdelta}_n^\dagger = (\tilde{\bZ}_n^{\dagger '}\tilde{\bZ}_n^{\dagger})^{-1} \tilde{\bZ}_n^{\dagger ' }\bY_n$ .
\end{corollary}
Classically we would use the $F$-statistic $F_n = \hat{\bdelta}_n'\bE'(\bE(\tilde{\bZ}_n'\tilde{\bZ}_n)^{-1}\bE')^{-1}\bE\hat{\bdelta}_n  / (ks^2_n) \sim F(k, \nu_n)$ to test $H_0: \bE\bdelta = \bzero$.
In the appendix we show that the term in the denominator $\hat{\bxi}_n^{\dagger '}\tilde{\bZ}_n^{\dagger '}\tilde{\bZ}_n^\dagger\hat{\bxi}_n^\dagger$ is simply equal to  $k F_n$. As before, a test of $\bE\bdelta = \be_0$ can be obtained by replacing $\bY_n$ with $\bY_n - \bZ_n\bE'\be_0$.
\subsection{Simple Alternatives}
\label{sec:simple_alternative}
In the development so far, we have considered composite alternatives and created a test martingale using the method of mixtures.
Alternatively, given our invariantly sufficient statistic $\hat{\bxi}_n$, we can also consider simple alternatives.
Observe $\hat{\delta}_n / \sigma - \delta / \sigma \sim N(\bzero, (\tilde{\bZ}_n'\tilde{\bZ}_n)^{-1})$,
while $s^2_n / \sigma^2 \sim \chi^2_{\nu_n}$,
and therefore
\begin{equation*}
    \hat{\bxi}_n = \frac{\bdelta / \sigma + \hat{\bdelta}_n / \sigma - \bdelta / \sigma}{\sqrt{s^2_n / \sigma^2}} \sim t^{nc}_{\nu_n}(\bxi_, (\tilde{\bZ}_n'\tilde{\bZ}_n)^{-1}),
\end{equation*}
a $d$-dimensional noncentral multivariate $t$-distribution with $\nu_n$. 
The $e$-variable
\begin{equation}
    E_n = \frac{t^{nc}_{\nu_n}(\hat{\bxi}_n | \bxi_1, (\tilde{\bZ}_n'\tilde{\bZ}_n)^{-1})}{t^{nc}_{\nu_n}(\hat{\bxi}_n | \bxi_0, (\tilde{\bZ}_n'\tilde{\bZ}_n)^{-1})},
\end{equation}
is therefore a test martingale for $H_0: \bxi = \bxi_0$ and is GROW with respect to $H_1: \bxi = \bxi_1$.

\subsection{Comparison to \citet{muriel}}
\label{sec:muriel}
\citet{muriel} consider the problem of testing a \textit{univariate} simple null hypothesis $\xi = \xi_0$ vs a simple alternative $\xi = \xi_1$.
Their construction of an $e$-process is based on a likelihood ratio of a maximal invariant statistic, following the Ansatz approach described in Section~\ref{sec:group_invariance}.
Specifically, they propose a test statistic $\bm_n\coloneqq \bm_n(\bY_n) = \bA_n \bY_n / \|\bA_n \bY_n\|_2 \in \mathbb{R}^{n-p}$ where $\bA_n$ is an $(n - p) \times n$ matrix with rows forming an orthonormal basis $C(\bX_n)^\perp$. 
This can be achieved with any matrix $\bA_n$ satisfying $\bA_n'\bA_n = \bI_n - \pxn$. 
The statistic $\bm_n$ is actually a unit vector in $\mathbb{R}^{n-p}$ belonging to the $(n-p)$-dimensional unit-sphere and is a maximal invariant statistic.
Under the null hypothesis $\bm_n$ is uniformly distributed over the unit sphere.
Under the alternative hypothesis $\bm_n$ is concentrated around a unit vector as shown in Figure~\ref{fig:U_scatter}. 
Importantly, this distribution is not the von Mises–Fisher distribution; instead, it follows a nonstandard form attributed to \citet{Bhowmik2007} who, citing \citet{mathematica}, give the density

\begin{equation}
    \label{eq:muriel_density}
            p(\bm_n \mid \xi) = \frac{1}{2} \Gamma\left(\frac{k}{2}\right) \pi^{-\frac{k}{2}} e^{c(\xi)}\left\{A_k(a(\bm_n,\xi)) + B_k(a(\bm_n,\xi)) \right\}
\end{equation}
where
\begin{equation}
\begin{split}
            A_k(a(\bm_n,\xi)) &= {}_1F_1\left(\frac{k}{2}, \frac{1}{2}, \frac{a^2(\bm_n, \xi)}{2}\right),\\
            B_k(a(\bm_n,\xi)) &= \sqrt{2} a(\bm_n, \xi) \frac{\Gamma\left(\frac{1 + k}{2}\right)}{\Gamma\left(\frac{k}{2}\right)} 
    {}_1F_1\left(\frac{1 + k}{2}, \frac{3}{2}, \frac{a^2(\bm_n, \xi)}{2}\right),\\
    \end{split}
\end{equation}
$k = n-p$, $a(\bm_n, \xi) = \xi \bZ_n'\bA_n'\bm_n$, $c(\xi) = -\frac{1}{2}\xi^2 \bZ_n' \bA_n' \bA_n \bZ_n$ and ${}_1F_1$ is the confluent hypergeometric function. 
As noted by \citet[Chapter 2.3.3]{tyronmaster}, this density is cumbersome to work with.

In contrast, our construction uses the \textit{invariantly sufficient} statistic $\hat{\bxi}_n$.
By theorem \ref{thm:stein}, the likelihood ratio of the maximal invariant statistic $\bm_n$ must coincide with the likelihood ratio of our invariantly sufficient statistic $\hat{\xi}_n$. 
Specifically,
\begin{equation}
    E_n = \frac{p(\bm_n|\xi_1)}{p(\bm_n|\xi_0)}= \frac{p(\hat{\xi}_n|\xi_1)}{p(\hat{\xi}_n|\xi_0)}  = \frac{t^{nc}_{\nu_n}(\hat{\xi}_n | \xi_1, \|\tilde{\bZ}_n\|_2^{-2})}{t^{nc}_{\nu_n}(\hat{\xi}_n | \xi_0,\|\tilde{\bZ}_n\|_2^{-2})},
\end{equation}
In summary, we view our contribution as complementary to \citet{muriel}, with several key distinctions.
First, whereas \citet{muriel} focus on a univariate test of a single coefficient, our approach accommodates multivariate tests of multiple coefficients. 
Second, their method targets a fixed alternative $\xi = \xi_1$, while we consider composite alternatives $\bxi \neq \bxi_0$. 
Third, our approach yields confidence sequences that shrink as the sample size $n$ increases, a property not shared by inverting tests against point alternatives.
Fourthly, we avoid complications of working with nonstandard densities, as noncentral and location-scale $t$ densities come standard in most statistical software.
This makes our method more accessible and easier to implement in practice.
Lastly, we avoid additional complications of working with a nonstandard maximal invariant statistic $\bm_n$.
Instead, our method only requires the OLS estimate $\hat{\delta}_n$ and residual variance estimate $\sqrt{s^2_n}$ (or standard error) to construct our invariantly sufficient statistic $\hat{\bxi}_n$, which are typically provided in published works,
and provided as outputs in statistical software\footnote{Using standard notation for the univariate case, the density $t^{nc}_{\nu_n}(\hat{\xi}_n | \xi_1, \|\tilde{\bZ}_n\|_2^{-2})$ would be evaluated using pseudocode nct\_pdf($\hat{\xi}_n$, df = $\nu_n$, noncentrality = $\|\tilde{\bZ}_n\|_2 \xi_1$, scale = $1/\|\tilde{\bZ}_n\|_2$)}.

\subsection{Tuning the Mixture Hyperparameter $\Phi$}
\label{sec:hyperparameter_choice}
The $e$-process in equation \eqref{eq:main_e_variable} is a Bayes factor constructed using a $\bdelta \sim N(\bzero, \sigma^2\bPhi^{-1})$ ($\Rightarrow \bxi \sim N(\bzero, \bPhi^{-1})$) prior over alternatives, and is growth rate optimal against such alternatives.
In large internet companies with large scale experimentation efforts, one can choose $\bPhi$ so that the prior resemble the observed historical distribution of effects.

Alternatively, we can ask what the optimal choice of $\bPhi$ should be to maximize growth against frequentist alternatives for a specific $\bxi^\star$.
We present this result for the univariate $d=1$ case.
\begin{theorem}
    \label{thm:optimal_phi_frequentist}
    Let $\Phi^{\star}_n$ be the value of $\Phi$ that maximizes $\mathbb{E}[\log E_n]$, where the expectation is with respect to the frequentist alternative $\hat{\xi}_n \sim t_{\nu_n}^{nc}(\xi^\star, 1 / \|\tilde{\bZ_n}\|_2^2)$.
    Under the assumption $\|\tilde{\bZ}_n\|_2^2 = \Omega(n)$, then $\underset{n\rightarrow \infty}{\lim} \Phi^\star_n = (\xi^{\star})^{-2}$.
\end{theorem}
Theorem \ref{thm:optimal_phi_frequentist} shows that the optimal choice of $\Phi$, in the sense of maximizing the value of $\mathbb{E}[\log E_n]$ for large $n$, under a frequentist alternative, is given by matching the standard deviation of the prior equal to the true standardized effect $\Phi^{-1/2} = \xi^\star$.
Naturally, $\xi^\star$ is not known, so one can consider a minimum detectable effect (MDE) $\xi_{MDE}$. Setting $\Phi = \xi_{MDE}^{-2}$ therefore maximizes the worst case growth over $\xi \in [\xi_{MDE}, \infty)$.
In contrast, the Bayesian approach maximizes a weighted average growth rate, where the weighted average is determined by $N(\bzero, \bPhi^{-1})$, chosen the practitioner to upweight performance toward more plausible alternatives at the expense of downweighting performance away from less plausible alternatives.
The approach using an MDE, however, does not generalize easily to multivariate $\bdelta$. In the next section we discuss choosing $\bPhi$ in the context when little prior information is available.

\subsubsection{Automatic Defaults with $g$-Priors}
\label{sec:zellner_g_prior}
Suppose $\frac{1}{n}\tilde{\bZ}_n'\tilde{\bZ}_n$ converges almost surely to a known positive definite matrix $\bOmega_{\tilde{\bZ}}$.
For example if $\bZ_n$ is a column vector with $z_i = T_i-\rho$ where $T_i \overset{iid}{\sim} \text{Bernoulli}(\rho)$ is a treatment indicator, then $\bOmega_{\tilde{\bZ}} = \rho(1-\rho)$.
A valid prior would be the choice $\bPhi = g\bOmega_{\tilde{\bZ}}$ for some $g$.
In this case the information about $\bdelta$ in the likelihood and prior is approximately $n\bOmega_{\tilde{\bZ}}$ and $g\bOmega_{\tilde{\bZ}}$ respectively,
and so for $n > g$ we argue that the behavior of $E_n$ is driven predominantly by the data while the mixture has diminishing influence.
In particular $g=1$ resembles a unit information prior.

When $\bOmega_{\tilde{\bZ}}$ is not known, a popular choice of objective/default prior in the Bayesian literature for a fixed-$n$ analysis is Zellner's $g$-prior \citep{zellner}. 
The version that best suits our purposes here is $\bdelta | \sigma^2 \sim N(0, \sigma^2  \left(\frac{g}{n}\tilde{\bZ}_n'\tilde{\bZ}_n\right)^{-1})$.
Subjective Bayesians may object to this prior as it depends on the data through $\tilde{\bZ}_n$. Moreover, how can this prior be used in a sequential context, when it depends on $\tilde{\bZ}_n$?
We argue that Zellner's $g$-prior is an approximation to the fully anytime-valid procedure described earlier using $\bPhi = g\bOmega_{\tilde{\bZ}}$, except $\frac{1}{n}\tilde{\bZ}_n'\tilde{\bZ}_n$
is used as a plugin estimator for the unknown $\bOmega_{\tilde{\bZ}}$. This approximation yields an \textit{asymptotic} $e$-process that converges rapidly to the a valid $e$-process.

\begin{theorem}
    \label{thm:zellner_e_variable}
    Let $E_n$ be defined as in Theorem \ref{thm:bayesian_lm_e_variable} with $\bPhi = g\bOmega_{\tilde{\bZ}}$. Define for any $g\in\mathbb{N}$
    \begin{equation}
        \label{eq:zellner_e_process}
        G_n = \left(\frac{g}{g+n}\right)^{\frac{d}{2}}\left(\frac{1 + \frac{g}{g+n} \frac{d}{\nu_n} F_n}{1+\frac{d}{\nu_n} F_n}\right)^{-\frac{\nu_n + d}{2}}.
    \end{equation}
    Under the assumption $\frac{1}{n}\tilde{\bZ}_n'\tilde{\bZ}_n \overset{a.s.}{\rightarrow} \bOmega_{\tilde{\bZ}}$, then the approximation error is  
    $$|\log E_n - \log G_n| = \frac{1}{2}\Tr(\bOmega_{\tilde{\bZ}}^{-1}\bR_n) - \frac{g}{2}\bxi_n'\bR_n\bxi_n + O_{a.s.}\left(\frac{\bR_n}{n}\right),$$
    where $\bR_n = \bOmega_{\tilde{\bZ}} - \frac{1}{n}\tilde{\bZ}_n'\tilde{\bZ}_n.$
\end{theorem}
The appeal of $G_n$ is threefold.
First, the expression for $G_n$ is elegantly simple compared to $E_n$ in equation \eqref{eq:main_e_variable}.
It only depends on the $f$-statistic $F_n$ and a real-valued tuning parameter $g > 0$, and can therefore be considered a simple rule for converting a classical $f$-test to an anytime-valid test.

Second, the tuning parameter $g$ can be chosen to minimize the confidence region at a particular sample size $n$, independent of the covariates. The confidence set $C_n^\alpha$ from corollary \ref{cor:nonasymp_cs_F} with $\bPhi = \frac{g}{n}\tilde{\bZ}_n'\tilde{\bZ}_n$ reduces to
\begin{equation}
    \label{eq:F_confidence_sequence}
    \begin{split}
        C_n^\alpha &= \{\bdelta \in \mathbb{R}^d : F_n^\bdelta \leq R(g,n,\alpha)\},\\
        \text{where } R(g,n,\alpha) &= \frac{\nu_n}{d} \frac{1 - \left(\frac{\alpha^2 g^d}{(n+g)^d}\right)^{\frac{1}{\nu_n + d}}}{\max\left(0, \left(\frac{\alpha^2 g^d}{(n+g)^d}\right)^{\frac{1}{\nu_n + d}} - \frac{g}{n+g}\right)},
    \end{split}
\end{equation}
and $F_n^\bdelta = (\bdelta - \hat{\bdelta}_n)'\tilde{\bZ}_n'\tilde{\bZ}_n(\bdelta - \hat{\bdelta}_n) / (d s_n^2)$ is the $f$-statistic for testing the null $\bdelta$.
For a given $n$ and $\alpha$, the value of $g$ that provides the smallest confidence set can be determined by numerically minimizing the bound on the $f$-statistic $R(g,n,\alpha)$.

Thirdly, the choice $\bPhi = \frac{g}{n}(\tilde{\bZ}_n'\tilde{\bZ}_n)^{-1}$ possesses the following \textit{marginalization-consistency} property.
Consider two regressions, dropping the subscript $n$ for clarity
\begin{equation}
    \label{eq:subregression}
    \begin{split}
        \bY &= \bX\bbeta + \bZ\bdelta + \bvarepsilon,\\
        \bY &= \bX^\star\bbeta^\star + \bZ_1\bdelta_1 + \bvarepsilon',
    \end{split}
\end{equation}
where $\bZ = [\bZ_1, \bZ_2]$, $\bdelta' = (\bdelta_1', \bdelta_2')$ and where the second regression has absorbed $\bdelta_2$ via $\bX^\star = [\bX, \bZ_2]$ and $\bbeta^\star = [\bbeta', \bdelta_2']$.
A default prior procedure is marginalization consistent if the prior on the subvector $\bdelta_1$, obtained as the marginal of the default prior on the encompassing vector $\bdelta$, is exactly the same as the prior one would obtain by applying the procedure to $\bdelta_1$ directly.
In the following proposition we show that the $g$-prior is marginalization consistent.
\begin{proposition}
    \label{prop:marginalization_consistency}
    Consider the regressions in equation \eqref{eq:subregression} and let $\tilde{\bZ} = (\bI - \bP_{\bX})\bZ$ and $\tilde{\bZ}_1^\star = (\bI - \bP_{\bX^\star})\bZ$. Then $((\tilde{\bZ}'\tilde{\bZ})^{-1})_{11} = (\tilde{\bZ}_1^{\star ' }\tilde{\bZ}_1^\star)^{-1}$
\end{proposition}
Thus for $\bPhi = \frac{g}{n}(\tilde{\bZ}_n'\tilde{\bZ}_n)^{-1}$, the prior on any subset of coefficients is identical whether formed from direct construction or marginalizing an encompassing prior, which is crucial for testing nested models.

While $G_n$ only approximates $E_n$, we argue that the approximation error diminishes rapidly by theorem \ref{thm:zellner_e_variable} and that in practical applications the linear model itself is likely the dominant source of error, as rarely do we expect the linear model to hold exactly in the first place.
In the next section, we relax the assumption of Gaussian homoskedastic residuals, followed by an additional relaxation of the linearity assumption altogether in the context of randomized experiments.

\section{Heteroskedasticity Robust Asymptotic $e$-Processes}
\label{sec:heteroskedastic}
In this section we relax the assumption of Gaussian homoskedastic residuals. Specifically, consider the linear model
 $y_i = \bw_i'\bgamma + \varepsilon_i$ under the weaker assumption that $(\bw_i, \varepsilon_i)$ are i.i.d.\ with $\mathbb{E}[\varepsilon_i \bw_i]=0$ and $\mathbb{E}[\varepsilon_i \bw_i \bw_i'] = \bSigma$. 
This accommodates covariate level heteroskedasticity via $\mathbb{E}[\varepsilon_i^2  |\bw_i] = \sigma_i^2 = \sigma^2(\bw_i)$.
The OLS estimator $\hat{\bgamma}_n$ is a strongly consistent estimator of $\bgamma$ provided $\mathbb{E}[\varepsilon_i \bw_i]=0$ and $(1/n)\bW_n'\bW_n \overset{a.s.}{\rightarrow} \bOmega_{\bW}$, 

An asymptotic fixed-$n$ analysis proceeds by working with the asymptotic distribution of the OLS estimator, written
$\sqrt{n}(\hat{\bgamma}_n - \bgamma) = \left(\frac{1}{n}\sum_{i=1}^n \bw_i \bw_i'\right)^{-1}\left(\frac{1}{\sqrt{n}}\sum_{i=1}^n\bw_i\varepsilon_i\right)$. 
We sketch the argument here to parallel the anytime-valid construction.
Let $\bOmega_{\bW} = \mathbb{E}[\bw_i \bw_i']$, positive definite, then by the weak law $\hat{\bOmega}_{\bW_n} \coloneqq \frac{1}{n}\bW_n'\bW_n \overset{p}{\rightarrow} \bOmega_{\bW}$ and continuity of the inverse $\hat{\bOmega}_n^{-1}  \overset{p}{\rightarrow}\bOmega_{\bW}^{-1}$.
By the the central limit theorem $\frac{1}{\sqrt{n}}\sum_{i=1}^n \bw_i \varepsilon_i \overset{d}{\rightarrow} N(\bzero, \bSigma)$, 
and by Slutsky's theorem, $\sqrt{n}(\hat{\bgamma}_n - \bgamma) \overset{d}{\rightarrow} N(\bzero, \bOmega_{\bW}^{-1}\bSigma \bOmega_{\bW}^{-1})$.
Let $k = p + d$ and $\bE \in \mathbb{R}^{r\times k}$ full row rank matrix of linear hypotheses, then $\sqrt{n}\bE(\hat{\bgamma}_n - \bgamma) \overset{d}{\rightarrow} N(\bzero, \bSigma_{\bE})$ and $n (\hat{\gamma}_n - \bgamma)'\bE ' \bSigma_{\bE}^{-1} \bE (\hat{\bgamma}_n - \bgamma) \overset{d}{\rightarrow} \chi^2_r$ where $\bSigma_{\bE} = \bE\bOmega_{\bW}^{-1}\bSigma \bOmega_{\bW}^{-1}\bE'$.
Finally, $\hat{\bSigma}_n\coloneqq \frac{1}{n}\sum_{i=1}^n(y_i - \bw_i'\hat{\bgamma}_n)^2 \bw_i \bw_i'$ is a consistent estimator for $\bSigma$, and 
by Slutsky $n(\hat{\bgamma}_n - \bgamma)'\bE'\hat{\bSigma}_{\bE}^{-1}\bE(\hat{\bgamma}_n - \bgamma) \overset{d}{\rightarrow} \chi^2_r$, where $\hat{\bSigma}_{\bE}=\bE\hat{\bOmega}_{\bW}^{-1}\hat{\bSigma}_n\hat{\bOmega}_{\bW}^{-1}\bE'$.

To develop an asymptotic $e$-process, we follow the steps introduced by \citet{asymptoticcs}.
First, we construct an anytime-valid analogue of the fixed-$n$ $\chi^2$-test for Gaussian vectors.
\begin{lemma}
    \label{lem:multivariate_e_process}
Let $(\bG_i)_{i=1}^n$ be a sequence of i.i.d.\ $N(\bmu, \bSigma)$ random vectors in $\mathbb{R}^k$, then for any $g > 0$
\begin{equation}
    E_n \coloneqq \left(\frac{g}{g + n}\right)^{\frac{k}{2}}e^{\frac{1}{2}\frac{n}{g+n}Q_n},
\end{equation}
where $Q_n = n(\bar{\bG}_n-\bmu) \bSigma^{-1}(\bar{\bG}_n-\bmu) \sim \chi^2_k$ and $\bar{\bG} = \frac{1}{n}\sum_{i=1}^n \bG_i$, is a nonnegative supermartingale. 
\end{lemma}

Second, analogous to the asymptotic distribution of $\frac{1}{\sqrt{n}}\sum_{i=1}^n\bw_i \varepsilon_i$ via the central limit theorem, we leverage a strong approximation by \citet{einmahl1989} to approximate the sample path of $\sum_{i=1}^n\bw_i \varepsilon_i$ by
a sum of Gaussian vectors.
\begin{proposition}
    \label{prop:w_ie_i_strong_approximation}
Under the assumptions $\mathbb{E}[\varepsilon_i \bw_i] = 0$, $\mathrm{Cov}(\varepsilon_i\bw_i) = \bSigma \in \mathbb{R}^{k\times k}$, and
    $\mathbb{E}\bigl[\|\varepsilon_i\bw_i\|^s\bigr] \leq \infty$ for some $s > 3$.
    Then on a sufficiently rich probability space there is a sequence $(\bG_i)_{i=1}^\infty$ of i.i.d.\ Gaussian random vectors $\bG_i\sim N(\bzero,\bSigma)$, such that
    \begin{equation}
        \bigl\|\,\sum_{i=1}^n \varepsilon_i \bw_i - \sum_{i=1}^n \bG_i\,\bigr\|=O_{a.s.}\!\bigl(n^{\alpha}\bigr)
    \end{equation}
    for some $0<\alpha<\tfrac12$.
\end{proposition}

Last, we need to construct a consistent estimator of the nuisance parameter $\bSigma$.
\begin{proposition}
    \label{prop:strong_consistency}
    Let
            $\bSigma = \mathbb{E}[\varepsilon_i^2 \bw_i \bw_i']$, $\hat{\varepsilon}_{in} = y_i - \bw_i'\hat{\bgamma}_n$ and 
            $\hat{\bSigma}_n = \frac{1}{n-p-d}\sum_{i=1}^n \hat{\varepsilon}_{in}^2 \bw_i \bw_i'$.
Under the assumptions $\hat{\bOmega}_{\bW_n}\overset{a.s.}{\rightarrow} \bOmega_{\bW}$ and $\mathbb{E}[\bw_i\varepsilon_i] = 0$, then $\hat{\bgamma}_n \overset{a.s.}\rightarrow \bgamma$. Under additional assumptions $\mathbb{E}[\|\bw_i\|^4] < \infty$, $\mathbb{E}[|\varepsilon_i|\|\bw_i\|^3]<\infty$,  then $\hat{\bSigma}_n \overset{a.s.}{\rightarrow} \bSigma$ and $\hat{\bSigma}_{\bE}=\bE\hat{\bOmega}_{\bW}^{-1}\hat{\bSigma}_n\hat{\bOmega}_{\bW}^{-1}\bE'$ is a strongly consistent estimator of $\bSigma_{\bE}$.
\end{proposition}
With full details in the appendix, in addition to some required supplementary results, these three key results yield the following heteroskedasticity robust $e$-process.
\begin{theorem}
    \label{thm:robust_e_process}
    Let $\{(y_i,\bw_i)\}_{i=1}^\infty$ be a sequence of i.i.d.\ outcomes where $y_i = \bw_i'\bgamma + \varepsilon_i$ such that $\mathbb{E}[\bw_i\varepsilon_i] = 0$, $\mathbb{E}[\varepsilon_i^2 \bw_i \bw_i'] = \bSigma$, $\mathbb{E}[\bw_i \bw_i'] = \bOmega_{\bW}$, $\mathbb{E}[\|\bw_i\varepsilon_i\|^s]< \infty$ for some $s > 3$, $\mathbb{E}[\|\bw_i\|^4]< \infty$, $\mathbb{E}[|\varepsilon_i|\|\bw_i\|^3]<\infty$, and $\mathbb{E}[\bw_i y_i] < \infty$.
    For any full row rank matrix $\bE \in \mathbb{R}^{r \times k}$ and $g > 0 $ define
    \begin{equation}
        \label{eq:robust_e_variable}
        E_n := \left(\frac{g}{g+n}\right)^{\frac{r}{2}}\left(\frac{1 + \frac{g}{g+n} \frac{r}{\nu_n} \frac{Q_n}{r}}{1+\frac{r}{\nu_n} \frac{Q_n}{r}}\right)^{-\frac{\nu_n + r}{2}}.
    \end{equation}
    where $Q_n = n(\hat{\bgamma}_n )'\bE'\hat{\bSigma}_{\bE}^{-1}\bE (\hat{\bgamma}_n - \bgamma)$, $\hat{\bSigma}_{\bE}=(\bE\hat{\bOmega}_{\bW_n}^{-1}\hat{\bSigma}_n \hat{\bOmega}_{\bW_n}^{-1}\bE')$, $\hat{\bSigma}_n = \frac{1}{n-p-d}\sum_{i=1}^n(y_i - \bw_i'\hat{\bgamma}_n)^2 \bw_i \bw_i'$ and $\hat{\bOmega}_{\bW_n} = \frac{1}{n}\bW_n'\bW_n$. The process $(E_n)_{n=1}^\infty$ is an asymptotic $e$-process.
\end{theorem}
Note this gracefully recovers the previous results in the homoskedastic case. If we use $s^2_n$ to estimate $\sigma^2$, and instead use $\hat{\bSigma}_n = n s^2_n \hat{\bOmega}_{\bW_n}$, then $Q_n / r$ reduces to the $f$-statistic $F_n = (\hat{\bgamma}_n - \bgamma)\bE'(\bE (\bW_n'\bW_n)^{-1}\bE')^{-1}\bE(\hat{\bgamma}_n - \bgamma) / (r s^2_n)$.
For example, if  $\bE = (\bzero_{d\times p}, \bI_d)$ (now $r = d$) is used to test $\bdelta$, then $Q_n / r = (\hat{\bdelta}_n - \bdelta)'\tilde{\bZ}_n'\tilde{\bZ}_n(\hat{\bdelta}_n - \bdelta) / (d s_n^2) = F_n$, and equation \eqref{eq:robust_e_variable} recovers equation \eqref{eq:zellner_e_process}.
\subsection{Asymptotic ``$t$'' Confidence Sequences}
\label{sec:asymptotic_t_cs}
The heteroskedasticity robust $e$-process $E_n$ from equation \eqref{eq:robust_e_variable} provides an asymptotic ``$t$'' confidence sequence $C_n^\alpha = \{ \bgamma \in \mathbb{R}^k : Q_n  / r \leq R(g,n,\alpha)\}$, 
which differs from the asymptotic \textit{Gaussian} (``Robbins'') confidence sequence $C_n^{G,\alpha} = \{ \bgamma \in \mathbb{R}^k : Q_n  / r \leq R^G(g,n,\alpha)\}$,  used by \citet{asymptoticcs}, where
\begin{equation}
    \label{eq:radii}
        R(g,n,\alpha) = \frac{\nu_n}{r} \frac{1 - \left(\frac{\alpha^{\frac{2}{r}} g}{n+g}\right)^{\frac{r}{\nu_n + r}}}{\max\left(0, \left(\frac{\alpha^{\frac{2}{r}} g}{n+g}\right)^{\frac{r}{\nu_n + r}} - \frac{g}{n+g}\right)},\hspace{0.5cm}
        R^G(g, n, \alpha) = \frac{g + n}{n} \log\left(\frac{g+n}{\alpha^{\frac{2}{r}}g}\right).
\end{equation}
There are two key advantages to asymptotic $t$ confidence sequence. 
First, it aligns with the confidence sequence introduced in the first half of the paper for the parametric linear model. 
Second, while always a superset, it converges in the limit, made formal by the following theorem.
\begin{theorem}
    \label{thm:asmyptotic_radius}
    For any $\alpha \in (0,1)$ and $g > 0$, $R(g,n,\alpha) > R^G(g,n,\alpha)$ for all $n \in \mathbb{N}$ and $\underset{n\rightarrow \infty}{\lim}R(g,n,\alpha) / R^G(g,n,\alpha) \rightarrow 1$.
\end{theorem}
The increased size at small $n$ reflects of the $t$ confidence sequence's explicit accommodation of variance estimation uncertainty, in contrast to the Gaussian sequence, which assumes known variance and substitutes a plug-in estimator.
This property is especially important in the anytime-valid setting, where valid coverage is required uniformly over all $n$.
In such cases, the Gaussian sequence has been observed to over-reject at small sample sizes due to poor asymptotic approximations.
This concern motivates the “delayed-start” modifications in \citet[Section~2.6]{asymptoticcs} and \citet{bibaut2022}.
The $t$ confidence sequence naturally incorporates a delayed start: its width is infinite until $\left(\frac{\alpha^{\frac{2}{r}} g}{n+g}\right)^{\frac{r}{\nu_n + r}} > \frac{g}{n+g}$.
For example, using $r=k=1$, $\alpha=0.05$ and $g = 10^4$, requires a minimum sample size of $247$.
Alternatively, using $g=10^2$ requires a minimum sample size of $27$.
Figure \ref{fig:radii_comparison} in the appendix illustrates the difference between $R(g, n, \alpha)$ and $R^G(g, n, \alpha)$. 
The additional width at small $n$ provides robustness against early false positives, 
illustrated in Section~\ref{sec:examples}, allowing more time for the variance
estimate to stabilize and for the asymptotic behavior to become reliable.

\section{Treatment Effects in Randomized Experiments}
\label{sec:ate}
In this section we additionally remove the linearity assumption altogether within the context of randomized experiments. In particular, we allow the outcome variable $y_i$ to depend \textit{nonlinearly} on $\bw_i$.
In a potential outcomes framework we observe $y_i = T_i y_i(1) + (1-T_i)y_i(0)$, where $y_i(1)$ and $y_i(0)$ are the potential outcomes under treatment and control respectively and $T_i \sim \text{Bernoulli}(\rho)$ is a binary treatment indicator with assignment probability $\rho$.
The goal of many randomized experiments is to perform inference on the average treatment effect $\mathbb{E}[y_i(1) - y_i(0)]$ and conditional average treatment effects $\mathbb{E}[y_i(1) - y_i(0) | \bm_i]$.
In this section we make the stable unit treatment value assumption (SUTVA) and assume simple random sampling with one outcome per unit.
Despite the nonlinearity, consider the following lemma
\begin{lemma}
    \label{lem:population_least_squares}
Consider the population least squares defined as  
\begin{equation}
    \label{eq:population_least_squares}
    \begin{split}
        \bgamma^\star &\coloneqq (\alpha^\star, \bzeta^\star, \tau^\star, \beeta^\star) \\
        &= \underset{\alpha, \bzeta, \tau, \beeta}{\arg\min}\, \mathbb{E}\left[ (y_i - \alpha - (\bm_i - \bmu_m)'\bzeta - (T_i-\rho) \tau - T_i(\bm_i - \bmu_m)'\beeta)^2\right].
    \end{split}
\end{equation}
where $\alpha$ is an intercept, $T_i \sim \text{Bernoulli}(\rho)$ is a treatment indicator with assignment probability $\rho$, $\bm_i - \bmu_m$ are centered pre-treatment covariates, 
and $T_i(\bm_i - \bmu_m)$ is an interaction term. Write the regression as $y_i - \bw_i'\bgamma$. Under the assumption $\bOmega_{\bW}=\mathbb{E}[\bw_i\bw_i]'$ is positive definite and $\mathbb{E}[\bw_i y_i]$ is finite, then
$\bgamma{\star} = \bOmega_{\bW}^{-1}\mathbb{E}[\bw_i \by_i]$. In particular $\alpha^{\star} = \mathbb{E}[y_i] = \rho\mathbb{E}[y_i(1)] + (1-\rho)\mathbb{E}[y_i(0)]$ and $\bzeta^{\star}=\bOmega_{\bm}^{-1} \mathbb{E}[(\bm_i - \bmu_{\bm})y_i(0)]$, while
\begin{equation*}
\begin{split}
    \tau^{\star} =& \mathbb{E}[y_i(1) - y_i(0)]\\
    \beeta^{\star} =& \bOmega_{\bm}^{-1}\mathbb{E}[(\bm_i - \bmu_{\bm})y_i(1)] - \bOmega_{\bm}^{-1}\mathbb{E}[(\bm_i - \bmu_{\bm})y_i(0)],
\end{split}
\end{equation*}
where $\bOmega_{\bm} = \mathbb{E}[(\bm_i - \bmu_m)(\bm_i - \bmu_m)']$.
\end{lemma}
Observe $\tau^{\star}$ is the average treatment effect.
Also observe that treatment effect homogeneity, $\mathbb{E}[y_i(1) - y_i(0)| \bm_i = \bm] = \tau^\star$ for all $\bm$, implies $\beeta^\star = \bzero$.
The central idea is to estimate the population least squares using the sample ordinary least squares estimator, which we can use to test hypotheses about $\tau^\star$ and $\beeta^\star$.
First, observe the following lemma.

\begin{lemma}
    \label{lem:asymptotic_ols_params}
    Define $\varepsilon_i = y_i - \bw_i'\bgamma^\star$, then $\mathbb{E}[\bw_i \varepsilon_i] = \bzero$
\end{lemma}
In Section \ref{sec:heteroskedastic} we assumed a generative model $y_i = \bw_i'\bgamma + \varepsilon_i$. 
Instead, here we \textit{define} the residuals to be $\varepsilon_i = y_i - \bw_i'\bgamma^\star$, so that $y_i = \bw_i'\bgamma^\star + \varepsilon_i$ is true \textit{by definition}.
Hence the conditions described in Section \ref{sec:heteroskedastic} are satisfied and we can apply the preceding results for average treatment effects even if the outcome depends nonlinearly on the covariates.
In particular, under the assumptions of proposition \ref{prop:strong_consistency}, $\hat{\bgamma}_n$ and $\hat{\bSigma}_n$ are strongly consistent estimators of $\bgamma^\star$ and $\bSigma = \mathbb{E}[\varepsilon_i \bw_i \bw_i']$,
and theorem \ref{thm:robust_e_process} can be used to test for average or conditional average treatment effects.
For example, $\bE = (0, \bzero', 1, \bzero')$ allows us to test the average treatment effect $\bE(\hat{\bgamma}_n - \bgamma^\star) = \hat{\tau}_n - \tau^\star$. Alternatively,
$\bE$ can be defined such that $\bE(\hat{\bgamma} - \bgamma^\star) = \hat{\bv}_n - \bv^\star$ where $\hat{\bv}_n' = (\hat{\tau}_n, \hat{\beeta}_n')$ and $(\bv^\star)' = (\tau^\star, (\beeta^\star)')$ allowing us to test for \textit{any} treatment effects, or simply defined such that $\bE(\hat{\bgamma} - \bgamma^\star) = \hat{\beeta}_n - \hat{\beeta}^\star$, allowing us to test for the presence of conditional average treatment effects.

\section{Examples}
\label{sec:examples}
The regression specified in equation \eqref{eq:population_least_squares} is the \textit{fully-interacted} regression as it includes interaction terms, which are optional.
In the following simulations we simply fit the regression $\mathbb{E}[y_i | \bw_i] = \alpha + \bm_i'\bzeta + T_i \tau = \bw_i'\bgamma$.
We examine the asymptotic $t$, asymptotic Gaussian, and linear model confidence sequences i.e.
$\hat{\tau}_n \pm rse(\hat{\tau}_n) \sqrt{R(g,n,\alpha)}$, $\hat{\tau}_n \pm rse(\hat{\tau}_n) \sqrt{R^G(g,n,\alpha)}$, and $\hat{\tau}_n \pm se(\hat{\tau}_n) \sqrt{R(g,n,\alpha)}$ respectively, where $R$ and $R^G$ are defined in equation \eqref{eq:radii} with dimensionality $r=1$.
The standard and robust standard errors are $se(\hat{\tau}_n) = \sqrt{s^2_n} (\bW_n\bW_n)^{-1}_{\tau\tau}$ and $rse(\hat{\tau}_n) = \sqrt{((\bW_n'\bW_n)^{-1}\bW_n'\hat{\bSigma}_n\bW_n(\bW_n'\bW_n)^{-1})_{\tau\tau}}$ respectively, where $s^2_n = \frac{1}{ \nu_n}\sum_{i=1}^n e_i^2$, $\hat{\bSigma}_n = \frac{1}{\nu_n}\sum_{i=1}^n e_i^2 \bw_i \bw_i'$, and $e_i = y_i - \bw_i'\hat{\bgamma}_n$.
The corresponding (asymptotic) $e$-processes are 

\begin{equation}
    \label{eq:E_n_E_n_G_sim}
        E_n = \left(\frac{g}{g+n}\right)^{\frac{1}{2}}\left(\frac{1 + \frac{g}{g+n} \frac{Q_n}{\nu_n}}{1+\frac{Q_n}{\nu_n}}\right)^{-\frac{\nu_n + 1}{2}},\quad
        E^G_n = \left(\frac{g}{g+n}\right)^{\frac{1}{2}}e^{\frac{1}{2}\frac{n}{g+n}Q_n},
\end{equation}
and
\begin{equation}
    \label{eq:G_n_sim}
        G_n = \left(\frac{g}{g+n}\right)^{\frac{1}{2}}\left(\frac{1 + \frac{g}{g+n} \frac{F_n}{\nu_n}}{1+\frac{F_n}{\nu_n}}\right)^{-\frac{\nu_n + 1}{2}}
\end{equation}
where $Q_n = (\hat{\tau}_n/ rse(\hat{\tau}_n))^2$, $F_n = (\hat{\tau}_n / se(\hat{\tau}_n))^2$.
We examine the first time the $e$-process exceeds $\alpha^{-1}$ $\iff$ $(1-\alpha)$ confidence sequence excludes zero.
Note the OLS estimate $\hat{\tau}_n$ is identically equal to the AIPW estimate using the same regression function, provided the assignment probabilities are constant, and so
the asymptotic Gaussian confidence sequence is equal to the AIPW procedure of \citet{asymptoticcs} except no sample-splitting is required.
\subsection{Simulated Nonlinear Regression with $t$-residuals}
\label{sec:simulated_nonlinear}
We examine the robustness to model misspecification by simulating from a known model with clear violations of the linear model assumptions.
Figure~\ref{fig:asymp_t_simulated_combined} visualizes the probability of rejection when outcomes are generated from a model which is nonlinear in the covariates, has heavy-tailed, heteroskedastic residuals, and has an ATE of $\delta$.
With only three degrees of freedom, the $t$ residuals have finite variance but infinite higher order moments.
Figure~\ref{fig:asymp_t_simulated_null} demonstrates that the asymptotic $t$ confidence sequence provides better protection against early false positives than the asymptotic Gaussian confidence sequence.
It further demonstrates that the Type-I error is better calibrated at the nominal $\alpha$ for larger $g$.
Note also that while the empirical Type-I error may appear less than the nominal $\alpha$, we are only examining the rejections over $[1, 10^4]$, whereas the guarantee is for all $n \in \mathbb{N}$.

\begin{figure}[!ht]
  \centering
  
  \begin{subfigure}{\textwidth}
    \centering
    \includegraphics[width=\linewidth]{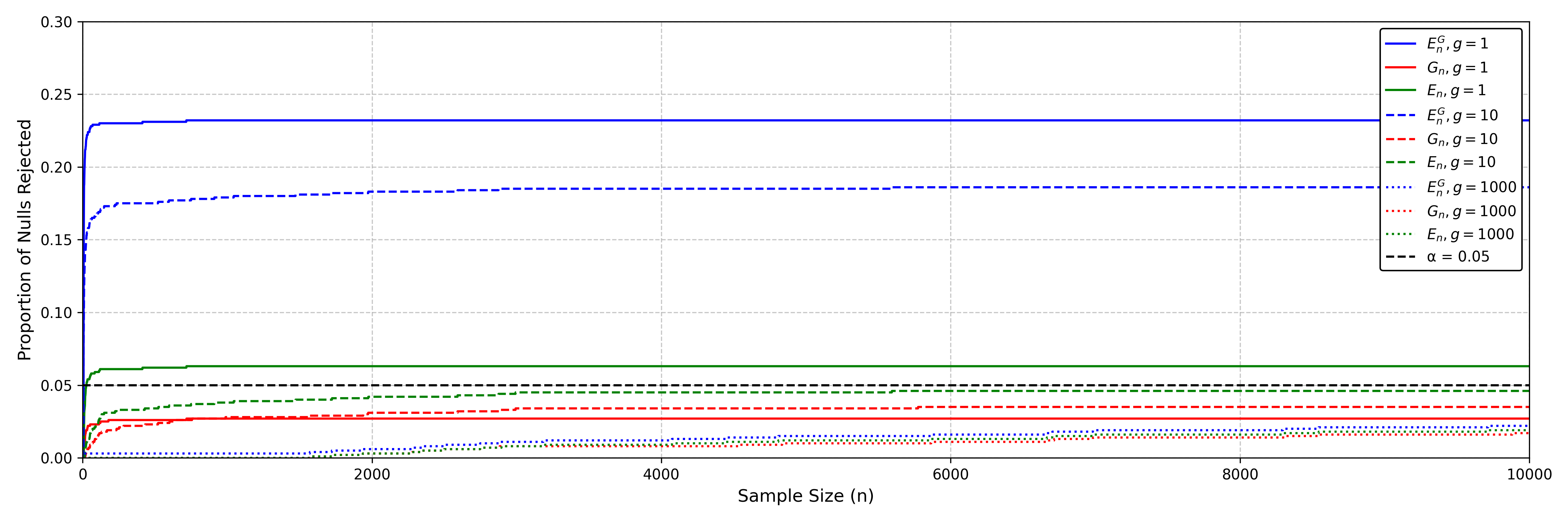}
    \caption{Null Simulation: $\delta = 0$}
    \label{fig:asymp_t_simulated_null}
  \end{subfigure}

  \begin{subfigure}{\textwidth}
    \centering
    \includegraphics[width=\linewidth]{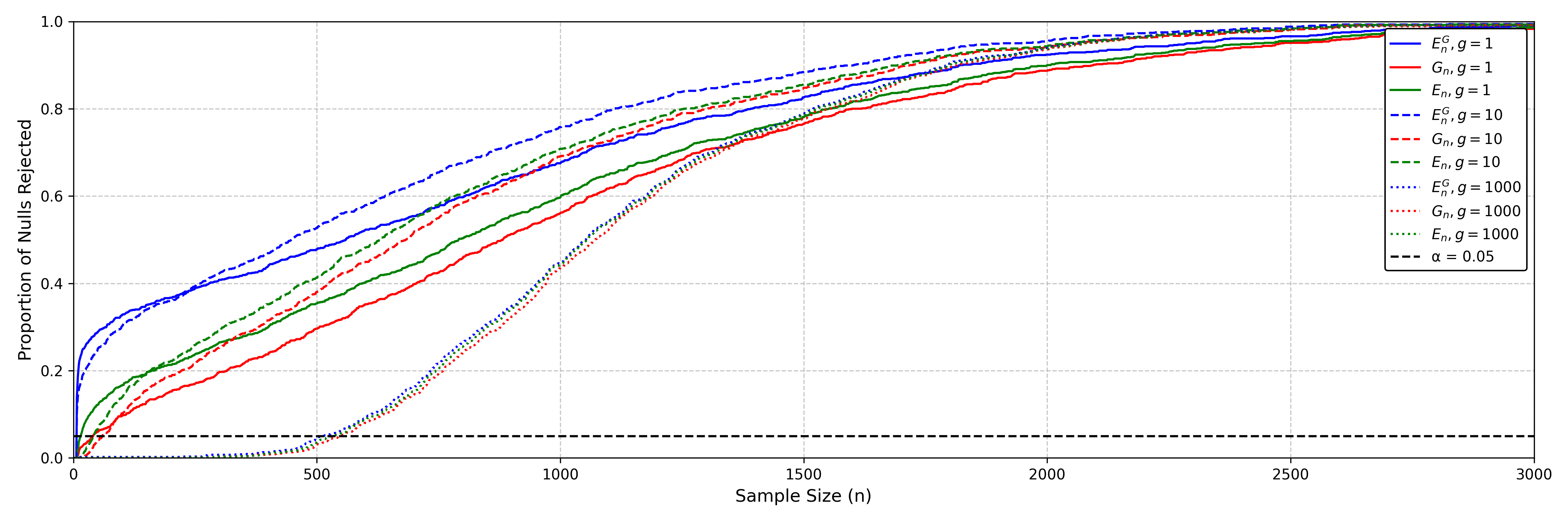}
    \caption{Alternative Simulation: $\delta = 1$}
    \label{fig:asymp_t_simulated_alt}
  \end{subfigure}
  \caption{Proportion of $e$-processes which have crossed the rejection threshold $\alpha^{-1}$ ($\alpha = 0.05$) by sample size $n$, based on $10^3$ simulations from the model
  $y_i = 1 - 2x_{i1}^2 - 2\sin x_{i2} + 3|x_{i3}| + z_i \delta + (1.5 + x_{i1}^2)\epsilon_i$,
where $\epsilon_i \sim t_3$ and $\bx_i \sim N(\bzero_3, \bSigma)$ with $\bSigma_{ij}=0.8^{|i-j|}$, $z_i \sim \text{Bernoulli}(1/2)$. 
Colors correspond to $e$-processes in equations \eqref{eq:E_n_E_n_G_sim} and \eqref{eq:G_n_sim}. Linestyles correspond to $g \in \{1,10,1000\}$.}
  \label{fig:asymp_t_simulated_combined}
\end{figure}

\subsection{PlayDelay: Heavy-Tailed and Right-Skewed Outcomes}
\label{sec:play_delay}
This dataset comes from a real Netflix A/B test in which devices were randomly assigned to a new software version. The outcome, \textit{PlayDelay}, measures the time between a playback request and its start. 
Figure~\ref{fig:jasa_play_delay_cdf} (appendix) shows empirical CDFs of the normalized PlayDelay (scaled by the dataset’s maximum value).
The skew and kurtosis of the data are approximately 100 and 13,000 respectively, indicating an extremely heavy-tailed distribution with significant right-skew and a high propensity for extreme outliers.
The $99.9$-percentile is $0.02$ while the largest value is $1$.
We fit the adjusted regression $y_i = \alpha + m_i \zeta + T_i \tau + \varepsilon_i$
where $m_i$ is a pre-treatment measurement of PlayDelay, strongly correlated with $y_i$, and $T_i$ is a treatment indicator.
As the outcomes are bounded we additionally compare with the Empirical Bernstein confidence sequence of \citet{howard} (Appendix~\ref{sec:aipw}).
The empirical Type-I error is shown in Figure \ref{fig:asymp_t_playdelay_null}. Note that testing only begins when $\nu_n > 0$.
Again, we observe better protection against early false positives for the asymptotic $t$ confidence sequence, which improves with larger $g$.
Considering the alternative simulation, we remark that these procedures rarely stop before $2^{20}$ observations. 
Extreme‐value outliers inflate the empirical variance on each arrival, which in turn broadens the confidence sequence and delays stopping.
By contrast, a log‐transform compresses the tail mass and stabilizes variance, yielding much earlier stopping times. For these reasons, we present the alternative simulation in Figure~\ref{fig:asymp_t_playdelay_alt} for log-PlayDelay.
In both cases, we observe that the price paid for the nonasymptotic guarantees provided by the Empirical Bernstein confidence sequence is a much larger stopping time.
\begin{figure}[!ht]
  \centering
  
  \begin{subfigure}{\textwidth}
    \centering
    \includegraphics[width=\textwidth]{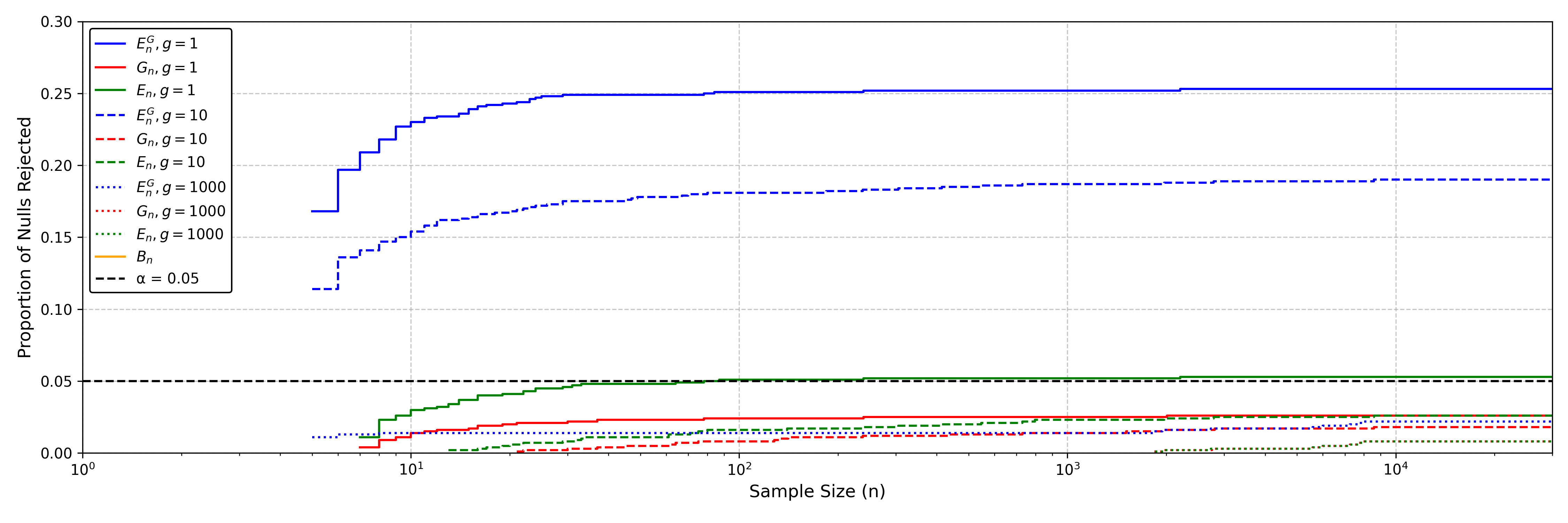}
    \caption{Null Simulation: A $T_i\sim\text{Bernoulli}(1/2)$ treatment indicator is drawn and the outcome is sampled with replacement from the control dataset.}
    \label{fig:asymp_t_playdelay_null}
  \end{subfigure}
  
  \begin{subfigure}{\textwidth}
    \centering
    \includegraphics[width=\textwidth]{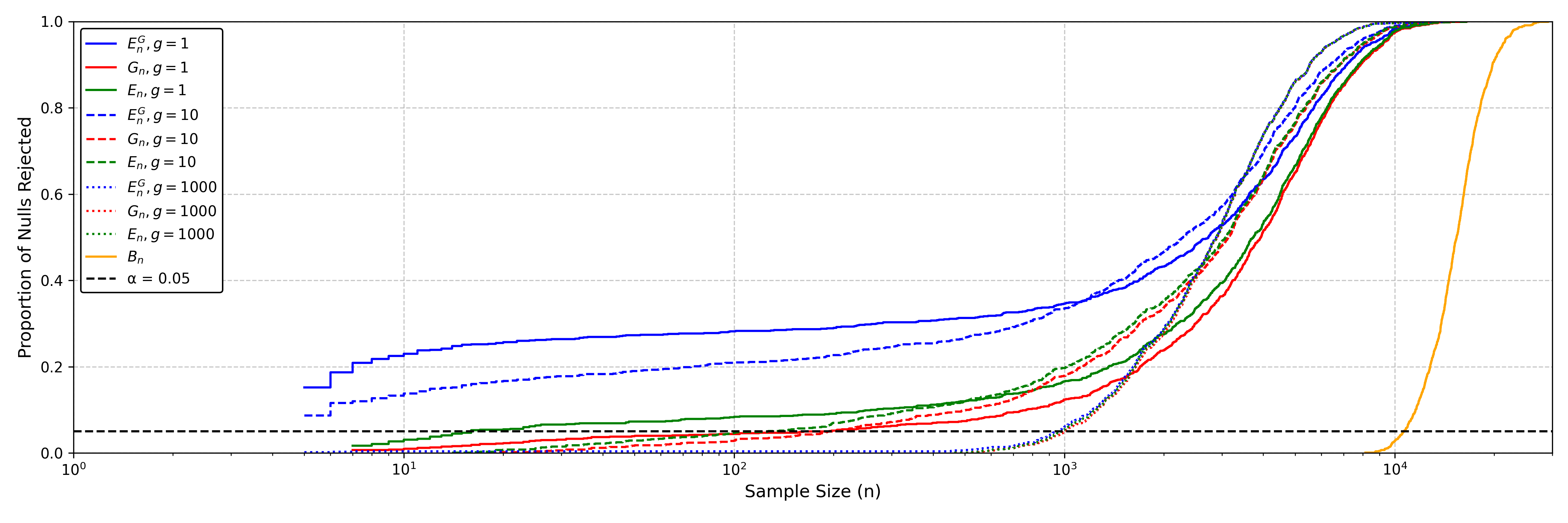}
    \caption{Alternative Simulation: A $T_i\sim \text{Bernoulli}(1/2)$ treatment indicator is drawn and the log-outcome is sampled with replacement from the control (treatment) dataset if $T_i = 0$ ($T_i=1$).}
    \label{fig:asymp_t_playdelay_alt}
  \end{subfigure}
  \caption{Proportion of $e$-processes which have crossed the rejection threshold $\alpha^{-1}$ ($\alpha = 0.05$) by sample size $n$, based on $10^3$ simulations.
Colors correspond to $e$-processes defined in equations \eqref{eq:E_n_E_n_G_sim}, \eqref{eq:G_n_sim} and Appendix~\ref{sec:aipw}. Linestyles correspond to $g \in \{1,10,1000\}$.}
  \label{fig:asymp_t_playdelay_combined}
\end{figure}

\section{Conclusion}
\label{sec:conclusion}
Motivated by the foundational role of linear models in practice, we outlined how to perform safe anytime-valid inference in this setting. 
The proposed tests (confidence sequences) provide \textit{time-uniform} Type-I error control (coverage) and are straightforward to implement:
the test statistic is simply the likelihood ratio of the standardized OLS vector $\hat{\bxi}_n = \hat{\bdelta}_n / \sqrt{s^2_n}$.
The resulting $e$-variable is optimal in the GROW and REGROW senses, for both point and composite alternatives. 
For composite alternatives, we introduced an approximate, automatic construction with a closed-form expression, enabling practitioners to convert $F$-statistics into anytime-valid $e$-processes.

In the second half of the paper, we relaxed linear model assumptions—first allowing heteroskedasticity, 
then dropping linearity altogether in the context of randomized experiments. 
This led to asymptotic, nonparametric $e$-processes for average and conditional average treatment effects under simple randomization. 
These tools allow practitioners to continuously monitor experiments and rigorously detect harmful or beneficial treatment effects sooner.

Our analysis is limited to independent outcomes. 
An important extension is to dependent outcomes, common in longitudinal and panel data settings with repeated measurements and within-unit serial correlation. 
We conjecture that the form of the heteroskedastic $e$-process may still apply in this case, but we leave this to future work.

\bibliographystyle{agsm}
\bibliography{main}

\appendix
\newpage
\begin{center}
{\large\bf SUPPLEMENTARY MATERIAL}
\end{center}

\section{Review: The classical fixed-$n$ $F$-Test} \label{app:classicalftest}
As this section concerns the fixed-$n$ case, we drop the $n$ superscript from $\bY_n$ to simplify the exposition.
An $\alpha$-level test of $H_0: \bdelta=0$, without loss of generality,  can be obtained by examining the likelihood ratio test statistic
\begin{equation}
\label{eq:likelihood-ratio}
    \Lambda(\bY) := \frac{\underset{\btheta \in \bTheta_1}{\sup} p(\bY|\btheta)}{\underset{\btheta \in \bTheta_0}{\sup} p(\bY|\btheta)}
\end{equation}
and rejecting the null hypothesis when $\Lambda(\bY) > c_\alpha$ for some constant $c_\alpha>0$ suitably chosen to provide a Type-I error probability of at most $\alpha$. The following lemma recalls the classical likelihood ratio test construction of the $F$-test.
\begin{theorem}
The likelihood ratio can be expressed
\begin{equation}
    \Lambda(\bY) = 1+ \frac{d}{n-p-d}f(\bY)
\end{equation}
where the $f$-statistic is defined as
\begin{equation}
\label{eq:fstatistic}
       f(\bY)= \frac{\frac{\bY'(\pw-\px)\bY}{d}}{\frac{\bY'(\bI-\pw)\bY}{n-p-d}} =\frac{\hat{\bdelta}(\bY)\tilde{\bZ}'\tilde{\bZ}\hat{\bdelta}(\bY)}{d s^2(\bY)}= \frac{\bt(\bY)'\bt(\bY)}{d}.
\end{equation}
Then $\Lambda(\bY) > c_\alpha$ $\iff$ $f(\bY) > f_\alpha$ for some $f_\alpha > 0$. 
The distributions of the $f$-statistic under $H_1$ and $H_0$ are
\begin{equation}
    \begin{split}
        f(\bY)|\bbeta, \bdelta, \sigma^2, H_1 &\sim F (d, n-p-d, \bdelta'\tilde{\bZ}'\tilde{\bZ}\bdelta/\sigma^2) \\ 
        f(\bY)|\bbeta, \sigma^2, H_0 &\sim F (d, n-p-d, 0) \\ 
    \end{split}
\end{equation}
Rejecting when $f(Y) > f_\alpha$, with $f_\alpha$ denoting the $1-\alpha$ quantile $F (d, n-p-d,0)$ yields a fixed-$n$ test with Type-I error probability $\alpha$.
\label{thm:classicalf}
\end{theorem}

Note that the $f$ statistic can be written in terms of the maximal invariant statistic $\bt(\bY)$. In the case of $d=1$, when there is only a single main effect, then $f$ can be identified as the square of the usual $t$-statistic ($t \sim t_{n-p-1} \Rightarrow t^2 \sim F(1,n-p-1)$). A $p$-value can be calculated by computing $\mathbb{P}[f \geq f(Y)]$ under the null $F(d,n-p-d,0)$ distribution. 

\begin{proof}
Starting with the denominator in \eqref{eq:likelihood-ratio}, consider first expressing the quadratic form in the Guassian likelihood as a component in $\mathcal{C}(\bX)$ and a component in $\mathcal{C}(\bX)^\perp$.
\begin{equation}
\begin{split}
        \|\bY - \bX\bbeta \|_2^2 &= \|\px(\bY - \bX\bbeta)\|_2^2+\|(\bI-\px)(\bY - \bX\bbeta)\|_2^2\\
        &=\|\px\bY - \bX\bbeta\|_2^2+\|(\bI-\px)\bY\|_2^2
\end{split}
\end{equation}
This is minimized by setting $\hat{\bbeta} = (\bX'\bX)\bX'\bY$, which sets the first term to zero.
The likelihood is then maximized by setting $\hat{\sigma^2} = \|(\bI-\px)\bY\|_2^2 / n$. It follows that
\begin{align}
\underset{\btheta \in \bTheta_0}{\sup}p(\bY|\btheta) = \left(\frac{n}{2\pi}\right)^{\frac{n}{2}}\left(\frac{1}{\|(\bI-\px)\bY\|_2^2}\right)^{\frac{n}{2}}e^{-\frac{n}{2}}
\end{align}
Now consider the numerator in \eqref{eq:likelihood-ratio}, expressing the quadratic form in the Gaussian likelihood as a component in $\mathcal{C}(\bW)$ a component in $\mathcal{C}(\bW)^\perp$.
\begin{equation}
\begin{split}
        \|\bY - \bX\bbeta - \bZ\bdelta\|_2^2  = \|\bY - \bW \bgamma\|_2^2 &= \|\pw(\bY - \bW\bgamma)\|_2^2+\|(\bI-\pw)(\bY - \bW\bgamma)\|_2^2\\
        &= \|\pw(\bY - \bW\bgamma)\|_2^2+\|(\bI-\pw)\bY\|_2^2\\
\end{split}
\end{equation}
where $\bW = [\bX, \bZ]$ and $\bgamma' = (\bbeta', \bdelta')$. Applying the same reasoning as before, this is minimized by setting $\hat{\bgamma} = (\bW'\bW)^{-1}\bW'Y$, which sets the first term to zero. The likelihood is then maximized by setting $\hat{\sigma}^2 = \|(\bI-\pw)\bY\|_2^2/n$. It follows that 
\begin{align}
\underset{\btheta \in \bTheta_1}{\sup}p(\bY|\btheta) = \left(\frac{n}{2\pi}\right)^{\frac{n}{2}}\left(\frac{1}{\|(\bI-\pw)\bY\|_2^2}\right)^{\frac{n}{2}}e^{-\frac{n}{2}}
\end{align}
and therefore
\begin{equation}
    \Lambda(\bY) = \left(\frac{\|(\bI-\px)\bY\|_2^2}{\|(\bI-\pw)\bY\|_2^2}\right)^{\frac{n}{2}}.
\end{equation}
However, the vector in the numerator be expressed as a component in $\mathcal{C}(\bW)$ and a component in $\mathcal{C}(\bW)^\perp$.
\begin{equation}
\begin{split}
        \|(\bI-\px)\bY\|_2^2 &= \|\pw(\bI-\px)\bY\|_2^2 + \|(\bI-\pw)(\bI-\px)\bY\|_2^2\\
         &= \|(\pw-\px)\bY\|_2^2 + \|(\bI-\pw)\bY\|_2^2\\
        &= \|(\pw-\px)\bY\|_2^2 + \|(\bI-\pw)\bY\|_2^2\\
\end{split}
\end{equation}
and so the likelihood ratio can be written in terms of the $f$-statistic as
\begin{equation}
    \Lambda(\bY) = 1+ \frac{d}{n-p-d}f(\bY).
\end{equation}
To show $f(\bY)$ can be expressed in terms of $\hat{\bdelta}(\bY)$ as in equation \eqref{eq:fstatistic}, note simply that $(\pw-\px)\bY = (\bI-\px)\pw\bY = (\bI-\px)(\bX\hat{\bbeta}(\bY) + \bZ\hat{\bdelta}(\bY)) = (\bI-\px)\bZ\hat{\bdelta}(\bY) = \tilde{\bZ}\hat{\bdelta}(\bY)$.
\end{proof}

A test of the null hypothesis $H_0: \bdelta=\bdelta_0$ can easily be obtained from a hypothesis test of $\bdelta = \bzero$ by replacing $\bY$ with $\bY-\bZ\bdelta_0$. In this case, the $f$- statistic becomes

\begin{equation}
       f(\bY)= \frac{\frac{(\bY-\bZ\bdelta_0)'(\pw-\px)(\bY-\bZ\bdelta_0)}{d}}{\frac{(\bY-\bZ\bdelta_0)'(\bI-\pw)(\bY-\bZ\bdelta_0)}{n-p-d}} =\frac{(\hat{\bdelta}(\bY)-\bdelta_0)'\tilde{\bZ}'\tilde{\bZ}(\hat{\bdelta}(\bY)-\bdelta_0)}{d s^2(\bY)}
\end{equation}

By finding the set of null-values that would not be rejected by this test one obtains a confidence set for the vector $\bdelta$.
\begin{corollary}
A $1-\alpha$ confidence set for $\bdelta$ is provided by
\begin{equation}
    \label{eq:classicalconfidenceset}
    \mathcal{C}_\alpha(\bY) := \{\bdelta : (\hat{\bdelta}(\bY)-\bdelta)'\tilde{\bZ}'\tilde{\bZ}(\hat{\bdelta}(\bY)-\bdelta) \leq d s^2(\bY) f_\alpha\},
\end{equation}
\end{corollary}

\section{Proofs for Section \ref{sec:linear_model_results}} 
\begin{lemma}
    \label{lem:Y_decomposition}
    $$\|\bY - \bW \bgamma\|_2^2 =\|\bX(\bbeta - \tilde{\bbeta})\|_2^2+\|\tilde{\bZ}(\hat{\bdelta} - \bdelta)\|_2^2+\|(\bI-\pw)\bY\|_2^2$$
where $\tilde{\bbeta} = \hat{\bbeta}+(\bX'\bX)^{-1}\bX'\bZ(\hat{\bdelta}-\bdelta)$.
\end{lemma}
\begin{proof}
\begin{equation*}
\begin{split}
\|\bY - \bW \bgamma\|_2^2 &= \|\pw(\bY - \bW\bgamma)\|_2^2+\|(\bI-\pw)(\bY - \bW\bgamma)\|_2^2\\
        &= \|\pw(\bY - \bW\bgamma)\|_2^2+\|(\bI-\pw)\bY\|_2^2\\
        &= \|\px\pw(\bY - \bW\bgamma)\|_2^2+\|(I-\px)\pw(\bY - \bW\bgamma)\|_2^2+\|(\bI-\pw)\bY\|_2^2\\
        &= \|X\hat{\bbeta}+\px\bZ\hat{\bdelta} - X\bbeta - \px\bZ\bdelta\|_2^2+\|(I-\px)\bZ(\hat{\bdelta} - \bdelta)\|_2^2\\
        &\hspace{1cm}+\|(\bI-\pw)\bY\|_2^2\\
        &= \|\bX(\bbeta - \tilde{\bbeta})\|_2^2+\|\tilde{\bZ}(\hat{\bdelta} - \bdelta)\|_2^2+\|(\bI-\pw)\bY\|_2^2\\
\end{split}
\end{equation*}
\end{proof}
\label{app:B_statistic_proof}

\subsection{Proof of Theorem \ref{thm:bayesian_lm_e_variable}}
We first prove the expression of $E_n$ using Wijsman's representation theorem, that is, manually computing the integrals in Lemma~\ref{lem:wijsman}.
We then examine the statistic through which the data enters the resulting expression.
Using a counterexample, we demonstrate that $\hat{\bxi}_n$ is not a maximal invariant.
We then prove that $\hat{\bxi}_n$ is, instead, an invariantly sufficient statistic.
\subsubsection*{Direct Calculation of $E_n$ from Wijsman's Representation Theorem}
Recall from Theorem~\ref{thm:bayesian_lm_e_variable} that the expression for $E_n$ is
    \label{sec:direct_calculation}
    \begin{equation}
            \begin{split}
                E_n \coloneqq& \frac{\int q(\bY_n| \bbeta, \sigma^2, \bdelta) \pi(\bbeta, \sigma^2) \pi(\bdelta|\sigma^2) d\bbeta d\sigma^2 d\bdelta}{\int q(\bY_n| \bbeta, \sigma^2, \bdelta) \pi(\bbeta, \sigma^2) d\bbeta d\sigma^2}\\
                =&\sqrt{\frac{\det(\bPhi)}{\det(\bPhi + \tilde{\bZ}_n'\tilde{\bZ}_n)}} \frac{\left(1+\frac{\hat{\bxi}_n'(\tilde{\bZ}_n'\tilde{\bZ}_n - \tilde{\bZ}_n'\tilde{\bZ}_n(\bPhi + \tilde{\bZ}_n'\tilde{\bZ}_n)^{-1}\tilde{\bZ}_n'\tilde{\bZ}_n)\hat{\bxi}_n}{\nu_n}\right)^{-\frac{\nu_n + d}{2}}}{\left(1+\frac{\hat{\bxi}_n'\tilde{\bZ}_n'\tilde{\bZ}_n\hat{\bxi}_n}{\nu_n}\right)^{-\frac{\nu_n + d}{2}}}.
            \end{split}
        \end{equation}
\begin{proof}
We prove this manually by computing the integrals in the integral representation.

To begin we compute the marginal likelihood of $\bY$ under the alternative $Q$ in the numerator.
\begin{equation}
    q(\bY) = \int \int \int q(\bY|\bbeta, \bdelta, \sigma^2) \pi(\bdelta|\sigma^2) \pi(\bbeta, \sigma^2) d\bbeta d\bdelta d\sigma^2.
\end{equation}
We use lemma \ref{lem:Y_decomposition} to decompose the likelihood throughout.

\begin{equation}
\begin{split}
    q(\bY|\bdelta, \sigma^2) =& \int q(\bY|\bbeta, \bdelta, \sigma^2) \pi(\bbeta) d\beta \\
    =&\left(\frac{1}{2\pi\sigma^2}\right)^{\frac{n}{2}}e^{-\frac{1}{2\sigma^2} \left(\|\tilde{\bZ}(\hat{\bdelta} - \bdelta)\|_2^2+\|(\bI-\pw)\bY\|_2^2 \right)}\int e^{-\frac{1}{2\sigma^2}\|\bX(\bbeta - \tilde{\bbeta}(\bY,\bdelta))\|_2}d\bbeta\\
    =&\left(\frac{1}{2\pi\sigma^2}\right)^{\frac{n-p}{2}}\left(\frac{1}{\det(\bX'\bX)}\right)^{\frac{1}{2}}e^{-\frac{1}{2\sigma^2}\|\tilde{\bZ}(\hat{\bdelta} - \bdelta)\|_2^2}e^{-\frac{1}{2\sigma^2}\|(\bI-\pw)\bY\|_2^2},
\end{split}
\end{equation}
where the last line follows from recognizing the integrand as the kernel of a multivariate Gaussian in $\bbeta$ with precision matrix $\bX'\bX/\sigma^2$.
To perform the marginalization over the prior on $\bdelta$ note that by completing the square
\begin{equation*}
    \begin{split}
            \|\tilde{\bZ}(\hat{\bdelta}-\bdelta)\|_2^2 + \bdelta'\bPhi\bdelta =& (\bdelta - \tilde{\bdelta})'(\bPhi + \tilde{\bZ}'\tilde{\bZ})(\bdelta - \tilde{\bdelta}) + \hat{\bdelta}'(\tilde{\bZ}'\tilde{\bZ} - \tilde{\bZ}'\tilde{\bZ}(\bPhi + \tilde{\bZ}'\tilde{\bZ})^{-1}\tilde{\bZ}'\tilde{\bZ})\hat{\bdelta}. 
    \end{split}
    \end{equation*}
\bigskip
where $\tilde{\bdelta} = (\bPhi + \tilde{\bZ}'\tilde{\bZ})^{-1}\tilde{\bZ}'\tilde{\bZ}\hat{\delta}$ is the posterior mean and $(\bPhi + \tilde{\bZ}'\tilde{\bZ})/\sigma^2$ the posterior precision.
Hence
\begin{equation*}
\begin{split}
       q(\bY|\sigma^2) =& \int q(\bY|\bdelta, \sigma^2) p(\bdelta|\sigma^2) d\bdelta\\ 
       =&\left(\frac{1}{2\pi\sigma^2}\right)^{\frac{n-p}{2}}\left(\frac{1}{\det(\bX'\bX)}\right)^{\frac{1}{2}}e^{-\frac{1}{2\sigma^2}\|(\bI-\pw)\bY\|_2^2}e^{-\frac{1}{2\sigma^2}\hat{\bdelta}'(\tilde{\bZ}'\tilde{\bZ} - \tilde{\bZ}'\tilde{\bZ}(\bPhi + \tilde{\bZ}'\tilde{\bZ})^{-1}\tilde{\bZ}'\tilde{\bZ})\hat{\bdelta}}\\
       &\left(\frac{1}{2\pi\sigma^2}\right)^{\frac{d}{2}}\det(\bPhi)^{\frac{1}{2}}\int e^{-\frac{1}{2\sigma^2}(\bdelta - \tilde{\bdelta})'(\bPhi + \tilde{\bZ}'\tilde{\bZ})(\bdelta - \tilde{\bdelta})}d\bdelta\\
    =&\left(\frac{1}{2\pi\sigma^2}\right)^{\frac{n-p}{2}}\left(\frac{1}{\det(\bX'\bX)}\right)^{\frac{1}{2}} \left(\frac{\det(\bPhi)}{\det(\bPhi + \tilde{\bZ}'\tilde{\bZ})}\right)^{\frac{1}{2}}\\
    &e^{-\frac{1}{2\sigma^2}\|(\bI-\pw)\bY\|_2^2}e^{-\frac{1}{2\sigma^2}\hat{\bdelta}'(\tilde{\bZ}'\tilde{\bZ} - \tilde{\bZ}'\tilde{\bZ}(\bPhi + \tilde{\bZ}'\tilde{\bZ})^{-1}\tilde{\bZ}'\tilde{\bZ})\hat{\bdelta}}\\
\end{split}
\end{equation*}
where the last step is achieved by recognizing the kernel of a $d$-dimensional multivariate Gaussian density in $\bdelta$ with precision $\bPhi+\tilde{\bZ}'\tilde{\bZ}$.
The marginalization over $\sigma^2$ is then
\begin{equation*}
\begin{split}
q(\bY) &= \int q(\bY|\sigma^2) \pi(\sigma^2)d\sigma^2\\
&= \left(\frac{1}{2\pi}\right)^{\frac{n-p}{2}}\left(\frac{1}{\det(\bX'\bX)}\right)^{\frac{1}{2}}\frac{\det(\bPhi)^{\frac{1}{2}}}{\det(\bPhi + \tilde{\bZ}'\tilde{\bZ})^{\frac{1}{2}}}\\
        &\int \left(\frac{1}{\sigma^2}\right)^{\frac{n-p}{2}+1}e^{-\frac{1}{2\sigma^2}\left(\|(\bI-\pw)\bY\|_2^2+\hat{\bdelta}'(\tilde{\bZ}'\tilde{\bZ} - \tilde{\bZ}'\tilde{\bZ}(\bPhi + \tilde{\bZ}'\tilde{\bZ})^{-1}\tilde{\bZ}'\tilde{\bZ})\hat{\bdelta}\right)}d\sigma^2\\
&= \left(\frac{1}{2\pi}\right)^{\frac{n-p}{2}}\left(\frac{1}{\det(\bX'\bX)}\right)^{\frac{1}{2}}\frac{\det(\bPhi)^{\frac{1}{2}}}{\det(\bPhi + \tilde{\bZ}'\tilde{\bZ})^{\frac{1}{2}}}\\
        &\Gamma\left(\frac{n-p}{2}\right)\left(\frac{\|(\bI-\pw)\bY\|_2^2}{2}+\frac{\hat{\bdelta}'(\tilde{\bZ}'\tilde{\bZ} - \tilde{\bZ}'\tilde{\bZ}(\bPhi + \tilde{\bZ}'\tilde{\bZ})^{-1}\tilde{\bZ}'\tilde{\bZ})\hat{\bdelta}}{2}\right)^{-\frac{n-p}{2}}\\
\end{split}
\end{equation*}
where the last line follows from recognizing the kernel of an Inverse Gamma. 
Rearranging terms yields 
\begin{equation*}
\begin{split}
        q(\bY)=& \left(\frac{1}{2\pi}\right)^{\frac{n-p}{2}}\left(\frac{1}{\det(\bX'\bX)}\right)^{\frac{1}{2}}\frac{\det(\bPhi)^{\frac{1}{2}}}{\det(\bPhi + \tilde{\bZ}'\tilde{\bZ})^{\frac{1}{2}}}\Gamma\left(\frac{n-p}{2}\right)\left(\frac{s^2(\bY)}{2}\right)^{-\frac{n-p}{2}}\\
        &\left(n-p-d\right)^{-\frac{n-p}{2}} \left(1+\frac{\hat{\bdelta}(\bY)'(\tilde{\bZ}'\tilde{\bZ} - \tilde{\bZ}'\tilde{\bZ}(\bPhi + \tilde{\bZ}'\tilde{\bZ})^{-1}\tilde{\bZ}'\tilde{\bZ})\hat{\bdelta}(\bY)}{s^2(\bY)(n-p-d)}\right)^{-\frac{n-p}{2}}
\end{split}
\end{equation*}
The steps to derive the marginal likelihood of $\bY$ under $P$ in the denominator proceed similarly.
\begin{equation*}
\begin{split}
     p(\bY|\sigma^2) =& \left(\frac{1}{2\pi\sigma^2}\right)^{\frac{n-p}{2}}\left(\frac{1}{\det(\bX'\bX)}\right)^{\frac{1}{2}}e^{-\frac{1}{2\sigma^2}\hat{\bdelta}'\tilde{\bZ}'\tilde{\bZ}\hat{\bdelta}}e^{-\frac{1}{2\sigma^2}\|(\bI-\pw)\bY\|_2^2}   \\
     =& \left(\frac{1}{2\pi\sigma^2}\right)^{\frac{n-p}{2}}\left(\frac{1}{\det(\bX'\bX)}\right)^{\frac{1}{2}}e^{-\frac{s^2(\bY)(n-p-d)}{2\sigma^2}\left(1+\frac{\hat{\bdelta}(\bY)'\tilde{\bZ}'\tilde{\bZ}\hat{\bdelta}(\bY)}{s^2(\bY)(n-p-d)}\right)}   
\end{split}
\end{equation*}
\begin{equation*}
\begin{split}
        p(\bY)=& \left(\frac{1}{2\pi}\right)^{\frac{n-p}{2}}\left(\frac{1}{\det(\bX'\bX)}\right)^{\frac{1}{2}}\Gamma\left(\frac{n-p}{2}\right)\left(\frac{s^2(\bY)}{2}\right)^{-\frac{n-p}{2}}\\
        &\left(n-p-d\right)^{-\frac{n-p}{2}} \left(1+\frac{\hat{\bdelta}(\bY)'\tilde{\bZ}'\tilde{\bZ}\hat{\bdelta}(\bY)}{s^2(\bY)(n-p-d)}\right)^{-\frac{n-p}{2}}
\end{split}
\end{equation*}

Therefore the final likelihood ratio is 
\begin{equation*}
\begin{split}
       E_n = \frac{q(\bY)}{p(\bY)} = \frac{\det(\bPhi)^{\frac{1}{2}}}{\det(\bPhi + \tilde{\bZ}'\tilde{\bZ})^{\frac{1}{2}}}\frac{\left(1+\frac{\hat{\bdelta}(\bY)'(\tilde{\bZ}'\tilde{\bZ} - \tilde{\bZ}'\tilde{\bZ}(\bPhi + \tilde{\bZ}'\tilde{\bZ})^{-1}\tilde{\bZ}'\tilde{\bZ})\hat{\bdelta}(\bY)}{s^2(\bY)(n-p-d)}\right)^{-\frac{n-p}{2}}}{\left(1+\frac{\hat{\bdelta}(\bY)'\tilde{\bZ}'\tilde{\bZ}\hat{\bdelta}(\bY)}{s^2(\bY)(n-p-d)}\right)^{-\frac{n-p}{2}}}
\end{split}
\end{equation*}
\end{proof}
\subsubsection*{Counterexample to Maximal Invariance of $\hat{\bxi}_n$}
\begin{proof}
Consider the following counterexample. For any $\bY_n$,
write $\bY^1_n = \pwn\bY_n + (\bI_n - \pwn)\bY_n$ and let $\bY^2_n = c\pwn\bY_n + c\bR(\bI_n - \pwn)\bY_n$ for an orthogonal matrix $\bR$.
The statistic $\hat{\bxi}_n(\bY_n^1) = \hat{\bxi}_n(\bY_n^2)$, but there is no $g \in G$ such that $\bY_n^2 = g\bY_n^1$. 
We have shown that the statistic $\hat{\bxi}_n$ takes the same value, even though $\bY_n^2$ and $\bY_n^1$ do not belong to the same orbit.
Hence $\hat{\bxi}_n$ cannot be a maximal invariant under the group action on $\mathbb{R}^n$ as it takes the same value
on two distinct orbits. 
\end{proof}
\subsubsection*{Invariant Sufficiency of $\hat{\bxi}_n$}
\label{sec:invariant_sufficiency}
\begin{proof}
The summary statistics for the linear model are given by $$S(\bY_n) = (\hat{\bbeta}_n(\bY_n), \hat{\bdelta}_n(\bY_n), s^2(\bY_n)) \in \mathcal{S} = \mathbb{R}^p \times \mathbb{R}^d\times \mathbb{R}_{\geq 0} $$. 
The group action $g: \bY_n \rightarrow c\bY_n + \bX_n\balpha$ on $\mathbb{R}^n$ induces a group action on $\mathcal{S}$ as
$$g: (\hat{\bbeta}_n, \hat{\bdelta}_n, s^2_n) \rightarrow (c\hat{\bbeta}_n + \balpha, c\hat{\bdelta}_n, c^2 s^2_n).$$
Note our statistic $\hat{\bxi}_n$ can be written as functions of summary statistics $\hat{\bxi}_n(\hat{\bdelta}_n, s^2_n) = \hat{\bdelta}_n / \sqrt{s^2_n}$, which is clearly invariant under the induced group.
To see that this is maximal invariant, consider two elements of the summary statistic space $(\bbeta_1, \bdelta_1, s^2_1)$ and $(\bbeta_2, \bdelta_2, s^2_2)$ which have the same value of $\hat{\bxi}_n$
$$
\hat{\bxi}_n(\bdelta_1, s^2_1) = \hat{\bxi}_n(\bdelta_2, s^2_2)
$$
Then $g (\bbeta_1, \bdelta_1, s^2_1) = (\bbeta_2, \bdelta_2, s^2_2)$ using $\balpha = \bbeta_2 - c\bbeta_1$ and $c = \sqrt{s^2_2/s^2_1}$, which implies both elements belong to the same orbit.
Hence $\hat{\bxi}_n$ is a maximal invariant for the action of the group on $\mathcal{S}$.
\end{proof}
\subsubsection*{$E_n$ as a likelihood ratio of $\hat{\bxi}_n$}
\begin{proof}
From Theorem~\ref{thm:stein} it follows that $E_n$ is equal to the likelihood ratio of the invariantly sufficient statistic $\hat{\bxi}_n$. 
To calculate the density of $\hat{\bxi}_n$ under $Q$ and $P$ first observe that $\hat{\bdelta}_n |\bdelta, \sigma^2 \sim N(\bdelta, \sigma^2 (\tilde{\bZ}_n'\tilde{\bZ}_n)^{-1})$. 
The Bayes marginal resulting from the prior $\bdelta |\sigma^2 \sim N(\bzero, \sigma^2\bPhi^{-1})$ is $\hat{\bdelta}_n | \sigma^2 \sim N(\bzero, \sigma^2 (\bPhi^{-1} + (\tilde{\bZ}_n'\tilde{\bZ}_n)^{-1}))$. Together this implies that $\hat{\bxi}_n \sim t_{\nu_n}(\bzero, (\bPhi^{-1} + (\tilde{\bZ}_n'\tilde{\bZ}_n)^{-1}))$ where $\nu_n = n-p-d$ under the alternative $Q$.
Similarly, $\hat{\bxi}_n \sim t_{\nu_n}(\bzero, (\tilde{\bZ}_n'\tilde{\bZ}_n)^{-1})$ under $P$.
\end{proof}
\begin{lemma}
\begin{equation*}
\begin{split}
    &(\bW'\bW)^{-1} = 
   \begin{pmatrix}
    \bX'\bX & \bX'\bZ \\
    \bZ'\bX & \bZ'\bZ \\
\end{pmatrix}^{-1}\\
&=
\begin{pmatrix}
    (\bX'\bX)^{-1} + (\bX'\bX)^{-1} \bX'\bZ (\tilde{\bZ}'\tilde{\bZ})^{-1} \bZ'\bX (\bX'\bX)^{-1} & -(\bX'\bX)^{-1} \bX'\bZ \tilde{\bZ}'\tilde{\bZ}^{-1} \\
    -\tilde{\bZ}'\tilde{\bZ}^{-1} \bZ'\bX (\bX'\bX)^{-1} & \tilde{\bZ}'\tilde{\bZ}^{-1} \\
\end{pmatrix}.
\end{split}
\end{equation*}
    \label{lem:block_inverse}
\end{lemma}

\begin{proof}
This follows immediately from block matrix inverse results and the observation that $\tilde{\bZ}'\tilde{\bZ} = (\bZ'\bZ - \bZ'\bX (\bX'\bX)^{-1}\bX'\bZ) = \bZ'(\bI- \bX (\bX'\bX)^{-1}\bX')\bZ = \bZ'(\bI-\px)\bZ = \bZ'(\pw-\px)\bZ$.
\end{proof}

\begin{lemma}
    \label{lem:strong_ols}
    Suppose $\frac{1}{n}\bW_n'\bW_n$ converges almost surely to a positive definite matrix $\bOmega_{\bW}$ and $\mathbb{E}[\varepsilon_i \bw_i ]=0$, then $\hat{\bgamma}_n \overset{a.s.}{\rightarrow} \bgamma$ and $s^2_n  \overset{a.s.}{\rightarrow} \sigma^2$ and $\frac{1}{n}\tilde{\bZ}_n'\tilde{\bZ}_n \overset{a.s.}{\rightarrow} \bOmega_{\tilde{\bZ}} = \bOmega_{\bZ} - \bOmega_{\bZ\bX}\bOmega_{\bX}^{-1}\bOmega_{\bX \bZ}$, where 
\begin{equation*}
        \bOmega_\bW = \begin{pmatrix}
        \bOmega_{\bX} & \bOmega_{\bX\bZ} \\ 
        \bOmega_{\bZ\bX} & \bOmega_{\bZ} \\ 
    \end{pmatrix}  
\end{equation*}
\end{lemma}
\begin{proof}
$\frac{1}{n}\bW_n'\bvarepsilon_n = \frac{1}{n}\sum_{i=1}^n \bw_i'\varepsilon_i \rightarrow \bzero$ almost surely by the strong law because $\mathbb{E}[\bw_i \varepsilon_i]=\bzero$.
By positive definiteness of $\bOmega_\bw$ and the continuous mapping theorem, $(\bW_n'\bW_n/n)^{-1} \rightarrow \bOmega_\bw^{-1}$ almost surely. By definition of the OLS estimator $$\hat{\bgamma}_n^{ols} = \bgamma + (\bW_n'\bW_n)^{-1}n n^{-1}\bW_n'\bvarepsilon \overset{a.s.}{\rightarrow} \bgamma + \bOmega_\bw^{-1}\bzero = \bgamma.$$ For strong consistency of $\sigma^2$, $\hat{\sigma}^2_n = (\bvarepsilon_n'\bvarepsilon_n - \bvarepsilon_n'\bW_n(\bW_n'\bW_n)^{-1}\bW_n'\bvarepsilon_n)/\nu_n$. Considering the last term, $$\bvarepsilon_n'\bW_n(\bW_n'\bW_n)^{-1}\bW_n'\bvarepsilon_n/\nu_n = \bvarepsilon_n'\bW_n n^{-1} n(\bW_n'\bW_n)^{-1} n^{-1}\bW_n'\bvarepsilon_n (n/\nu_n) \overset{a.s.}{\rightarrow} \bzero$$ while $\bvarepsilon_n'\bvarepsilon_n/\nu_n \overset{a.s.}{\rightarrow} \sigma^2 $ by the strong law ($\mathbb{E}[\varepsilon_i^2] = \sigma^2$). In addition
\begin{equation*}
    \frac{1}{n}\bW_n'\bW_n = \frac{1}{n}
   \begin{pmatrix}
    \bX_n'\bX_n & \bX_n'\bZ_n \\
    \bZ_n'\bX_n & \bZ_n'\bZ_n \\
\end{pmatrix}\overset{a.s.}{\rightarrow}
\begin{pmatrix}
    \bOmega_{\bX} & \bOmega_{\bX\bZ} \\ 
    \bOmega_{\bZ\bX} & \bOmega_{\bZ} \\ 
\end{pmatrix} = \bOmega_\bw.
\end{equation*}
By the continuous mapping theorem $\left((1/n)\bW_n'\bW_n\right)^{-1} \overset{a.s.}{\rightarrow}\bOmega_\bw^{-1}$, therefore $\frac{1}{n}\tilde{\bZ}_n'\tilde{\bZ}_n \overset{a.s.}{\rightarrow} \bOmega_{\bZ} - \bOmega_{\bZ\bX}\bOmega_{\bX}^{-1}\bOmega_{\bX\bZ}$
by lemma \ref{lem:block_inverse}
\end{proof}

\begin{lemma}
    \label{lem:neumann}
For a matrix $\bT$ with $\|\bT\| < 1$
$$\left(\bI - \bT\right)^{-1} = \sum_{k=0}^\infty \bT^k$$
\end{lemma}
\begin{proof}
This is simply the Neumann series, see \citet[Theorem 18.2.16]{harville2008matrix}
\end{proof}
\begin{lemma}
\label{lem:logdet}
For a matrix $\bT$ with $\|\bT\| < 1$
$$\log \det(\bI - \bT) = -\sum_{k=1}^\infty \frac{1}{k}\Tr(\bT^k)$$
\end{lemma}
\subsection{Proof of Theorem \ref{thm:almost_sure_limit}}
\begin{proof}
By assumption $(1/n)\tilde{\bZ}_n'\tilde{\bZ}_n\overset{a.s.}{\rightarrow}\bOmega_{\tilde{\bZ}}$.
Additionally $\mathbb{E}[\varepsilon_i| \bw_i ]=0$ is satisfied by the assumptions of the linear model, therefore
 $\bxi_n \overset{a.s.}{\rightarrow} \bxi$ by lemma \ref{lem:strong_ols}. First observe that
$$\frac{1}{2n}\log\det\bPhi \rightarrow 0.$$
Similarly
\begin{equation*}
    \begin{split}
    -\frac{1}{2n}\log\det(\bPhi + \tilde{\bZ}_n'\tilde{\bZ}_n) &= -\frac{1}{2n}\log\det n(\bPhi / n + \tilde{\bZ}_n'\tilde{\bZ}_n / n)\\
    &=-\frac{1}{2n} \log n^d - \frac{1}{2n}\log\det (\bPhi / n + \tilde{\bZ}_n'\tilde{\bZ}_n / n)\\
    &\overset{a.s.}{\rightarrow} 0.
    \end{split}
\end{equation*}
To handle the second term of $\log E_n$ write
$$\log \frac{\left(1+\frac{\hat{\bxi}_n'(\tilde{\bZ}_n'\tilde{\bZ}_n - \tilde{\bZ}_n'\tilde{\bZ}_n(\bPhi + \tilde{\bZ}_n'\tilde{\bZ}_n)^{-1}\tilde{\bZ}_n'\tilde{\bZ}_n)\hat{\bxi}_n}{\nu_n}\right)}{\left(1+\frac{\hat{\bxi}_n'\tilde{\bZ}_n'\tilde{\bZ}_n\hat{\bxi}_n}{\nu_n}\right)} = \log \frac{1+b_n}{1+a_n}$$
where $b_n = \bxi_n ' (\bPhi^{-1} + (\tilde{\bZ}_n'\tilde{\bZ}_n)^{-1})^{-1} \bxi_n / \nu_n$ and $a_n = \bxi_n' \tilde{\bZ}_n'\tilde{\bZ}_n \bxi_n / \nu_n$. Writing $$(\bPhi^{-1} + (\tilde{\bZ}_n'\tilde{\bZ}_n)^{-1})^{-1} = \bPhi(\bI -\bT_n)^{-1},$$ where $\bT_n = - (\tilde{\bZ}_n'\tilde{\bZ}_n)^{-1}\bPhi$, we have for
large $n$ by lemmas \ref{lem:strong_ols} and \ref{lem:neumann} 
\begin{equation*}
    \begin{split}
        b_n &= \bxi'\bPhi\bxi / \nu_n + O_{a.s.}(n^{-2}),\\
        a_n &= \bxi' \bOmega_{\tilde{\bZ}} \bxi + o_{a.s.}(1),
    \end{split}
\end{equation*}
and therefore
$$- \frac{1}{n} \frac{\nu_n + d}{2}\log \frac{1+b_n}{1+a_n} = - \frac{1}{2} \log \frac{1}{1 + \bxi'\bOmega_{\tilde{\bZ}}\bxi} + o_{a.s.}(1)$$.
\end{proof}

\subsection{Proof of Corollary \ref{cor:contrasts}}
Consider the regression $\bY_n \sim N(\bX_n\bbeta + \bZ_n\bdelta, \sigma^2 \bI_n)$ with the null hypothesis $H_0: \bE\bdelta = \bzero$ for $\bE \in \mathbb{R}^{k\times d}$ with $k\leq d$.
Let $P_E = \bE'(\bE\bE')^{-1}\bE$ be the projection onto the \textit{row} space (orthogonal complement of the null space) of $\bE$.
Let $\bN$ be an orthonormal basis that spans the null space of $\bE$.
Let $P_N = \bN\bN'$ denote the projection onto the nullspace of $\bE$, then $P_N = \bI - P_E$.
Decompose $\bdelta = \bdelta_1 + \bdelta_2$ where $\bdelta_1 = P_N\bdelta$ and $\bdelta_2 = P_E \bdelta$. 
Therefore we can write $\bdelta_1 = \bN\bv_1$ and $\bdelta_2 = \bE'\bv_2$ for some vectors $\bv_1 \in \mathbb{R}^{d-k}$ and $\bv_2 = (\bE\bE')^{-1}\bE\bdelta\in \mathbb{R}^k$.
Hence the outcome model can be written as 
$$\bY_n \sim N(\bX_n\bbeta + \bZ_n\bN\bv_1 + \bZ_n\bE'\bv_2, \sigma^2 \bI_n),$$ which can be redefined as
$$\bY_n \sim N(\bX_n^{\dagger}\bbeta^\dagger + \bZ_n^\dagger\bdelta^\dagger, \sigma^2 \bI_n)$$
where $\bX_n^{\dagger} = [\bX_n, \bZ_n\bN]$, $\bZ_n^\dagger = \bZ_n\bE'$, $\bbeta^\dagger = \begin{pmatrix} \bbeta \\ \bv_1 \end{pmatrix}\in \mathbb{R}^{p^\dagger}$, where $p^\dagger = p + d - k$, and $\bdelta^\dagger = \bv_2 \in \mathbb{R}^{d^\dagger}$ where $d^\dagger = k$.
Additionally the implied prior on $\bdelta^\dagger \sim N(\bzero, \sigma^2 (\bPhi^\dagger)^{-1})$ where $(\bPhi^\dagger)^{-1} =  (\bE\bE')^{-1}\bE\bPhi^{-1}\bE' (\bE\bE')^{-1}$

This is exactly the form of the linear model considered in theorem \ref{thm:bayesian_lm_e_variable}, and so $\bE\bdelta=0$ can be tested by using the $e$-variable for testing $\bdelta^\dagger = \bzero$.

\begin{equation}
    \label{eq:E_n_delta_version}
    E_n=\sqrt{\frac{\det(\bPhi)}{\det(\bPhi + \tilde{\bZ}_n^{\dagger '}\tilde{\bZ}_n^\dagger)}} \frac{\left(1+\frac{\hat{\bdelta}_n^{\dagger '}(\tilde{\bZ}_n^{\dagger '}\tilde{\bZ}_n^\dagger - \tilde{\bZ}_n^{\dagger '}\tilde{\bZ}_n^\dagger(\bPhi + \tilde{\bZ}_n^{\dagger '}\tilde{\bZ}_n^\dagger)^{-1}\tilde{\bZ}_n^{\dagger '}\tilde{\bZ}_n^\dagger)\hat{\bdelta}_n}{s^2_n\nu_n}\right)^{-\frac{\nu_n + k}{2}}}{\left(1+\frac{\hat{\bdelta}_n^{\dagger '}\tilde{\bZ}_n^{\dagger '}\tilde{\bZ}_n^\dagger\hat{\bdelta}_n^\dagger}{s^2_n\nu_n}\right)^{-\frac{\nu_n + k}{2}}},
\end{equation}
where we have substituted $\hat{\bxi}_n = \hat{\bdelta}_n / \sqrt{s^2_n}$.
As before, testing a specific point null $\bE\bdelta = \be_0$ can be achieved by replacing $\bY_n$ by $\bY_n - \bZ_n \be_0$.
However, we do not need to define $\bX_n^\dagger$ or $\bZ_n^\dagger$ explicitly, as we rewrite expressions in terms of the original parameterization.
First, observe that
$s^2_n = \frac{\bY_n'(\bI_n - \bP_{\bW^\dagger})\bY_n }{n-p^\dagger - d^\dagger}= \frac{\bY_n'(\bI-\pw)\bY_n}{n-p-d}$
Secondly, $\tilde{\bZ}_n^\dagger$ can be written in terms of $\tilde{\bZ}_n$ and $\bE$ as
\begin{equation*}
    \begin{split}
\tilde{\bZ_n^\dagger} &= (\bI_n - \bP_{\bX_n^\dagger})\bZ_n^\dagger\\
&= (\bI_n - \bP_{\bX_n^\dagger})\bZ_n\bE'\\
 &= (\bI_n - \bP_{\bX_n})\bZ_n\bE'\\
 &= \tilde{\bZ_n}\bE'\\
    \end{split}
\end{equation*}
where the second line follows because projecting $\bZ_n\bE'$ onto $C(\bX_n^\dagger)$ is equivalent to projecting onto $C(\bX_n)$ as the components $\bZ_n\bE'$ (in the row space of $\bE$) and $\bZ_n\bN$ (in the null space of $\bE$) are orthogonal.
Therefore by the Frisch-Waugh-Lovell theorem
$$\hat{\bdelta}_n^\dagger = (\tilde{\bZ_n^\dagger}'\tilde{\bZ_n^\dagger})^{-1}\tilde{\bZ_n^\dagger}'\bY_n = (\bE\tilde{\bZ_n}'\tilde{\bZ}_n\bE')^{-1}\bE\tilde{\bZ_n}' \bY_n = (\bE\tilde{\bZ_n}'\tilde{\bZ}_n\bE')^{-1}\bE\tilde{\bZ}_n'\tilde{\bZ}_n\hat{\bdelta}_n, $$
and so 
\begin{equation*}
    \begin{split}
        \hat{\bdelta}_n^{\dagger '}\tilde{\bZ_n^\dagger}'\tilde{\bZ_n^\dagger}\hat{\bdelta}_n^\dagger =& \hat{\bdelta}_n\tilde{\bZ}_n'\tilde{\bZ}_n\bE'(\bE\tilde{\bZ_n}'\tilde{\bZ}_n\bE')^{-1}\bE\tilde{\bZ}_n'\tilde{\bZ}_n\hat{\bdelta}_n\\
         =& \hat{\bdelta}_n\bE'(\bE(\tilde{\bZ_n}'\tilde{\bZ}_n)^{-1}\bE')^{-1}\bE\hat{\bdelta}_n.\\
    \end{split}
\end{equation*}

\subsection{Proof of Theorem \ref{thm:optimal_phi_frequentist}}
The expected value of the log $e$-variable for testing a single coefficient can be written as
\begin{equation*}
    \begin{split}
        \mathbb{E}[\log E_n] =& \frac{1}{2}\log \phi - \frac{1}{2}\log(\phi + \|\tilde{\bZ}_n\|^2)\\
        &+ \frac{\nu_n + 1}{2}\mathbb{E}\left[\log\left(1 + \frac{\|\tilde{\bZ}_n\|_2^2\hat{\xi}_n^2}{\nu_n}\right)\right]\\
        &- \frac{\nu_n + 1}{2}\mathbb{E}\left[\log\left(1 + \frac{\phi}{\phi + \|\tilde{\bZ}_n\|_2^2}\frac{\|\tilde{\bZ}_n\|_2^2\hat{\xi}_n^2}{\nu_n}\right)\right]\\
    \end{split}
\end{equation*}
Defining $X_n =  \frac{\phi}{\phi + \|\tilde{\bZ}_n\|_2^2}\frac{\|\tilde{\bZ}_n\|_2^2\hat{\xi}_n^2}{\nu_n}$, the last expectation can be expressed as
\begin{equation*}
    \begin{split}
        \mathbb{E}\left[\log(1+X_n)\right] &= \mathbb{E}\left[\log(1+X_n)1_{\{X_n < 1\}}\right] + \mathbb{E}\left[\log(1+X_n)1_{\{X_n \geq 1\}}\right]\\
        &= \mathbb{E}\left[\sum_{k=1}^\infty (-1)^{k+1}\frac{X_n^k}{k}1_{\{X_n < 1\}} \right] + \mathbb{E}\left[\log(1+X_n)1_{\{X_n \geq 1\}}\right]\\
        &= \mathbb{E}[X_n 1_{\{X_n < 1\}}] + \mathbb{E}\left[\sum_{k=2}^\infty (-1)^{k+1}\frac{X_n^k}{k}1_{\{X_n < 1\}} \right] + \mathbb{E}\left[\log(1+X_n)1_{\{X_n \geq 1\}}\right]\\
    \end{split}
\end{equation*}
The second and third terms are sandwiched as
\begin{equation*}
    \begin{split}
        -\frac{1}{2}\mathbb{E}[X_n^2] &\leq \mathbb{E}\left[\sum_{k=2}^\infty (-1)^{k+1}\frac{X_n^k}{k}1_{\{X_n < 1\}} \right] \leq 0 \\
        0 & \leq \mathbb{E}\left[\log(1+X_n)1_{\{X_n \geq 1\}}\right] \leq \mathbb{E}[X_n^2]\\
    \end{split}
\end{equation*}
where the last line follows from $\log(1+X_n)1_{\{X_n \geq 1\}} \leq X_n 1_{\{X_n \geq 1\}} \leq X_n^2$. 
Recall $\|\tilde{\bZ}_n\|_2^2\hat{\xi}_n^2 \sim F(1, \nu_n, \|\tilde{\bZ}_n\|_2^2\xi^{\star 2})$.
Let $\lambda_F =  \|\tilde{\bZ}_n\|_2^2\xi^{\star 2}$.
Then, the expectation of $\log(1+X_n)$ can be written as
\begin{equation*}
    \begin{split}
        \mathbb{E}\left[\log(1+X_n)\right] &= \mathbb{E}[X_n 1_{\{X_n < 1\}}] + O(n^{-2}),\\
    \end{split}
\end{equation*}
as $\mathbb{E}[X_n^2] = \left(\frac{\phi}{\phi + \|\tilde{\bZ}_n\|_2^2}\frac{1}{\nu_n}\right)^2 \mathbb{E}[(\|\tilde{\bZ}_n\|_2^2\hat{\xi}_n^2)^2]$ where $\mathbb{E}[(\|\tilde{\bZ}_n\|_2^2\hat{\xi}_n^2)^2]=\nu_n^2 (\lambda_F^2 + 6\lambda_F + 4)/ ((\nu_n - 2)(\nu_n-4))$.
Moreover because $E[X_n] = E[X_n 1_{\{X_n < 1\}}] + E[X_n 1_{\{X_n \geq 1\}}]$ and $E[X_n 1_{\{X_n \geq 1\}}] = O(n^{-2})$, we can write
$E[X_n 1_{\{X_n < 1\}}] = E[X_n] + O(n^{-2})$. Substituting this into the expression for $\mathbb{E}[\log E_n]$ and replacing $X_n$ with its definition gives
for the log $e$-variable becomes
\begin{equation*}
    \begin{split}
        \mathbb{E}[\log E_n] =& \frac{1}{2}\log \phi - \frac{1}{2}\log(\phi + \|\tilde{\bZ}_n\|^2)\\
        &+ \frac{\nu_n + 1}{2}\mathbb{E}\left[\log\left(1 + \frac{\|\tilde{\bZ}_n\|_2^2\hat{\xi}_n^2}{\nu_n}\right)\right]\\
        &- \frac{\nu_n + 1}{2}\frac{\phi}{\phi + \|\tilde{\bZ}_n\|_2^2}\frac{1}{\nu_n} \frac{\nu_n}{\nu_n - 2}\left(1+ \|\tilde{\bZ}_n\|_2^2\xi^{\star 2}  \right)+ O(n^{-2})\\
    \end{split}
\end{equation*}
The value $\phi^\star_n$ that maximizes the leading order term is
\begin{equation*}
    \phi^\star_n = \frac{1}{\frac{\nu_n+1}{\nu_n-2}\left(\frac{1}{\|\tilde{\bZ}_n\|_2^2} + \xi^{\star 2}\right) - \frac{1}{ \|\tilde{\bZ}_n\|_2^2}}
\end{equation*}
As $\nu_n \rightarrow \infty$ and $\|\tilde{\bZ}_n\|_2^2=\Omega(n)$, $\lim_{n \to \infty} \phi^\star_n = (\xi^\star)^{-2}$.

\subsection{Proof of Theorem \ref{thm:zellner_e_variable}}
Recall
\begin{equation*}
    \begin{split}
                E_n=&\sqrt{\frac{\det(g\bOmega_{\tilde{\bZ}})}{\det(g\bOmega_{\tilde{\bZ}} + \tilde{\bZ}_n'\tilde{\bZ}_n)}} \frac{\left(1+\frac{\hat{\bxi}_n'(\tilde{\bZ}_n'\tilde{\bZ}_n - \tilde{\bZ}_n'\tilde{\bZ}_n(g\bOmega_{\tilde{\bZ}} + \tilde{\bZ}_n'\tilde{\bZ}_n)^{-1}\tilde{\bZ}_n'\tilde{\bZ}_n)\hat{\bxi}_n}{\nu_n}\right)^{-\frac{\nu_n + d}{2}}}{\left(1+\frac{\hat{\bxi}_n'\tilde{\bZ}_n'\tilde{\bZ}_n\hat{\bxi}_n}{\nu_n}\right)^{-\frac{\nu_n + d}{2}}},\\
                G_n=&\left(\frac{g}{g+n}\right)^{\frac{d}{2}}\left(\frac{1 + \frac{g}{g+n} \frac{d}{\nu_n} F_n}{1+\frac{d}{\nu_n} F_n}\right)^{-\frac{\nu_n + d}{2}},
    \end{split}
\end{equation*}
where $F_n = \bxi_n'\tilde{\bZ}_n'\tilde{\bZ}_n\bxi_n / d$, $\nu_n = n-p-d$.
We would like to calculate the order of $| \log E_n - \log G_n|$.
By assumption $\frac{1}{n}\tilde{\bZ}_n'\tilde{\bZ}_n \overset{a.s.}{\rightarrow} \bOmega_{\tilde{\bZ}}$, which by lemma \ref{lem:strong_ols} implies $\bxi_n \rightarrow \bxi$ and $\frac{d}{\nu_n}F_n \rightarrow \bxi'\bOmega_{\tilde{\bZ}}\bxi$.
Let $\bR_n = \bOmega_{\tilde{\bZ}} - \frac{1}{n}\tilde{\bZ}_n'\tilde{\bZ}_n$, then by assumption $\bR_n = o_{a.s.}(1)$.
Consider the determinant terms
\begin{equation*}
    \frac{\det g\bOmega_{\tilde{\bZ}} }{\det(g\bOmega_{\tilde{\bZ}} + \tilde{\bZ}_n'\tilde{\bZ})} = \frac{\det g\bOmega_{\tilde{\bZ}} }{\det((g + n)\bOmega_{\tilde{\bZ}} - n \bR_n)} = \left(\frac{g}{g+n}\right)^d\frac{1}{\det \left(\bI - \frac{n}{g+n}\bOmega_{\tilde{\bZ}}^{-1}\bR_n\right)}.\\
\end{equation*}
The difference determinant can then be written via lemma \ref{lem:logdet} as 
\begin{equation*}
    \begin{split}
        \frac{1}{2}\log\left(\frac{\det g\bOmega_{\tilde{\bZ}} }{\det(g\bOmega_{\tilde{\bZ}} + \tilde{\bZ}_n'\tilde{\bZ})} \right)=& \frac{d}{2}\log\left(\frac{g}{g+n}\right) - \frac{1}{2}\log\det \left(\bI - \frac{n}{g+n}\bOmega_{\tilde{\bZ}}^{-1}\bR_n\right)\\
=&\frac{d}{2}\log\left(\frac{g}{g+n}\right) -\frac{1}{2}\sum_{k=1}^\infty \frac{1}{k}\Tr\left(\left(\frac{n}{g+n}\bOmega_{\tilde{\bZ}}^{-1}\bR_n\right)^k\right)\\
    \end{split}
\end{equation*}
Now consider the term
$$b_n = \frac{\bxi_n(\tilde{\bZ}_n'\tilde{\bZ}_n - \tilde{\bZ}_n'\tilde{\bZ}_n(g\bOmega_{\tilde{\bZ}} + \tilde{\bZ}_n'\tilde{\bZ}_n)^{-1}\tilde{\bZ}_n'\tilde{\bZ}_n)\bxi_n}{\nu_n}$$
First observe that
\begin{equation*}
    \begin{split}
        \tilde{\bZ}_n'\tilde{\bZ}_n(g\bOmega_{\tilde{\bZ}} + \tilde{\bZ}_n'\tilde{\bZ}_n)^{-1}\tilde{\bZ}_n'\tilde{\bZ}_n
        &=\tilde{\bZ}_n'\tilde{\bZ}_n(\frac{g+n}{n}\tilde{\bZ}_n'\tilde{\bZ}_n + g \bR_n)^{-1}\tilde{\bZ}_n'\tilde{\bZ}_n\\
        &=\frac{n}{g+n}(\bI_n - \bT_n)^{-1}\tilde{\bZ}_n'\tilde{\bZ}_n\\
        &=\frac{n}{g+n}\tilde{\bZ}_n'\tilde{\bZ}_n + \frac{n}{g+n} \sum_{k=1}^\infty \bT_n^k\tilde{\bZ}_n'\tilde{\bZ}_n\\
    \end{split}
\end{equation*}
where $\bT_n = -\frac{gn}{g+n}  \bR_n(\tilde{\bZ}_n'\tilde{\bZ}_n)^{-1}$, and so $b_n$ can be written as
\begin{equation*}
        b_n=\frac{g}{g+n}\frac{d}{\nu_n}F_n  + \frac{n}{\nu_n(g+n)}\sum_{k=1}^\infty \bxi_n ' \bT_n^k\tilde{\bZ}_n'\tilde{\bZ}_n\bxi_n.\\
\end{equation*}
It follows that 
\begin{equation*}
    \begin{split}
            \log(1+b_n) =& \log(1 + \frac{g}{g+n}\frac{d}{\nu_n}F_n)\\
            &+\log\left(1 + \frac{1}{1 + \frac{g}{g+n}\frac{d}{\nu_n}F_n}\frac{n}{\nu_n(g+n)}\sum_{k=1}^\infty \bxi_n ' \bT_n^k\tilde{\bZ}_n'\tilde{\bZ}_n\bxi_n\right)\\
        \end{split}
\end{equation*}
The term $(d / \nu_n) F_n \overset{a.s.}{\rightarrow} \bxi' \bOmega_{\tilde{\bZ}}\bxi$, and so the term preceding the sum behaves like $\sim n^{-1}$
\begin{equation*}
    \begin{split}
        -\frac{\nu_n + d}{2}\log(1+b_n) =& -\frac{\nu_n + d}{2}\log(1 + \frac{g}{g+n}\frac{d}{\nu_n}F_n)\\
        &- \frac{\nu_n + d}{2}\left(\frac{1}{1 + \frac{g}{g+n}\frac{d}{\nu_n}F_n}\frac{n}{\nu_n(g+n)}\right)\frac{gn}{g+n}\bxi_n'\bR_n\bxi_n\\
        &+ O_{a.s.}\left(\frac{\bR_n}{n}\right)
    \end{split}
\end{equation*}
Therefore
\begin{equation*}
    \begin{split}
        | \log E_n - \log G_n | &= \frac{1}{2}\frac{n}{g+n}\Tr\left(\bOmega_{\tilde{\bZ}}^{-1}\bR_n\right)\\
        &\hspace{1cm} - \frac{\nu_n + d}{2}\left(\frac{1}{1 + \frac{g}{g+n}\frac{d}{\nu_n}F_n}\frac{n}{\nu_n(g+n)}\right)\frac{gn}{g+n}\bxi_n'\bR_n\bxi_n\\
        &\hspace{1cm}+ O_{a.s.}\left(\frac{\bR_n}{n}\right)\\
        &= \frac{1}{2}\Tr\left(\bOmega_{\tilde{\bZ}}^{-1}\bR_n\right) - \frac{g}{2}\bxi_n ' \bR_n\bxi_n + O_{a.s.}\left(\frac{\bR_n}{n}\right)
    \end{split}
\end{equation*}

\subsection{Proof of Proposition \ref{prop:marginalization_consistency}}
We will show that the $(1,1)$-block of $(\tilde{\bZ}'\tilde{\bZ})^{-1}$ coincides with $(\bZ_1^{\star '}\bZ_1^\star)^{-1}$.
\begin{proof}
$\tilde{\bZ} = (\bI - \bP_{\bX})\bZ = (\bI - \bP_{\bX})[\bZ_1, \bZ_2] = [\tilde{\bZ}_1, \tilde{\bZ}_2]$.
From the block matrix inverse and Schur complement formulas
\begin{equation*}
    \begin{split}
        ((\tilde{\bZ}'\tilde{\bZ})^{-1})_{11} =& (\tilde{\bZ}_1'\tilde{\bZ}_1 - \tilde{\bZ}_1'\tilde{\bZ}_2(\tilde{\bZ}_2'\tilde{\bZ}_2)^{-1}\tilde{\bZ}_2'\tilde{\bZ}_1)^{-1},\\
        =& (\bZ_1'(\bI - \bP_{\bX} - \tilde{\bZ}_2(\tilde{\bZ}_2'\tilde{\bZ}_2)^{-1}\tilde{\bZ}_2')\bZ_1)^{-1},\\
        =& (\bZ_1'(\bI - \bP_{\bX} - \bP_{\tilde{\bZ}_2})\bZ_1)^{-1},\\
        =& (\bZ_1'(\bI - \bP_{[\bX, \bZ_2]})\bZ_1)^{-1},\\
        =& (\bZ_1'(\bI - \bP_{\bX^\star})\bZ_1)^{-1},\\
        =& (\bZ_1^{\star '}\bZ_1^\star)^{-1}.\\
    \end{split}
\end{equation*}
\end{proof}

\subsection{Comparisons with Fixed-$n$ Procedures}
\label{sec:fixed_n_comparison}
In this section we contrast our anytime-valid procedure against the analogous fixed-$n$ procedure under two different frequentist alternatives, and under the null hypothesis.
We simulate independent outcomes from the following regression model.
$$y_i = 1 + \bx_i'\bbeta + Z_i\delta + \varepsilon_i$$
where $Z_i = (T_i - \rho)$ with $T_i \sim \text{Bernoulli}(\rho)$ and $\rho = 0.5$, $\bbeta' = (0,1,2)$, $\bx_i \sim N(\bzero_3, \bSigma)$ with $\bSigma_{ij}=0.8^{|i-j|}$, and $\varepsilon_i \sim N(0,1)$.
Let $F_n = \hat{\xi}_n^2 \|\tilde{\bZ}_n\|^2_2 = (\hat{\delta}_n / s_n)^2 \|\tilde{\bZ}_n\|^2_2 \sim f(1,\nu_n, \|\tilde{\bZ}_n\|_2^2 \xi^2)$ be the $f$-statistic for testing the univariate $(d=1)$ null hypothesis $\delta = 0$.
The $\alpha$-level fixed-$n$ test rejects the null hypothesis when $F_n > f_{1,\nu_n}^{1-\alpha}$, where $f_{1,\nu_n}^{1-\alpha}$ is the $(1-\alpha)$-quantile of the $f(1,\nu_n, 0)$ distribution.
The sample size $n^\star(\xi_{MDE})$ required to achieve a power $\geq c$ for all $\xi$ with $|\xi| > \xi_{MDE}$ is determined by finding the smallest integer $n$ satisfying $\mathbb{P}[F_n^{\xi_{MDE}} > f_{1,\nu_n}^{1-\alpha}] \geq c$, where $F_n^{\xi_{MDE}} \sim f(1, \nu_n, \|\tilde{\bZ}_n\|_2^2 \xi_{MDE}^2)$.
In this simulation the regressors are random, not fixed, so we use the approximation $\|\tilde{\bZ}_n\|_2^2 \approx n \rho(1-\rho)$.
Note that being able to specify an MDE implies some prior knowledge about the expected effect size or, at least, the size of a practically relevant effect.
In order to provide a fair comparison, the same information should be input to the anytime-valid procedure, which is encoded by the mixture precision.

To ensure that the anytime-valid procedure is well calibrated against the fixed-$n$ procedure tuned to the MDE, we consider two strategies for choosing the hyperparameter $g$.
The first approach chooses $g$ to maximize the worst-case growth rate of the $e$-process as per theorem \ref{thm:optimal_phi_frequentist}, which we denote $g^g(\xi_{MDE}) = 1 / (\xi_{MDE}^2 \rho(1-\rho))$.
The second approach minimizes the width of the confidence sequence, following the discussion in Section \ref{sec:zellner_g_prior}, at the fixed-$n$ sample size $n^\star(\xi_{MDE})$, which we denote $g^w(n^\star(\xi_{MDE}))$.
For $\alpha = 0.01$, $c = 0.95$ and $\xi_{MDE} = 0.2$, we find $n^\star(\xi_{MDE}) = 1785$, $g^g(\xi_{MDE}) = 100$, and $g^w(\xi_{MDE}) = 151.29$. 
We perform $10^4$ simulations and examine the stopping time distributions.
We compare the proportion of nulls rejected by $n^\star(\xi_{MDE})$ and the average stopping times.
We also remark that the anytime-valid procedure has power 1 \citep{power1} by theorem \ref{thm:almost_sure_limit}, in the sense that it eventually rejects the null, under the alternative, with probability $1$.
We recognize that practitioners may not want to run the procedure indefinitely. Instead, we encourage practitioners to reason about the size of a practically meaningful effect and continue to monitor the confidence sequence.
Once the confidence sequence becomes a sufficiently small interval around zero such that all practically meaningful effects are excluded, then the experiment can be responsibly concluded.

\subsubsection{$\xi = \xi_{MDE}$}
\label{sec:fixed_n_comparison_xi_equal_mde}
This case is most favorable to the fixed-$n$ test, corresponding to a scenario in which the effect size has been guessed exactly, using $\xi = \xi_{MDE}$.
The stopping time distributions are visualized in Figure \ref{fig:stopping_times_comparison_0.2}.
At $n^\star(\xi_{MDE}) = 1785$ the fixed-$n$ test, anytime-valid test with $g=g^g(\xi_{MDE})$, and anytime-valid test with $g=g^w(n^\star(\xi_{MDE}))$, reject $95.02\%$, $81.69\%$, and $81.68\%$ of nulls respectively.
The average stopping times with $g = g^g(\xi_{MDE})$ and $g = g^w(n^\star(\xi_{MDE}))$ are $1214.45$ and $1236.03$ respectively.
In addition, the anytime-valid procedures rejected all nulls \textit{eventually}.
As expected, the tuning parameter $g^g(\xi_{MDE})$ leads to a smaller average stopping time than $g^w(n^\star(\xi_{MDE}))$ because of the optimality result in Theorem \ref{thm:optimal_phi_frequentist}.
\begin{figure}
\centering
\includegraphics[width=1.0\textwidth]{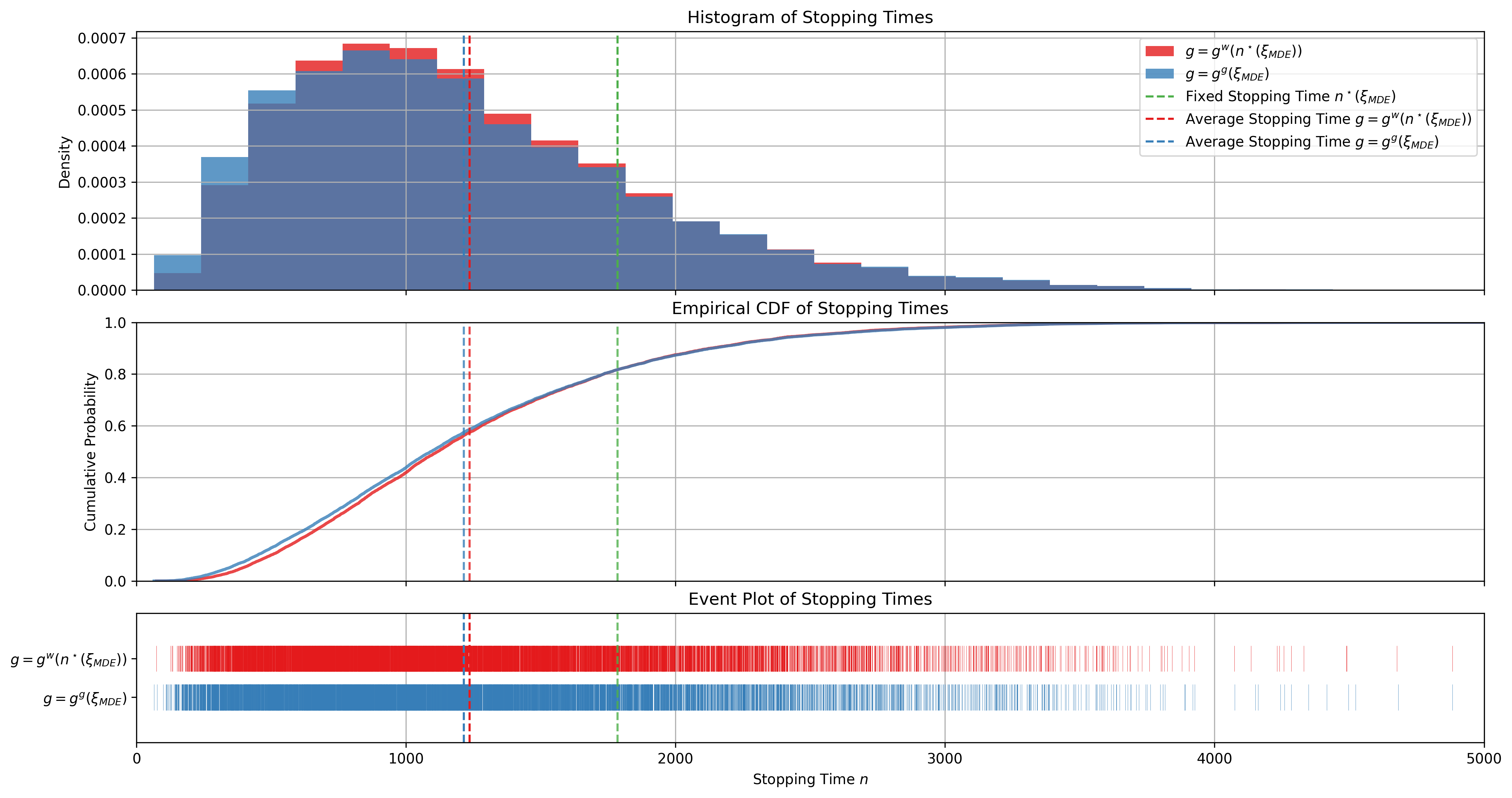}
\caption{Stopping time distributions for the fixed-$n$ analysis and the anytime-valid analyses, with $g^g(\xi_{MDE})$ and $g^w(n^\star(\xi_{MDE}))$, when $\xi = \xi_{MDE} = 0.2$.}
\label{fig:stopping_times_comparison_0.2}
\end{figure}

\subsubsection{$\xi > \xi_{MDE}$}
This case is more favorable to the anytime-valid test, corresponding to a scenario in which the effect size is larger than the MDE, $\xi = 2 \xi_{MDE} = 0.4$, which occurs when the treatment is unexpectedly harmful or effective, and early stopping even more desirable.
Note that the tuning parameter $g$ is still configured against $\xi_{MDE}$, not the true $\xi$.
Figure \ref{fig:stopping_times_comparison_0.4} shows the stopping time distributions for this case.
By $n^\star(\xi_{MDE}) = 1785$, all procedures had rejected $100\%$ of the nulls.
The average stopping times for the anytime-valid test with $g = g^g(\xi_{MDE})$ and $g = g^w(n^\star(\xi_{MDE}))$ are $350.37$ and $376.06$ respectively.
This simulation illustrates an important point that the average stopping time is dramatically less than $n^\star(\xi_{MDE}) = 1785$ when the effect size is under-estimated.
This is critical for mitigating the risk of experimentation - the more harmful a treatment, the larger the effect size and the sooner the experiment can be stopped, preventing further harm to subsequent units.
\begin{figure}
\centering
\includegraphics[width=1.0\textwidth]{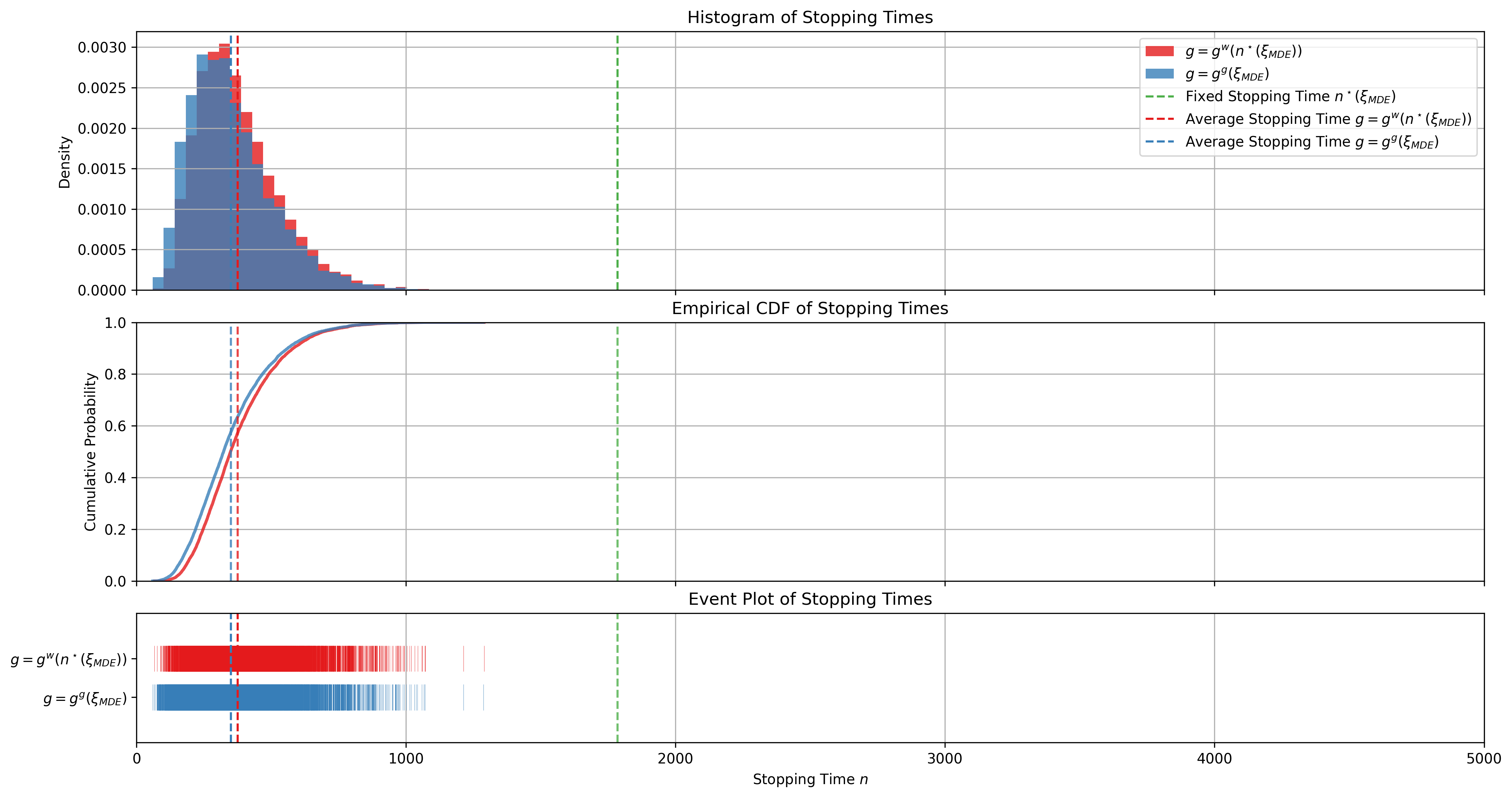}
\caption{Stopping time distributions for the fixed-$n$ analysis and the anytime-valid analyses, with $g^g(\xi_{MDE})$ and $g^w(n^\star(\xi_{MDE}))$, when $\xi = 2\xi_{MDE} = 0.4$.}
\label{fig:stopping_times_comparison_0.4}
\end{figure}
\subsubsection{$\xi = 0$}
We now consider when the null hypothesis is true. Figure \ref{fig:stopping_times_comparison_0.0} shows the Type-I error increasing with the sample size $n$ toward the nominal $\alpha$.
At $n^\star(\xi_{MDE}) = 1785$ the fixed-$n$ test, anytime-valid test with $g=g^g(\xi_{MDE})$, and anytime-valid test with $g=g^w(n^\star(\xi_{MDE}))$, reject $1.07\%$, $0.04\%$, and $0.03\%$ of nulls respectively.
\begin{figure}
\centering
\includegraphics[width=1.0\textwidth]{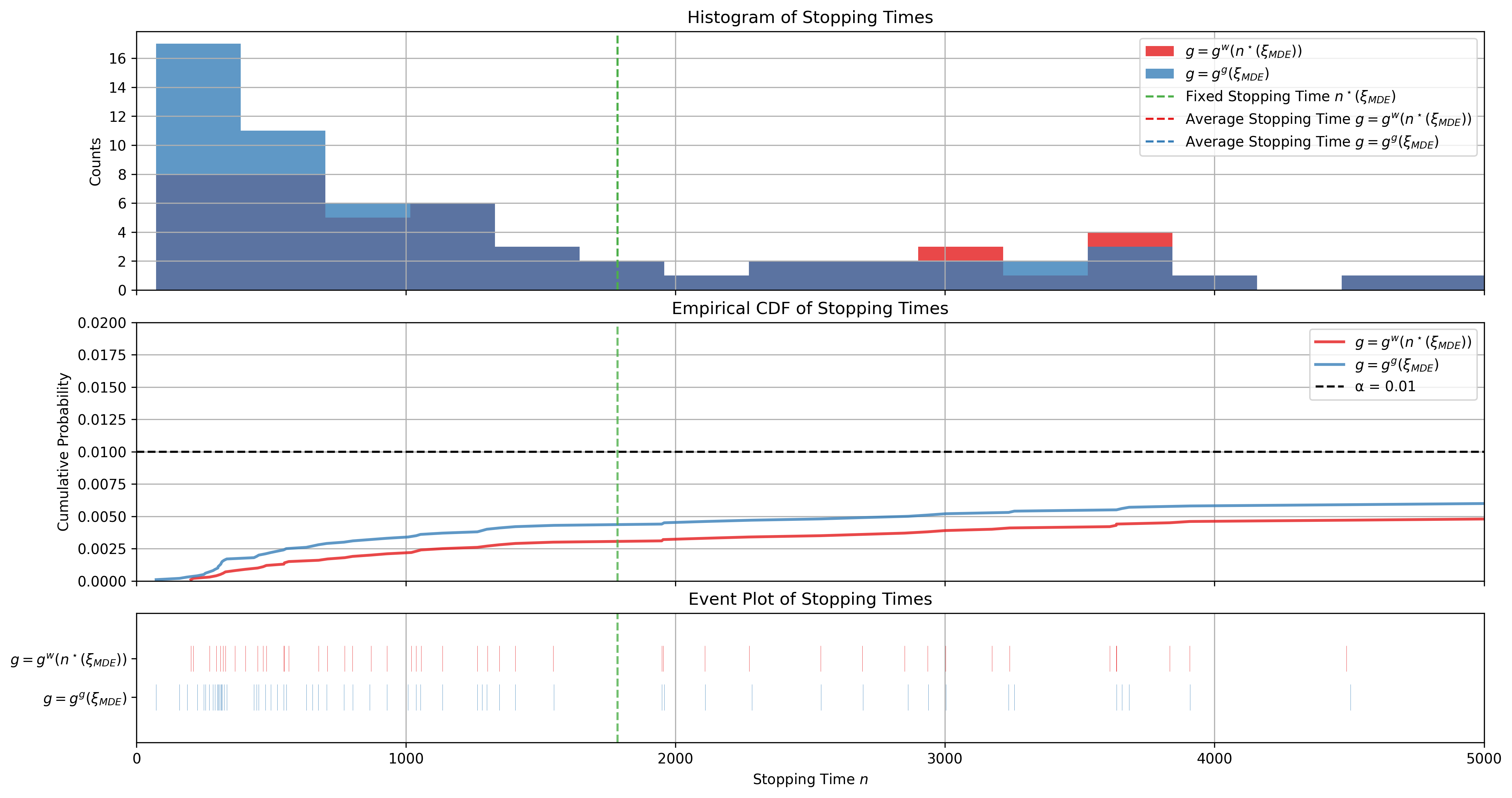}
\caption{Stopping time distributions for the fixed-$n$ analysis and the anytime-valid analyses, with $g^g(\xi_{MDE})$ and $g^w(n^\star(\xi_{MDE}))$, when $\xi = 0$.}
\label{fig:stopping_times_comparison_0.0}
\end{figure}

\section{Proofs of Section \ref{sec:heteroskedastic}}

\subsection{Proof of Lemma \ref{lem:multivariate_e_process}}
\begin{proof}
    Let
    \begin{equation}
        E_n(\bmu) := e^{\sum_{i=1}^{n} \left(\bG_i\bSigma^{-1}\bmu - \frac{1}{2}\bmu'\bSigma^{-1}\bmu\right)} = e^{\left(\bG_i\bSigma^{-1}\bmu - \frac{1}{2}\bmu'\bSigma^{-1}\bmu\right) }M_{n-1}(\bmu)
    \end{equation}
    Taking the conditional expectation gives
    \begin{equation}
        \mathbb{E}[E_n(\bmu)|\mathcal{F}_{n-1}] = \mathbb{E}\left[e^{\bG_i'\bSigma^{-1}\bmu} | \mathcal{F}_{n-1}\right]e^{-\bmu'\bSigma^{-1}\bmu}M_{n-1}(\bmu) = M_{n-1}(\bmu)  
    \end{equation}
    where the last line follows from the moment generating function of a $N(\bzero, \bSigma)$ random vector. 
    Hence, $E_n(\bmu)$ is a nonnegative supermartingale.
    Define
    \begin{equation}
        E_n := \int E_n(\bmu) dF(\bmu)= \left(\frac{g}{g + n}\right)^{\frac{d}{2}}e^{\frac{1}{2}\frac{n^2}{g+n}\left(\bar{\bG}_n \bSigma^{-1}\bar{\bG}_n\right)},
    \end{equation}
    where $\bar{\bG}_n = (1/n)\sum_i \bG_i$ and $F$ is a Gaussian measure with mean 0 and variance $g^{-1}\bSigma$.
    Then by Fubini's theorem
    \begin{align*}
            \mathbb{E}[E_n | \mathcal{F}_{n-1}] &= \mathbb{E}[\int E_n(\bmu) dF(\bmu) |\mathcal{F}_{n-1}] = \int \mathbb{E}[E_n(\bmu)  |\mathcal{F}_{n-1}]dF(\bmu) = \int E_{n-1}(\bmu)dF(\bmu)\\
            &=E_{n-1}
    \end{align*}
    hence $E_n$ is a nonnegative supermartingale.
    \end{proof}

\subsection{Proof of Proposition \ref{prop:w_ie_i_strong_approximation}}
We first restate the strong approximation result from \citet{einmahl1989} in the context of our problem.
\begin{lemma}[\citet{einmahl1989}]
    \label{lem:einmahl_iid}
    Let $(\bX_i)_{i=1}^\infty$ be an i.i.d.\ sequence of random vectors in $\mathbb{R}^k$ with $\mathbb{E}[\bX_1] = 0$, $\mathrm{Cov}(\bX_1) = \Sigma \in \mathbb{R}^{k\times k}$, and
    $\mathbb{E}\bigl[\|X_1\|^s\bigr] \leq \infty$ for some $s > 3$.
    Then on a sufficiently rich probability space there is a sequence $(\bG_i)_{i=1}^\infty$ of i.i.d.\ Gaussian random vectors $\bG_i\sim N(\bzero,\bSigma)$, such that
    
    \begin{equation}
        \bigl\|\,\sum_{i=1}^n \bX_i - \sum_{i=1}^n \bG_i\,\bigr\|=O_{a.s.}\!\bigl(n^{\alpha}\bigr)
    \end{equation}
    for some $0<\alpha<\tfrac12$. 
    \end{lemma}

By assumption $\mathbb{E}[\varepsilon_i \bw_i]=0$ and $\bSigma = \mathbb{E}[\varepsilon_i \bw_i\bw_i']$. The result therefore follows from the lemma \ref{lem:einmahl_iid} by setting $\bX_i = \bw_i(y_i - \bw_i'\bgamma) = \bw_i \varepsilon_i$.

\subsection{Proof of Proposition \ref{prop:strong_consistency}}
\begin{proof}
Write the OLS estimator as 
$$\hat{\bgamma}_n = \left(\frac{1}{n}\sum_{i=1}^n \bw_i \bw_i'\right)^{-1}\left(\frac{1}{n}\sum_{i=1}^n \bw_i y_i\right).$$
By the strong law $\frac{1}{n}\sum_{i=1}^n \bw_i y_i \overset{a.s.}{\rightarrow} \mathbb{E}[\bw_i y_i]$ and with the addition of continuity of the inverse $\left(\frac{1}{n}\sum_{i=1}^n \bw_i \bw_i'\right)^{-1}\overset{a.s.}{\rightarrow} \bOmega_{\bW}$, hence $\hat{\bgamma}_n \overset{a.s.}{\rightarrow}\bgamma = \bOmega_{\bW}^{-1}\mathbb{E}[\bw_i y_i]$.
Define the true residual $\varepsilon_i = y_i - \bw_i'\bgamma$ and estimator $\hat{\varepsilon}_{in} = y_i - \bw_i'\hat{\bgamma}_n$ and write
$$\hat{\varepsilon}_{in} = \varepsilon_i + \bw_i'(\bgamma - \hat{\bgamma}_n).$$
As $\mathbb{E}[\|\bw_i\|] < \infty$, $\|\bw_i\|$ is finite almost surely, therefore $\bw_i'(\bgamma - \hat{\bgamma}_n) \overset{a.s.}{\rightarrow} 0$ and $\hat{\varepsilon}_{in} \overset{a.s.}{\rightarrow} \varepsilon_i$.
The squared residual is
$$\hat{\varepsilon}_{in}^2 = \varepsilon_i^2 + 2\varepsilon_i\bw_i'(\bgamma - \hat{\bgamma}_n)+ (\bw_i'(\bgamma - \hat{\bgamma}_n))^2.$$
Therefore 
$$\hat{\bSigma}_n = \frac{1}{n}\sum_{i=1}^n \varepsilon_i^2\bw_i\bw_i' + 2\frac{1}{n}\sum_{i=1}^n \varepsilon_i\bw_i\bw_i' \bw_i'(\bgamma - \hat{\bgamma}_n)+ \frac{1}{n}\sum_{i=1}^n (\bw_i'(\bgamma - \hat{\bgamma}_n))^2 \bw_i\bw_i'.$$
Considering the second term
$$\| \frac{1}{n}\sum_{i=1}^n \varepsilon_i\bw_i\bw_i' \bw_i'(\bgamma - \hat{\bgamma}_n)\| \leq \|\hat{\bgamma}_n - \bgamma\| \frac{1}{n}\sum_{i=1}^n |\varepsilon_i|\|\bw_i\|^3,$$
where $\frac{1}{n}\sum_{i=1}^n |\varepsilon_i|\|\bw_i\|^3 \overset{a.s.}{\rightarrow} \mathbb{E}[|\varepsilon_i|\|\bw_i\|^3] < \infty$ by the strong law, whereas $\|\hat{\bgamma}_n - \bgamma\| \overset{a.s.}{\rightarrow} 0$, hence $\frac{1}{n}\sum_{i=1}^n \varepsilon_i\bw_i\bw_i' \bw_i'(\bgamma - \hat{\bgamma}_n) \overset{a.s.}{\rightarrow} 0$. 
Similarly $$\|\frac{1}{n}\sum_{i=1}^n (\bw_i'(\bgamma - \hat{\bgamma}_n))^2 \bw_i\bw_i'\| \leq \|\bgamma - \hat{\bgamma}_n\|^2 \frac{1}{n}\sum_{i=1}^n \|\bw_i\|^4$$ where $\frac{1}{n}\sum_{i=1}^n \|\bw_i\|^4 \overset{a.s.}{\rightarrow} \mathbb{E}[\|\bw_i\|^4] < \infty$ by the strong law, hence $$\frac{1}{n}\sum_{i=1}^n (\bw_i'(\bgamma - \hat{\bgamma}_n))^2 \bw_i\bw_i' \overset{a.s.}{\rightarrow} 0.$$
The second and third terms go to zero almost surely, whereas the first term
$$\frac{1}{n}\sum_{i=1}^n \varepsilon_i^2\bw_i\bw_i' \overset{a.s.}{\rightarrow}\bSigma,$$
hence $\hat{\bSigma}_n \overset{a.s.}{\rightarrow}\bSigma$.
\end{proof}

\subsection{Proof of Theorem \ref{thm:robust_e_process}}
We build up the proof in three steps

\subsubsection*{Step 1: Strong Approximation of OLS Estimates}
From proposition \ref{prop:w_ie_i_strong_approximation} we have that $\sum_{i=1}^n \bw_i \varepsilon_i =\sum_{i=1}^n \bG_i + o_{a.s.}(n^{\alpha})$ for some $0 < \alpha < 1/2$ where $\bG_i \overset{iid}{\sim} N(\bzero, \bSigma)$.
Under the assumption $\frac{1}{n}(\bW_n'\bW_n)\overset{a.s.}{\rightarrow} \bOmega_{\bW}$ we may write $(\bW_n'\bW_n)^{-1} = n^{-1} \bOmega_{\bW} ^{-1}+ o_{a.s.}(n^{-1})$.
Consequently,
$$\hat{\bgamma}_n-\bgamma = (\bW_n'\bW_n)^{-1}\sum_{i=1}^n \bw_i \varepsilon_i = \frac{1}{n}\sum_{i=1}^n \bOmega_{\bW}^{-1}\bG_i + O_{a.s.}(n^{\alpha - 1}) + o_{a.s.}(n^{-1})\sum_{i=1}^n \bG_i.$$
The sum $\sum_{i=1}^n \bG_i = O_{a.s.}(\sqrt{n\log\log n})$ by the law of iterated logarithm which, when combined with the $o_{a.s.}(n^{-1})$ term, is larger than $O_{a.s.}(n^{\alpha - 1})$.
Hence,
$$\bE(\hat{\bgamma}_n-\bgamma) = \frac{1}{n}\sum_{i=1}^n \tilde{\bG}_i + o_{a.s.}\left(\sqrt{\frac{\log \log n}{n}}\right),$$ where $\tilde{\bG}_i \sim N(\bzero, \bE\bOmega_{\bW}^{-1}\bSigma \bOmega_{\bW}^{-1}\bE')$,
This implies the quadratic form
$n(\hat{\bgamma}_n-\bgamma)'\bE'\bSigma_{\bE}^{-1}\bE(\hat{\bgamma}_n-\bgamma)  = n \bar{\bG}_n\bSigma_{\bE}^{-1}\bar{\bG}_n + o_{a.s.}(\log \log n)$, where $\bSigma_{\bE} = \bE\bOmega_{\bW}^{-1}\bSigma \bOmega_{\bW}^{-1}\bE'$ and $\bar{\bG}_n = (1/n)\sum_{i=1}^n \tilde{\bG}_i$.
In other words, $n(\hat{\bgamma}_n-\bgamma)'\bE'\bSigma_{\bE}^{-1}\bE(\hat{\bgamma}_n-\bgamma)$ can be approximated by a $\chi^2_r$ random variable with approximation error $o_{a.s.}(\log \log n)$. 

\subsubsection*{Step 2: An Asymptotic $e$-process}
\begin{lemma}
    \label{lem:asymptotic_e_process}
Define,
\begin{equation}
    \begin{split}
        E_n^S &\coloneqq  \left(\frac{g}{g + n}\right)^{\frac{r}{2}}e^{\frac{1}{2}\frac{n}{g+n}n\bar{\bG}_n \bSigma_{\bdelta}^{-1}\bar{\bG}_n},\\
        E_n^G &\coloneqq  \left(\frac{g}{g + n}\right)^{\frac{r}{2}}e^{\frac{1}{2}\frac{n}{g+n}n(\hat{\bgamma}_n-\bgamma)'\bE'\bSigma_{\bE}^{-1}\bE(\hat{\bgamma}_n-\bgamma)},
    \end{split}
\end{equation}
Then $E_n$ is an asymptotic $e$-process with $\log E_n^S / \log E_n^G \rightarrow 1$ almost surely.
\end{lemma}
\begin{proof}
    By lemma \ref{lem:multivariate_e_process}, $E_n^S$ is a nonnegative supermartingale.
\begin{equation}
    \begin{split}
        \frac{\log E^S_n}{\log E^G_n} &= 1 + \frac{\frac{1}{2}\frac{n}{g+n}o_{a.s.}(\log \log n)}{E^G_n}\\
        &= 1 + \frac{\frac{1}{2}\frac{n}{g+n}o_{a.s.}(\log \log n)}{\frac{d}{2} \log\left(\frac{g}{g+n}\right) + \frac{1}{2}\frac{n}{g+n}n(\hat{\bgamma}_n - \bgamma)\bE'\bSigma_{\bE}^{-1}\bE(\hat{\bgamma}_n - \bgamma)}\\
    \end{split}
\end{equation}
The term $\frac{1}{2}\frac{n}{g+n}n(\hat{\bgamma}_n - \bgamma)\bE'\bSigma_{\bE}^{-1}\bE(\hat{\bgamma}_n - \bgamma) = O_{a.s.}(\log \log n)$ by the law of iterated logarithm whereas the term $-\frac{d}{2}\log(g+n) = O(\log n)$. The denominator, therefore, grows more quickly in $n$ and
    \begin{equation}
            \frac{\log E^S_n}{\log E^G_n} = 1 + o_{a.s.}(1)
    \end{equation}
\end{proof}
We may write $E_n^S = E_n^G e^{\frac{1}{2}\frac{n}{g+n}r_n}$ where $r_n = o_{a.s.}(\log \log n)$. Although $e^{\frac{1}{2}\frac{n}{g+n}r_n}$ grows with $n$, it does not grow fast enough to meaningfully affect the probability of $E_n^G$ crossing $\alpha^{-1}$.
To see this observe that
\begin{equation*}
    \begin{split}
    &\mathbb{P}[\exists n \in \mathbb{N}:E_n^S \geq \alpha^{-1}] = \mathbb{P}[\exists n \in \mathbb{N}:E_n^G e^{\frac{1}{2}\frac{n}{g+n}r_n} \geq \alpha^{-1}]\\
    &= \mathbb{P}\left[\exists n \in \mathbb{N}: (\hat{\bgamma}_n-\bgamma)'\bE'\bSigma_{\bE}^{-1}\bE(\hat{\bgamma}_n-\bgamma) \geq \frac{1}{n}\left(\frac{g+n}{n}\right)\log\left(\frac{1}{\alpha^2}\left(\frac{g+n}{g}\right)^r\right) - \frac{r_n}{n}\right] \leq \alpha.
    \end{split}
\end{equation*}
The main boundary term $$b_n \coloneqq \frac{1}{n}\left(\frac{g+n}{n}\right)\log\left(\frac{1}{\alpha^2}\left(\frac{g+n}{g}\right)^r\right) = O\left(\frac{\log n}{n}\right),$$ whereas $\frac{r_n}{n} = o_{a.s.}\left(\frac{\log \log n}{n}\right)$ which tends to zero faster than the main boundary term.
This is aided by even smaller values of $\alpha$ which make the main boundary term relatively larger than $r_n / n$.

\subsubsection*{Step 3: Using a Plugin Estimator for $\bSigma_{\bE}$}
By Step 1, the quadratic form $n(\hat{\bgamma}_n - \bgamma)'\bE'\bSigma_{\bE}^{-1}\bE(\hat{\bgamma}_n - \bgamma)$ can be approximated as 
$$n(\hat{\bgamma}_n - \bgamma)'\bE'\bSigma_{\bE}^{-1}\bE(\hat{\bgamma}_n - \bgamma)  = n \bar{\bG}_n\bSigma_{\bdelta}^{-1}\bar{\bG}_n + o_{a.s.}(\log \log n).$$
By proposition \ref{prop:strong_consistency}, $\hat{\bSigma}_{\bE}$ is a strongly consistent estimator of $\bSigma_{\bE}$, and so we may write
$$ \bSigma_{\bE}^{-1} = \hat{\bSigma}_{\bE}^{-1} + o_{a.s.}(1)\bSigma_{\bE}^{-1}.$$
Substituting into the strong approximation gives
\begin{equation*}
    \begin{split}
            n(\hat{\bdelta}_n-\bdelta)'\hat{\bSigma}_{\bdelta_n}^{-1}(\hat{\bdelta}_n-\bdelta) =& n(\hat{\bgamma}_n - \bgamma)'\bE'\bSigma_{\bE}^{-1}\bE(\hat{\bgamma}_n - \bgamma) + o_{a.s.}(1)n(\hat{\bgamma}_n - \bgamma)'\bE'\bSigma_{\bE}^{-1}\bE(\hat{\bgamma}_n - \bgamma)\\
            =& n \bar{\bG}_n\bSigma_{\bdelta}^{-1}\bar{\bG}_n + o_{a.s.}(\log \log n)+ o_{a.s.}(1)n(\hat{\bgamma}_n - \bgamma)'\bE'\bSigma_{\bE}^{-1}\bE(\hat{\bgamma}_n - \bgamma)\\
    \end{split}
\end{equation*}
The quadratic form $n(\hat{\bgamma}_n - \bgamma)'\bE'\bSigma_{\bE}^{-1}\bE(\hat{\bgamma}_n - \bgamma)=O_{a.s.}(\log \log n)$ by the law of the iterated logarithm, and so the last term is also $o_{a.s.}(\log \log n)$. Therefore
\begin{equation*}
    n(\hat{\bgamma}_n - \bgamma)'\bE'\hat{\bSigma}_{\bE}^{-1}\bE(\hat{\bgamma}_n - \bgamma) = n(\hat{\bgamma}_n - \bgamma)'\bE'\bSigma_{\bE}^{-1}\bE(\hat{\bgamma}_n - \bgamma) + r_n
\end{equation*}
where $r_n = o_{a.s.}(\log \log n)$. The $e$-process using the plugin estimator can be written as 
\begin{equation*}
        E^p_n \coloneqq  \left(\frac{g}{g + n}\right)^{\frac{r}{2}}e^{\frac{1}{2}\frac{n}{g+n}n(\hat{\bgamma}_n - \bgamma)'\bE'\hat{\bSigma}_{\bE}^{-1}\bE(\hat{\bgamma}_n - \bgamma)}=E_n^S e^{\frac{1}{2}\frac{n}{g+n}r_n},
\end{equation*}
The quotient $\log E^p_n / \log E_n^S \rightarrow 1$ almost surely by the same argument as lemma \ref{lem:asymptotic_e_process}, therefore $E_n$ is an asymptotic $e$-process.

\subsubsection*{Step 4: Swap With Asymptotic ``$t$'' $e$-process}
Let $Q_n= n(\hat{\bgamma}_n - \bgamma)'\bE'\hat{\bSigma}_{\bE}^{-1}\bE(\hat{\bgamma}_n - \bgamma)$ define
\begin{equation*}
    \begin{split}
        E_n &= \left(\frac{g}{g+n}\right)^{\frac{r}{2}}\left(\frac{1 + \frac{g}{g+n} \frac{r}{\nu_n} \frac{Q_n}{r}}{1+\frac{r}{\nu_n} \frac{Q_n}{r}}\right)^{-\frac{\nu_n + r}{2}},\\
        E^p_n &\coloneqq  \left(\frac{g}{g + n}\right)^{\frac{r}{2}}e^{\frac{1}{2}\frac{n}{g+n}Q_n}.
    \end{split}
\end{equation*}
The logarithmic difference is simply
\begin{align*}
\log E_n - \log E_n^p
&= - \frac{\nu_n + r}{2} \log\left( \frac{1 + \frac{g}{g+n} \cdot \frac{Q_n}{\nu_n}}{1 + \frac{Q_n}{\nu_n}} \right) - \frac{1}{2} \cdot \frac{n}{g+n} Q_n.
\end{align*}
Observe $\frac{r}{\nu_n} \frac{Q_n}{r} = \frac{Q_n}{\nu_n} \overset{a.s.}{\rightarrow} 0$.
Let $x = \frac{Q_n}{\nu_n}$ and $a = \frac{g}{g+n}$. Using the expansion \( \log(1 + x) = x - \frac{1}{2}x^2 + o(x^2) \), we have:
\begin{align*}
\log\left( \frac{1 + a x}{1 + x} \right)
&= \log(1 + a x) - \log(1 + x) \\
&= (a - 1)x - \frac{1}{2}(a^2 - 1)x^2 + o_{a.s.}(x^2),
\end{align*}
as $x \overset{a.s.}{\rightarrow} 0$. Hence,
\begin{align*}
\log\left( \frac{1 + \frac{g}{g+n} \cdot \frac{Q_n}{\nu_n}}{1 + \frac{Q_n}{\nu_n}} \right) = -\frac{n}{n+g}\frac{Q_n}{\nu_n} + O_{a.s.}\left( \frac{Q_n^2}{\nu_n^2} \right),\\
\end{align*}
and
\begin{align*}
- \frac{\nu_n + r}{2}\log\left( \frac{1 + \frac{g}{g+n} \cdot \frac{Q_n}{\nu_n}}{1 + \frac{Q_n}{\nu_n}} \right) = \frac{\nu_n + r}{2}\frac{n}{n+g}\frac{Q_n}{\nu_n} + O_{a.s.}\left( \frac{Q_n^2}{\nu_n} \right),\\
\end{align*}

Plugging back in:
\begin{align*}
\log E_n - \log E_n^p 
&=  \frac{1}{2}\left(\frac{\nu_n + r}{\nu_n} - 1 \right)\frac{n}{n+g}Q_n + O_{a.s.}\left( \frac{Q_n^2}{\nu_n} \right),\\
&=  \frac{r}{2} \frac{n}{n+g}\frac{Q_n}{\nu_n} + O_{a.s.}\left( \frac{Q_n^2}{\nu_n} \right),\\
\end{align*}
and since $\frac{Q_n}{\nu_n} \overset{a.s.}{\rightarrow} 0$, we conclude
$$\log E_n = \log E_n^p + o_{a.s.}(1).$$
As neither $\log E_n$ nor $\log E_n^p$ converge to zero, we have that $E_n$ is an asymptotic $e$-process,
$$\frac{\log E_n}{\log E_n^p} \overset{a.s.}{\rightarrow} 1.$$

\subsection{Proof of Theorem \ref{thm:asmyptotic_radius}}
\begin{lemma}
    \label{lem:inequality}
$(v^{\lambda} - v^{1-\lambda}) \leq (1-\lambda)(1-v)$ for $v \in (0,1)$ and $\lambda \in (0,1)$
\end{lemma}
\begin{proof}
consider the function:
$$
f(v) = v^{\lambda} - v^{1-\lambda} - (1-\lambda)(1-v).
$$
We will show that $f(v) \leq 0$ for all $v \in (0,1)$ and $\lambda \in (0,1)$. 
The boundaries of the functions are $f(1) = 0$ and $f(0) = - (1-\lambda) < 0$.
Clearly when $\lambda = 0.5$, we have $f(v) = -0.5(1-v) \leq 0$.
The first and second derivatives are
\begin{align*}
f'(v) &= \lambda v^{\lambda - 1} - (1-\lambda)v^{-\lambda} + (1-\lambda),\\
f''(v) &= \lambda(1-\lambda) (v^{-\lambda - 1} - v^{\lambda - 2})\\
&=\lambda(1-\lambda) v^{-(\lambda + 1)}(1-v^{2\lambda - 1})
\end{align*}
When $\lambda > 0.5$, $2\lambda - 1 > 0$ and $f''(v) \geq 0$ for all $v \in (0,1)$ and therefore $f(v)$ is strictly concave, decreasing from $-(1-\lambda)$ to a minimum on $(0,1)$ and then increasing toward 0, therefore $f(v) \leq 0$ for all $v$.
When $\lambda < 0.5$, $2\lambda - 1 > 0$ and $f''(v) \leq 0$ for all $v \in (0,1)$ and therefore $f(v)$ is strictly convex - increasing from $-(1-\lambda)$ to 0 as $v$ goes from 0 to 1,  therefore $f(v) \leq 0$ for all $v$.

\end{proof}
\subsubsection{Part 1: }
\begin{proof}
Let $\nu_n = n - k$, $m = \nu_n + r$ and set
\begin{align*}
    u &= \left(\frac{\alpha^{\frac{2}{r}}g}{n+g}\right)^{\frac{r}{\nu_n + r}}\\
    v &=\frac{g}{n+g}.
\end{align*}
Note $v \in (0,1)$ and $u \in (0, v^{\frac{r}{\nu_n + r}})$. When $u \leq v$, we have that $R(g,n,\alpha)$ is infinite and $R(g,n,\alpha) \geq R^G(g,n,\alpha)$ is satisfied.
We now consider the case $u \in (v, v^{\frac{r}{\nu_n+r}})$, where

\begin{align*}
R(g,n,\alpha) = \frac{\nu_n}{r}\frac{1 - u}{u - v},\\
\end{align*}
Note
$$
R^G(g,n,\alpha)=-\frac{\nu_n + r}{r}\frac{n+g}{n} \log u.
$$
Therefore the statement
$$R(g,n,\alpha)  - R^G(g,n,\alpha) = \frac{\nu_n}{r}\frac{1 - u}{u - v} + \frac{\nu_n + r}{r}\frac{n+g}{n} \log u > 0,$$
is equivalent to
$$
\nu_n(1 - u)(1-v) +(\nu_n + r)(u-v) \log u > 0.
$$
We prove this by considering the function
$$
 F(b)\coloneqq \nu_n(1 - bw)(1-v) +(\nu_n + r)(bw-v) \log bw,
$$
where $w = v^{\frac{r}{\nu_n + r}}$ and $b \in (v^{\frac{\nu_n}{\nu_n + r}} , 1)$.
In particular $u > v$ implies $\alpha^\frac{2}{\nu_n + r} \in (v^{\frac{\nu_n}{\nu_n + r}} , 1)$.
If we can show that $F(b) > 0 $ for all $b \in (v^{\frac{\nu_n + r}{r}} , 1)$, then $F(\alpha^{\frac{2}{\nu_n + r}}) > 0$ which implies the result because $u = \alpha^{\frac{2}{\nu_n + r}} v^{\frac{r}{\nu_n + r}}$.
First we observe that $F(b)$ is strictly convex on $(v^{\frac{\nu_n + r}{r}} , 1)$
$$
F''(b) = (\nu_n + r) \left(\frac{wb + v}{b^2}\right)> 0.
$$
Second note that $F(b)$ is strictly decreasing on $(v^{\frac{\nu_n + r}{r}} , 1)$, since

\begin{align*}
    F'(b) &= w\left(-\nu_n(1-v) + (\nu_n + r)(\log(bw) + 1 - \frac{v}{bw})\right)\\
    &\leq w\left(-\nu_n(1-v) + (\nu_n + r)(bw - \frac{v}{bw})\right)\\
    &\leq w\left(-\nu_n(1-v) + (\nu_n + r)(w - \frac{v}{w})\right)\\
    &= w\left(-\nu_n(1-v) + (\nu_n + r)(v^{\lambda} - v^{1-\lambda})\right)\\
    &\leq w\left(-\nu_n(1-v) + (\nu_n + r)(1-\lambda)(1-v)\right)\\
    &= w\left(-\nu_n(1-v) + \nu_n(1-v)\right)\\
    &=0,
\end{align*}
where $\lambda = \frac{r}{\nu_n + r}$ and by using lemma \ref{lem:inequality} to achieve the 5th line.
Therefore $F(b)$ achieves a minimum at the largest value of $b$, which is $b = 1$.
Therefore $F(b) \geq F(1)$ which implies $F(\alpha^{\frac{2}{\nu_n + r}}) \geq  0$. We now show that $F(1) \geq 0$.

\begin{align*}
    F(1) &= \nu_n(1 - w)(1-v) +(\nu_n + r)(w-v) \log w\\
    &=(v^{\lambda}-v)\left(\nu_n(1-v)\frac{(1-v^{\lambda})}{v^{\lambda}-v} + (\nu_n+r)\lambda \log v\right)\\
    &=\frac{(v^{\lambda}-v)}{(\nu_n + r)}\left((1-\lambda)(1-v)\frac{(1-v^{\lambda})}{v^{\lambda}-v} + \lambda \log v\right)\\
    &\geq 0.
\end{align*}

\end{proof}
\subsubsection*{Proof of $\underset{n\rightarrow \infty}{\lim} \frac{R(g, n, \alpha)}{R^G(g, n, \alpha)} = 1$}

\begin{proof}
First we observe that $R^G(g, n, \alpha) \sim \log n$,
\begin{align*}
        R^G(g, n, \alpha) &= \frac{g+n}{n} \log n + \frac{g+n}{n} \log\left(\frac{1 + \frac{g}{n}}{\alpha^{\frac{2}{r}}g}\right),\\
        &= \frac{g+n}{n} \log n -\frac{g+n}{n}\log(\alpha^{\frac{2}{r}}g)+ \frac{g+n}{n} \log\left(1 + \frac{g}{n}\right),\\
        &= \frac{g+n}{n} \log n -\frac{g+n}{n}\log(\alpha^{\frac{2}{r}}g)+ \frac{g+n}{n} \left(\frac{g}{n} + O(n^{-2})\right),\\
        &= \log n - \log(\alpha^{\frac{2}{r}}g) +O\left(\frac{\log n}{n}\right).
\end{align*}
To see $R(g, n, \alpha)\sim \log n$, first write
$$
\left(\frac{\alpha^{\frac{2}{r}}g}{g+n}\right)^{\frac{r}{\nu_n + r}
} = e^{\frac{r}{r + \nu_n}\log\left(\frac{\alpha^{\frac{2}{r}}g}{g+n}\right)},
$$
and consider the exponent. The first term 
$$
\frac{r}{\nu_n + r} = \frac{r}{\nu_n}\frac{1}{1 + \frac{r}{\nu_n}} =\frac{r}{\nu_n}\left(1 + O(n^{-1})\right).
$$
The logarithmic term 
\begin{align*}
    \log\left(\frac{\alpha^{\frac{2}{r}}g}{g+n}\right) &= \log \alpha^{\frac{2}{r}}g - \log n - \log\left(1 + \frac{g}{n}\right) \\
    &= \log \alpha^{\frac{2}{r}}g - \log n - O(n^{-1}). \\
\end{align*}
Multiplying the two terms together we have 
$$
\frac{r}{\nu_n + r}\log\left(\frac{\alpha^{\frac{2}{r}}g}{g+n}\right) = \frac{r}{\nu_n}\log \alpha^{\frac{2}{r}}g - \frac{r}{\nu_n}\log n - O\left(\frac{\log(n)}{n^2}\right),
$$
which converges to zero as $n\rightarrow \infty$. Therefore when we exponentiate this expression 
$$
\left(\frac{\alpha^{\frac{2}{r}}g}{g+n}\right)^{\frac{r}{\nu_n + r}
} = e^{\frac{r}{r + \nu_n}\log\left(\frac{\alpha^{\frac{2}{r}}g}{g+n}\right)} = 1 + \frac{r}{\nu_n}\log \alpha^{\frac{2}{r}}g - \frac{r}{\nu_n}\log n + O\left(\frac{\log(n)}{n^2}\right).
$$
Substituting back into the expansion for $R(g, n, \alpha)$ we have, after collecting terms,
$$
R(g, n, \alpha) = \log n - \log(\alpha^{\frac{2}{r}}g)  + O\left(\frac{\log(n)}{n}\right).
$$
Therefore,
$$
\underset{n\rightarrow \infty}{\lim} \frac{R(g, n, \alpha)}{R^G(g, n, \alpha)} = \frac{\log n - \log(\alpha^{\frac{2}{r}}g)  + O\left(\frac{\log(n)}{n}\right)}{\log n - \log(\alpha^{\frac{2}{r}}g)  + O\left(\frac{\log(n)}{n}\right)} = 1.
$$
\end{proof}

\section{Proofs for Section \ref{sec:ate}}

\subsection{Proof of Lemma \ref{lem:population_least_squares}}
The first order conditions satisfy
\begin{equation*}
    \begin{pmatrix}
     \mathbb{E}[   y_i - \alpha^\star - (\bm_i - \bmu_{\bm})'\bzeta^\star - (T_i-\rho)\tau^\star - T_i(\bm_i - \bmu_{\bm})'\beeta^\star]\\
        \mathbb{E}[ (\bm_i - \bmu_{\bm})(y_i - \alpha^\star - (\bm_i - \bmu_{\bm})'\bzeta^\star - (T_i-\rho)\tau^\star - T_i(\bm_i - \bmu_{\bm})'\beeta^\star)]\\
        \mathbb{E}[ (T_i-\rho)(y_i - \alpha^\star - (\bm_i - \bmu_{\bm})'\bzeta^\star - (T_i-\rho)\tau^\star - T_i(\bm_i - \bmu_{\bm})'\beeta^\star)]\\
        \mathbb{E}[ T_i(\bm_i - \bmu_{\bm})(y_i - \alpha^\star - (\bm_i - \bmu_{\bm})'\bzeta^\star - (T_i-\rho)\tau^\star - T_i(\bm_i - \bmu_{\bm})'\beeta^\star)]\\
    \end{pmatrix} = \bzero
\end{equation*}
These expressions simplify to
\begin{equation*}
    \begin{pmatrix}
     \mathbb{E}[y_i] - \alpha^\star\\
        \mathbb{E}[ (\bm_i - \bmu_{\bm})y_i] - \bOmega_{\bm}\bzeta^\star - \rho\bOmega_{\bm}\beeta^\star\\
        \mathbb{E}[ (T_i-\rho)y_i] -  \rho(1-\rho)\tau^\star\\
        \mathbb{E}[ T_i(\bm_i - \bmu_{\bm})y_i] - \rho\bOmega_{\bm}\bzeta^\star - \rho\bOmega_{\bm}\beeta^\star\\
    \end{pmatrix} = \bzero
\end{equation*}
Clearly $\alpha^\star = \mathbb{E}[y_i] = \rho\mathbb{E}[y_i(1)] + (1-\rho)\mathbb{E}[y_i(0)]$ and $\tau^{\star} = \mathbb{E}[y_i(1)] - \mathbb{E}[y_i(0)]$. 
Hence if the average treatment effect is zero, then $\tau^{\star}=0$. From the last element we have $\bOmega_{\bm}^{-1}\mathbb{E}[(\bm_i - \bmu_{\bm})y_i(1)] = \bzeta^\star + \beeta^\star $, and from the second line we have $\bOmega_{\bm}^{-1}\mathbb{E}[(\bm_i - \bmu_{\bm})y_i] = \bzeta^\star + \rho\beeta^\star$. 
Therefore
\begin{equation*}
\begin{split}
        \beeta^{\star} &= (1-\rho)^{-1}\bOmega_{\bm}^{-1}(\mathbb{E}[(\bm_i - \bmu_{\bm})y_i(1)] - \mathbb{E}[(\bm_i - \bmu_{\bm})y_i] )\\
         &= \bOmega_{\bm}^{-1}(\mathbb{E}[(\bm_i - \bmu_{\bm})y_i(1)] - \mathbb{E}[(\bm_i - \bmu_{\bm})y_i(0)] )\\
\bzeta^{\star} &=\bOmega_{\bm}^{-1} \mathbb{E}[(\bm_i - \bmu_{\bm})y_i(0)]\\
\end{split}
\end{equation*}
Observe that $\mathbb{E}[y_i(1) - y_i(0)|\bm_i = \bm]  =0$ for all $\bm$ (no conditional average treatment effect) implies $\beeta^\star = \bzero$ because
\begin{equation*}
\begin{split}
        \beeta^{\star} &= \bOmega_{\bm}^{-1}(\mathbb{E}[(\bm_i - \bmu_{\bm})y_i(1)] - \mathbb{E}[(\bm_i - \bmu_{\bm})y_i(0)] )\\
        &= \bOmega_{\bm}^{-1}(\mathbb{E}[(\bm_i - \bmu_{\bm})\mathbb{E}[y_i(1) - y_i(0)|\bm_i]] )\\
        &=\bzero
\end{split}
\end{equation*}

\subsection{Proof of Lemma \ref{lem:asymptotic_ols_params}}
\begin{equation*}
\begin{split}
&\mathbb{E}[\bw_i\varepsilon_i]=\\
&\mathbb{E} \left[\begin{pmatrix}
    1\\
    \bm_i-\bmu_\bm\\
    T_i - \rho \\
    T_i(\bm_i-\bmu_\bm)\\
\end{pmatrix}(y_i - \alpha^{\star} -(\bm_i-\bmu_\bm)'\bzeta^{\star} -(T_i-\rho)\tau^{\star} - T_i(\bm_i-\bmu_\bm)'\beeta^{\star})\right]
\end{split}
\end{equation*}
We show each element is zero, starting with the first.
    \begin{equation*}
        \begin{split}
                &\mathbb{E}[y_i - \alpha^{\star} - (\bm_i-\bmu_\bm)'\bzeta^{\star} - (T_i-\rho)\tau^\star - T_i(\bm_i-\bmu_\bm)'\beeta^{\star})] \\
                &= \mathbb{E}[y_i] - \alpha^{\star} - 0 - 0\\
                &= \rho \mathbb{E}[y_i(1)] + (1-\rho) \mathbb{E}[y_i(0)]  - \alpha^{\star} = 0
        \end{split}
    \end{equation*}
    Line 2 follows from $\mathbb{E}[T_i-\rho]=0$, $\mathbb{E}[\bm_i-\bmu_{\bm}]=0$, and the final line from the definition of $\alpha^{\star}$ in lemma \ref{lem:population_least_squares}.
    \begin{equation*}
            \begin{split}
                &\mathbb{E}[(\bm_i-\bmu_\bm)(y_i - \alpha^{\star}- (\bm_i-\bmu_\bm)'\bzeta^{\star} - (T_i-\rho)\tau^\star - T_i(\bm_i-\bmu_\bm)'\beeta^{\star})]\\
                &=\mathbb{E}[(\bm_i-\bmu_\bm)y_i] - 0  - \bOmega_\bm\bzeta^{\star} - 0 - \rho \bOmega_{\bm}\beeta^{\star}\\
                &=0
        \end{split}
    \end{equation*}
    The first equality follows from $\mathbb{E}[(\bm_i-\bmu_\bm)\alpha^{\star}]=0$ and $\mathbb{E}[(\bm_i-\bmu_\bm)(T_i-\rho)]=0$ by independence of $T_i$ and $\bm_i$. 
    The last line follows from the definition of $\bzeta^{\star}$ and $\beeta^{\star}$.
        \begin{equation*}
            \begin{split}
                &\mathbb{E}[(T_i-\rho)(y_i - \alpha^{\star} - (\bm_i-\bmu_\bm)'\bzeta^{\star} - (T_i-\rho)\tau^\star - T_i(\bm_i-\bmu_\bm)'\beeta^{\star})]\\
                &=\mathbb{E}[(T_i-\rho)y_i] - 0 -0 - \mathbb{E}[(T_i-\rho)^2]\tau^\star - 0\\
                &=\rho(1-\rho)\mathbb{E}[y_i(1)] - \rho(1-\rho)\mathbb{E}[y_i(0)] - \rho(1-\rho)\tau^\star\\
                &=0
        \end{split}
    \end{equation*}
    The first line follows because  $\mathbb{E}[(T_i - \rho)\alpha^{\star}]=0$ and $\mathbb{E}[(\bm_i-\bmu_\bm)(T_i-\rho)]=0$ by independence of $T_i$ and $\bm_i$. The last line follows from the definition of $\tau^\star$.
    \begin{equation*}
            \begin{split}
                &\mathbb{E}[T_i(\bm_i-\bmu_\bm)(y_i - \alpha^{\star}- (\bm_i-\bmu_\bm)'\bzeta^{\star} - (T_i-\rho)\tau^\star - T_i(\bm_i-\bmu_\bm)'\beeta^{\star})]\\
                &=\rho\mathbb{E}[(\bm_i-\bmu_\bm)y_i(1)] - 0  - \rho\bOmega_\bm\bzeta^{\star} - 0 - \rho \bOmega_{\bm}\beeta^{\star}\\
                &=0
        \end{split}
    \end{equation*}
    where the last line follows from the definitions of $\beeta^{\star}$ and $\bzeta^{\star}$.
Having established that this vector is zero mean, $\mathbb{V}[\bw_i(y_i - \bw_i'\bgamma^{\star})] = \mathbb{E}[(y_i - \bw_i'\bgamma^{\star})^2\bw_i\bw_i']=\mathbb{E}[\mathbb{E}[(y_i - \bw_i'\bgamma^{\star})^2|\bw_i]\bw_i\bw_i']$.
The outer product of the vector $\bw_i$ is given by
\begin{equation*}
    \begin{split}
        &\bw_i \bw_i' =\\
        & \begin{pmatrix} 1 & (\bm_i-\bmu_\bm)' & T_i-\rho & \bigl[T_i(\bm_i-\bmu_\bm)\bigr]' \\ \bm_i-\bmu_\bm & (\bm_i-\bmu_\bm)(\bm_i-\bmu_\bm)' & (T_i-\rho)(\bm_i-\bmu_\bm) & T_i(\bm_i-\bmu_\bm)(\bm_i-\bmu_\bm)' \\ T_i-\rho & (T_i-\rho)(\bm_i-\bmu_\bm)' & (T_i-\rho)^2 & (T_i-\rho)T_i(\bm_i-\bmu_\bm)' \\ T_i(\bm_i-\bmu_\bm) & T_i(\bm_i-\bmu_\bm)(\bm_i-\bmu_\bm)' & T_i(\bm_i-\bmu_\bm)(T_i-\rho) & T_i^2 (\bm_i-\bmu_\bm)(\bm_i-\bmu_\bm)' \end{pmatrix}.
    \end{split}
\end{equation*}
Hence
\begin{equation*}
    \bOmega_{\bW} = \mathbb{E}[\bw_i \bw_i'] = \begin{pmatrix} 1 & \mathbf{0}' & 0 & \mathbf{0}' \\ \mathbf{0} & \bOmega_{\bm} & \mathbf{0} & \rho\,\bOmega_{\bm} \\ 0 & \mathbf{0}' & \rho(1-\rho) & \mathbf{0}' \\ \mathbf{0} & \rho\,\bOmega_{\bm} & \mathbf{0} & \rho\,\bOmega_{\bm} \end{pmatrix}.
\end{equation*}

\section{Comparison with AIPW and Empirical Bernstein Confidence Sequences}
\label{sec:aipw}
The AIPW write estimator, using a linear regression function, is identical the OLS estimator from the linear model.
Specifically, consider the influence functions 
 \begin{equation*}
  \phi_i(m^0, m^1) 
= \frac{Z_i -\rho}{\rho(1-\rho)}(Y_i -  m^{T_i}(\bX_i))
+ m^1(\bX_i) - m^0(\bX_i),
\end{equation*}
and write AIPW estimate of the average treatment effect using regression functions $m^0$ and $m^1$ as
\begin{equation}
  \label{eq:aipw}
  \hat{\delta}_n^{AIPW}(m^0, m^1) = \frac{1}{n} \sum_{i=1}^n \phi_i(m^0, m^1),
\end{equation}
Suppose we fit a linear model such that $\hat{m}^{Z_i}_n(\bx_i) = \beta_0 + \bx_i'\hat{\bbeta}_n + Z_i\hat{\delta}_n$. 
Then the following equality holds between the OLS and AIPW estimators:

$$\hat{\delta}_n = \hat{\delta}_n^{AIPW}(\hat{m}^0_n, \hat{m}^1_n).$$

An important distinction, however, is that recent literature on anytime-valid inference uses a different \textit{sequential version} of the AIPW estimator \citep{ham2022,howard}.
Given a sequence of observations $(y_i, \bX_i, Z_i)_{i=1}^\infty$ we can define a new sequence of $(\phi_i)_{i=1}^\infty$ using the AIPW approach as:

\begin{equation*}
  \phi_i 
= \frac{Z_i - \rho}{\rho(1-\rho)}(Y_i - \hat m^{T_i}_i(\bX_i))
+ \hat m^1_i(\bX_i) - \hat m^0_i(\bX_i)
\end{equation*}
where $\hat{m}^0_i$ and $\hat{m}^1_i$ are regression functions trained on data up to—but not including—time $i$, ensuring measurability with respect to the filtration $\mathcal{F}_{i-1}$.
Under this construction we have $\mathbb{E}[\phi_i | \mathcal{F}_{i-1}] = \mathbb{E}[y_i | Z_i=1] - \mathbb{E}[y_i|Z_i=0]$. This yields the \textit{sequential} AIPW estimator:
\begin{equation}
  \label{eq:aipw_seq}
  \hat{\delta}_n^{SAIPW} = \frac{1}{n} \sum_{i=1}^n \phi_i(\hat{m}^0_{i}, \hat{m}^1_{i}),
\end{equation}
where, in contrast to $\hat{\delta}_n^{AIPW}(\hat{m}^0_n, \hat{m}^1_n)$, $\phi_i(\hat{m}^0_{n}, \hat{m}^1_{n})$ have been replaced by $\phi_i(\hat{m}^0_{i}, \hat{m}^1_{i})$.
In general, $\hat{\delta}_n^{SAIPW}  \neq \hat{\delta}_n^{AIPW}(\hat{m}^0_n, \hat{m}^1_n)$, even when both estimators are constructed from linear models.
The strength of this approach is that is can accommodate nonlinear regression functions and varying propensities. A weakness is that is necessitates an online learning procedure for estimating $m^0$ and $m^1$, iterating through each outcome in sequence, which can be computationally intensive.

In contrast, our approach allows the practitioner to compute a single OLS fit at time $n$ and evaluate the corresponding $e$-process directly, without requiring sequential updates for each observation. This is especially advantageous in scenarios where only summary statistics (e.g., the point estimate and standard error) are available, such as in published regression tables. It also simplifies integration into existing software pipelines: users can apply our procedure using standard OLS output without additional model refitting or data iteration.

The PlayDelay outcome in section \ref{sec:examples} is bounded between 0 and 1, which enables the usage of the Empirical Bernstein confidence sequence for bounded outcomes by \citet{howard}.
The AIPW-transformed outcomes $\phi_i$ are bounded between $[-2, 2]$ due to the propensity score reweighting with $\rho = 1/2$.
To configure the Empirical Bernstein confidence sequence, we use a scale of $4$ and use the Gamma-Exponential conjugate mixture \citep[Proposition~9]{howard} with tuning parameter $100$ to provide an equitable comparison
to the Gaussian mixture with precision $100$, as suggested by \citet[Section~3.5]{howard}.
Figure~\ref{fig:asymp_t_playdelay_alt} illustrates the average stopping time of the Empirical Bernstein procedure is much longer, which communicates that conservativeness is the price paid for nonasymptotic guarantees.
Indeed, the bound on the outcomes, required by the Empirical Bernstein procedure, is very large due to the presence of extreme values in the PlayDelay data, which translates to a very large stopping time.

\section{Supplemental Figures}

\begin{figure}[ht]
        \centering
        \includegraphics[width=1\linewidth]{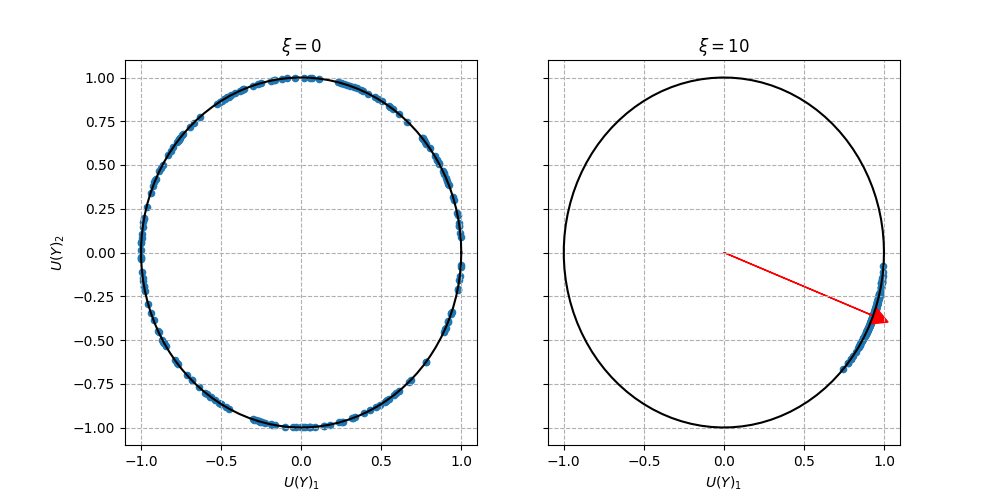}
        \caption{Two hundred samples of $\bm_n$ from the linear model $\bY_n = \bX_n\beta + \bZ_n\xi\sigma + \sigma\epsilon_n$ where $\bX_n$ and $\bZ_n$ are $5 \times 3$ and $5 \times 1$ matrices of standard normals respectively, $\beta = (1,2,3)$ and $\sigma = 1.5$. (Left) The null hypothesis is true and $\bm_n$ is uniformly distributed over the unit circle. (Right) The alternative hypothesis is true and $\bm_n(\bY_n)$ has mean direction $\bA_n\bZ_n / \|\bA_n\bZ_n\|_2$ indicated by the red arrow.}
        \label{fig:U_scatter}
    \end{figure}

\begin{figure}[ht]
    \centering
    \includegraphics[width=\textwidth]{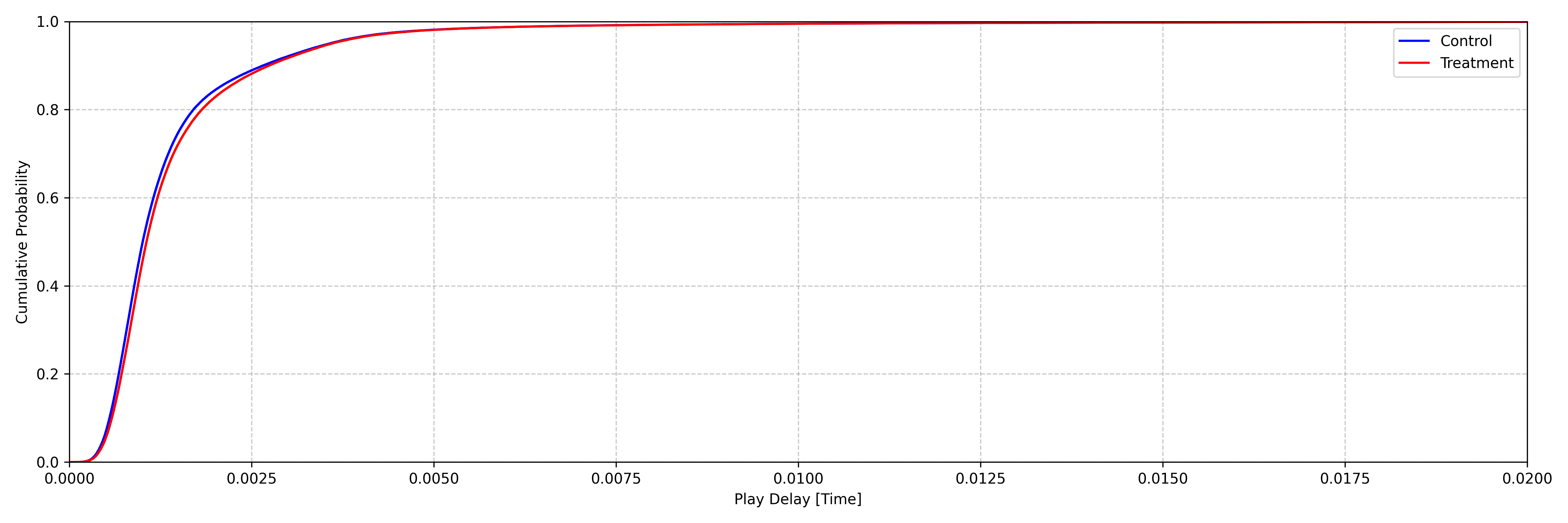}
    \caption{Empirical CDF's of normalized \textit{Play Delay} measurements for treatment and control groups}
    \label{fig:jasa_play_delay_cdf}
\end{figure}

\begin{figure}[ht]
    \centering
    \includegraphics[width=\textwidth]{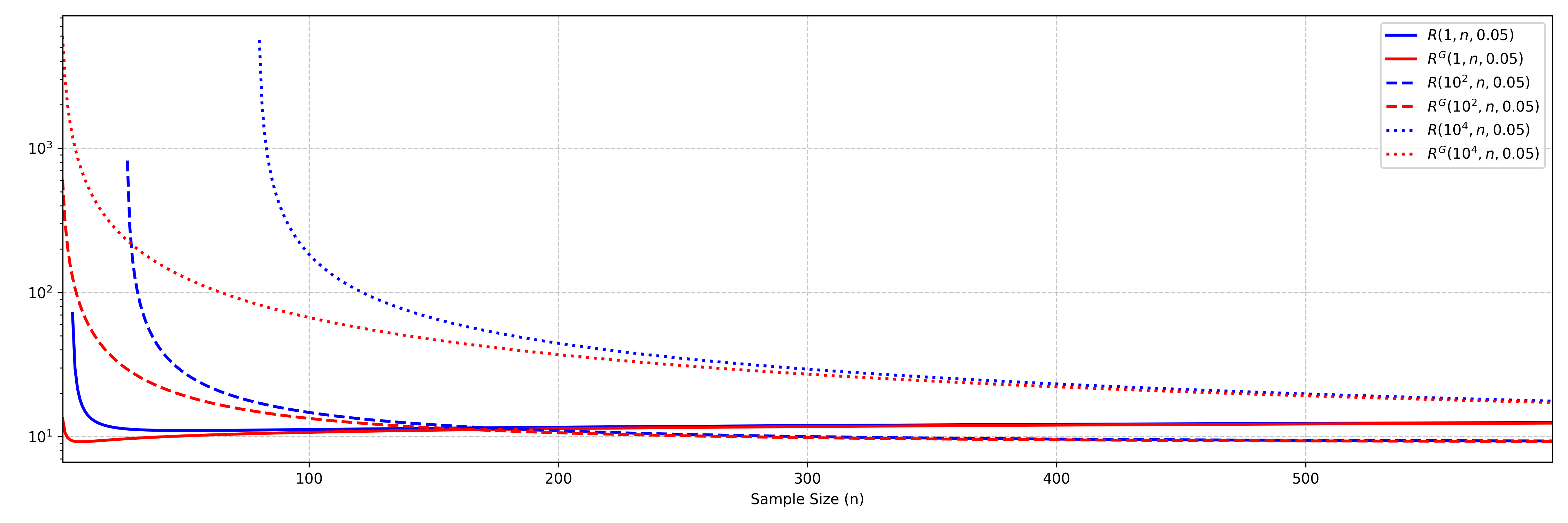}
    \caption{Comparison of the functions $R(g,n,\alpha)$ and $R^G(g,n,\alpha)$ from equations \eqref{eq:radii} for $\alpha = 0.05$, $r=k=1$, and $g \in \{1,100,1000\}$.
    There are three distinct features. 
    Firstly, $R(g,n,\alpha) \geq R^G(g,n,\alpha)$ for all $n$.
    Secondly, $\underset{n\rightarrow \infty}{\lim} \frac{R(g,n,\alpha)}{R^G(g,n,\alpha)} = 1$.
    Thirdly, $R(g,n,\alpha)$ is infinite for small $n$, whereas $R^G(g,n,\alpha)$ is always finite.}
    \label{fig:radii_comparison}
\end{figure}
\end{document}